\ifdefined\RFSplit\def\RFptsize{12pt}\else\ifdefined\RFInternetAppendix\def\RFptsize{12pt}\else\def\RFptsize{11pt}\fi\fi
\documentclass[\RFptsize]{article}

\newif\ifRFInternetAppendixMode
\ifdefined\RFInternetAppendix
  \RFInternetAppendixModetrue
\else
  \RFInternetAppendixModefalse
\fi
\newif\ifRFBlind
\ifdefined\RFBlinded
  \RFBlindtrue
\else
  \RFBlindfalse
\fi
\newif\ifRFCombined
\RFCombinedtrue
\ifdefined\RFInternetAppendix\RFCombinedfalse\fi
\ifdefined\RFSplit\RFCombinedfalse\fi
\newif\ifRFBody
\RFBodytrue
\ifRFInternetAppendixMode\RFBodyfalse\fi
\newif\ifRFAppendix
\RFAppendixfalse
\ifRFInternetAppendixMode\RFAppendixtrue\fi
\ifRFCombined\RFAppendixtrue\fi

\newcommand{\IAname}{\ifRFCombined Appendix\else Internet Appendix\fi}
\newcommand{\IAref}[1]{\ifRFCombined Appendix~\ref{#1}\else Internet Appendix~\ref{IA-#1}\fi}

\usepackage[utf8]{inputenc}
\usepackage[T1]{fontenc}
\usepackage{amsmath,amssymb,amsthm}
\usepackage{mathtools}
\usepackage{geometry}
\usepackage{booktabs}
\usepackage{array}
\usepackage{enumitem}
\usepackage[round]{natbib}
\usepackage{setspace}
\usepackage{adjustbox}
\usepackage{graphicx}
\usepackage[section]{placeins}
\usepackage{tikz}
\usetikzlibrary{decorations.pathreplacing}
\usepackage{xr-hyper}
\usepackage{hyperref}
\usepackage{cleveref}

\ifRFCombined\else\fi
\definecolor{nberred}{rgb}{0.72,0.20,0.05}
\hypersetup{colorlinks=true,linkcolor=nberred,citecolor=nberred,urlcolor=nberred}
\newif\ifRFShowEdits
\RFShowEditsfalse
\ifRFShowEdits\hypersetup{allcolors=black}\fi

\ifRFCombined\else
\ifRFInternetAppendixMode
\else
\fi
\fi

\newtheorem{proposition}{Proposition}

\newtheorem{corollary}{Corollary}
\newtheorem{lemma}{Lemma}

\theoremstyle{definition}

\AddToHook{env/proposition/begin}{\crefalias{section}{proposition}}
\AddToHook{env/theorem/begin}{\crefalias{section}{theorem}}
\AddToHook{env/corollary/begin}{\crefalias{section}{corollary}}
\AddToHook{env/lemma/begin}{\crefalias{section}{lemma}}
\AddToHook{env/remark/begin}{\crefalias{section}{remark}}
\AddToHook{env/definition/begin}{\crefalias{section}{definition}}
\AddToHook{env/assumption/begin}{\crefalias{section}{assumption}}
\AddToHook{env/example/begin}{\crefalias{section}{example}}

\newcommand{\abar}{\bar{\alpha}}

\newcommand{\inparen}[1]{\left( #1 \right)}
\newcommand{\inbrak}[1]{\left[ #1 \right]}
\newcommand{\Rev}{\operatorname{Rev}}
\newcommand{\Piz}{\Pi_0}
\newcommand{\LL}{L}%
\newcommand{\nn}{N}%
\newcommand{\ww}{w}%
\newcommand{\cc}{c}%
\newcommand{\kk}{k}%
\renewcommand{\ss}{s}%
\newcommand{\ee}{\epsilon}%
\newcommand{\EE}{N_\epsilon}%
\newcommand{\ehat}{\hat{\eta}}%
\newcommand{\nntil}{\nn_{\ehat}}%
\newcommand{\nnhat}{\hat{\nn}}%

\ifRFCombined
\else
\renewcommand{\footnotesize}{\normalsize\setstretch{1.5}}
\fi

\newcommand{\RFtitle}{The AI Layoff Trap}
\newcommand{\RFauthors}{Brett Hemenway Falk, Gerry Tsoukalas}
\newcommand{\RFdate}{November 20, 2026 \\ (first version: March 2, 2026)}
\newcommand{\RFthanks}{Hemenway Falk: University of Pennsylvania (\href{mailto:fbrett@cis.upenn.edu}{fbrett@cis.upenn.edu}). Tsoukalas: University of Oxford (\href{mailto:gerry.tsoukalas@sbs.ox.ac.uk}{gerry.tsoukalas@sbs.ox.ac.uk}) and Boston University (\href{mailto:gerryt@bu.edu}{gerryt@bu.edu}). We thank Ricardo Caballero, Robert Engle, Rep.\ Bill Foster, Anindya Ghose, Kartik Hosanagar, Kose John, Leonid Kogan, Giacomo Mantegazza, Serguei Netessine, Fahad Saleh, Larry Summers, Prasanna Tambe, and S.\ Alex Yang for helpful comments and discussions, as well as seminar participants at Carnegie Mellon University and participants at the 2026 INFORMS Annual Meeting and ETCEI 2026 (Emerging Tech Conference: Edge Intelligence). All errors are our own.}
\ifRFInternetAppendixMode
\title{Internet Appendix to ``\RFtitle''}

\author{\RFauthors}
\date{\RFdate}

\else
\title{\RFtitle}

\author{\RFauthors\thanks{\RFthanks}}
\date{\RFdate}

\fi

\begin{document}

\maketitle
\ifRFBody
\ifRFCombined\begin{spacing}{1.0}\fi
\begin{abstract}
\ifRFCombined\else\normalsize\fi
    If AI displaces workers faster than the economy can reabsorb them, it risks eroding the consumer demand firms depend on. We show that knowing this is not enough for firms to stop it. In a competitive task-based economy, each firm captures the full cost saving from automation but bears only a fraction of the demand loss it creates in the product market; the rest falls on rivals. This \emph{demand externality} traps rational firms in an automation arms race, displacing workers beyond what is collectively optimal. The mechanism harms both workers and firm owners and can survive in general equilibrium under adjustment frictions. More competition and higher AI productivity can amplify the excess. Wage adjustment, free entry, transfers, worker equity, and Coasean bargaining need not eliminate it. A Pigouvian automation tax implements the cooperative optimum. Policy should address not only the aftermath of automation but also the competitive incentives driving it.\\[10pt]
{\bf Keywords:} artificial intelligence, automation, labor displacement, Pigouvian tax.

\end{abstract}
\ifRFCombined\end{spacing}\fi

\section{Introduction}\label{sec:intro}

The fear that technology will displace workers is at least as old as the Industrial Revolution \citep{ricardo1821principles,keynes1930economic}. Historically, displacement has largely been self-correcting: automation of existing tasks has been offset by the creation of new tasks and occupations. What \citet{acemoglu2018race,acemoglu2019automation} call the reinstatement effect has tended to stabilize the labor market. Whether this balance will hold in the age of AI is an open question: \citet{autor2024new} find that displacement has intensified over the past four decades while the creation of new work has not always kept pace, and early signs suggest the current wave is disproportionately affecting entry-level workers \citep{brynjolfsson2025canaries}. 

Even if reinstatement eventually occurs, a problem arises along the way: displaced workers are also consumers, and when their lost income is not replaced, each round of layoffs erodes the purchasing power all firms depend on. At the limit, this becomes self-destructive: firms automate their way to ever greater productivity and ever weaker demand. Public discourse increasingly treats this dynamic as an inevitable process with no natural brake \citep{shah2026intelligence}. But rational, forward-looking firms \emph{should} be the brake; if the cliff ahead is visible to all, why would they race toward it?

Yet firms appear to be heading in precisely that direction. In February 2026, Block cut nearly half its 10,000-person workforce, with CEO Jack Dorsey stating that AI had made many of those roles unnecessary and that ``within the next year, I believe the majority of companies will reach the same conclusion'' \citep{cnbc2026block}. In 2025, U.S.\ employers announced more than a million job cuts, and AI was explicitly cited in roughly 55,000 of them, with technology firms such as Amazon, Salesforce, and IBM citing AI for cuts in customer support, human resources, and management layers \citep{cnbc2025layoffs}; by the first nine months of 2026 it had become the most-cited reason, named in about a fifth of announced cuts \citep{challenger2026september}. The exposure extends beyond tech: \citet{eloundou2024gpts} estimate that roughly 80\% of U.S.\ workers hold jobs with tasks exposed to large language models.
None of this is hidden. Against this backdrop, we ask under what conditions rationality and perfect foresight are enough to prevent competitive over-automation, what determines the size of the distortion when they are not enough, and which proposed policy responses correct it.

To answer these questions, we develop a task-based automation model inspired by \citet{acemoglu2018race}, but refocused from the labor market to the product market: when automation displaces workers, their forgone spending reduces every firm's revenue. Each of several symmetric firms chooses what fraction of its workforce to replace with AI. Automated tasks are performed at lower cost, but integration frictions make each successive task harder to automate. On the demand side, workers spend a fraction of their income on the sector's output; firm owners spend less, normalized to zero in the baseline. Some displaced wage income is recovered through reemployment or transfers, but the remainder is lost to the sector. The model is deliberately stripped down to make this channel transparent, and the demand cliff ahead visible to all firms. The baseline holds wages fixed (relaxed in \Cref{sec:endo_wages}) and shuts down capital-income recycling (relaxed in \Cref{sec:recycling}); other baseline assumptions are also relaxed in \Cref{sec:generalizations}. Despite its parsimony, the framework accommodates a range of policy instruments (\Cref{sec:policy}) and robustness checks.

We show that competition creates a demand externality that traps firms. An automating firm captures the full cost saving but, under competitive pricing, bears only a fraction of the resulting aggregate demand destruction; the rest falls on rivals. Each firm's profit-maximizing automation rate is a strictly dominant strategy that exceeds the cooperatively efficient level, so foresight alone cannot prevent the race toward the cliff. The distortion deepens with competition: a monopolist fully internalizes the externality, while fragmented markets exhibit the widest gap. In the frictionless limit, where every task is equally easy to automate, the game sharpens into a Prisoner's Dilemma whenever the cost saving per task falls short of the demand loss but exceeds each firm's own share of it: every firm displaces its entire human workforce with AI, even though collective restraint would raise all profits. The resulting surplus loss is not a transfer from workers to firm owners; it is a deadweight loss that harms both.

Since the loss falls on both sides, a natural question is whether policy can correct it. We evaluate six instruments against the externality margin. Upskilling and worker equity participation narrow the wedge but cannot eliminate it. Nor can Coasean bargaining: because automation is a dominant strategy, no voluntary agreement among firms is self-enforcing. Capital income taxes do not alter the equilibrium automation rate, operating on profit levels rather than the per-task margin where the externality resides. Neither does universal basic income: it raises the floor on living standards but leaves the automation incentive unchanged. Only a Pigouvian automation tax, set equal to the uninternalized demand loss per task, implements the cooperative optimum; its revenue can fund retraining that raises income replacement, shrinking the externality over time and making the tax potentially self-limiting.

The core result is also robust to several generalizations. Higher AI productivity widens the wedge rather than resolving it: each firm perceives a market-share gain from automating beyond rivals, but at the symmetric equilibrium these gains cancel, leaving only the additional distortion. This Red Queen effect means that ``better'' AI, far from mitigating over-automation, amplifies it. Endogenous wage adjustment, a key self-correcting channel in the framework of \citet{acemoglu2018race}, dampens the incentive to automate but cannot close the wedge once automation begins. Free entry, capital-income recycling, and richer product-market structures likewise need not eliminate the distortion.

These generalizations are partial-equilibrium results. In a frictionless general equilibrium (GE), the spending that displaced workers withdraw would be reabsorbed through lower prices or a lower interest rate. \IAref{ge:appendix} shows that the loss to the automating firms' owners can survive when neither adjustment is available. In a two-date economy in which workers cannot borrow against future income, the interest rate cannot fall below a floor, and the sector's output prices and wages are fixed by inherited contracts while the price of an outside resource adjusts, firms maximize their owners' consumption utility and anticipate all equilibrium responses. Every operating owner is worse off at Nash than under coordinated restraint, and a balanced transfer to outside owners makes restraint a Pareto improvement.

Our work contributes to several literatures. We build on the task-based approach to automation \citep{zeira1998workers, autor2003skill, acemoglu2018race, acemoglu2019automation}, which emphasizes offsetting forces that restore labor demand after displacement, notably new task creation and a self-correcting wage channel. \citet{acemoglu2025simple} evaluates the aggregate productivity effects of AI within this framework. These contributions focus on whether and how the labor market rebalances; we ask what happens on the product-market side when rebalancing is slow or incomplete. Our focus on firms' automation choices also connects to finance research showing that AI adoption affects firm growth and product innovation and that access to low-cost labor changes firms' incentives to develop cost-saving production technologies \citep{babina2024artificial,bena2025machines}.

A growing literature argues that automation may be excessive. The closest to our setting is \citet{beraja2025inefficient}, who show that automation is inefficient when displaced workers face borrowing constraints during reallocation. Their mechanism operates through the labor market: firms ignore the welfare cost imposed on credit-constrained workers. Ours runs through the product market instead, with firms disregarding the demand they destroy for rivals even though they foresee it. Two further differences follow. Their inefficiency arises even for a single firm in isolation, whereas ours requires competition and vanishes under monopoly. And while both corrections rest on efficiency rather than distributional preferences, ours would reduce automation even with zero weight on workers, because over-automation harms firm profits themselves.
Other channels for excessive automation share the feature that they would distort a single firm's decision even in isolation: the technology ecosystem may be biased toward ``so-so'' automation that displaces workers without large productivity gains \citep{acemoglu2020wrong}, automation may disproportionately target high-rent tasks, dissipating worker surplus rather than raising output \citep{acemoglu2024rent}, and corrective taxation has been justified by transitional frictions \citep{guerreiro2022should} and distributional concerns \citep{costinot2023robots}. More recently, AI-specific channels of the same in-isolation type have been identified through skill erosion, productivity mismeasurement, and copyright dynamics that can depress long-run welfare absent regulation \citep{caosun2026augmentation,bondi2026skill,yang2024generative}. Our externality, by contrast, arises only under competition and persists even when automation is highly productive and the planner places no weight on workers.

The demand externality we study belongs to the family of aggregate demand spillovers introduced by \citet{rosenstein1943problems} and formalized by \citet{murphy1989industrialization}. In their ``big push'' models, demand complementarities across sectors can prevent individually unprofitable investments from being made even though simultaneous adoption would be collectively profitable. Our mechanism is the mirror image: \emph{individually profitable} automation is collectively destructive because each firm's cost saving erodes the revenue base all firms share. \citet{cooper1988coordinating} provide the canonical framework for coordination failures driven by aggregate demand externalities; our game shares this setting of aggregate demand spillovers but not its strategic-complementarity structure: automation here is a strictly dominant strategy, making the problem a true externality rather than a coordination failure that communication could resolve. \citet{benzell2015robots} show in a dynamic overlapping-generations economy that automation can erode workers' purchasing power and generate immiserating growth through capital-accumulation channels. Our contribution is distinct: even absent these dynamic channels, and even when firms perfectly foresee the demand loss, decentralized competition alone induces excessive automation, because each firm captures the full labor-cost saving while bearing only a fraction of the resulting aggregate demand destruction. \citet{korinek2019artificial} similarly examines AI-driven income redistribution at the aggregate level, and \citet{caballero2026speculative} models AI investment that shifts income toward capital owners, whose rising saving funds a transition to a high-capital steady state, so that temporary overvaluation can leave a permanent real legacy. Neither features the firm-level strategic competition in which one firm's automation choice imposes a demand externality on rivals.

A growing literature in economics and information systems has established that AI systems deliver substantial productivity gains \citep{brynjolfsson2025generative} and are increasingly deployed in strategic roles such as pricing. Because the externality we study runs through competition among rival firms, it also matters how AI changes that competition. Pricing algorithms can learn to collude \citep{banchio2022artificial}. In oligopoly settings, however, such collusion proves fragile, weakening when agents differ in patience or data access and breaking down as competitors multiply, and when AI only advises human sellers, its effect on prices depends on who is in the market and on whether sellers come to trust the advice \citep{keppo2026fragility,keppo2026trusting}. On the adoption side, \citet{li2025forced} show that firms under labor-issue scrutiny invest specifically in AI automation rather than other forms of IT, and \citet{bastani2025human} show that as AI reliability improves, incentivizing effective human oversight becomes prohibitively expensive, weakening a key check on automation. What this literature has not modeled is how automation decisions, rather than prices, interact across firms: each adoption decision is rational in isolation, but collectively they erode the consumer demand all firms depend on. We provide that model, connecting the micro-level evidence this literature has documented to a macro-level market failure that no individual firm can prevent.

The remainder of the paper is organized as follows. \Cref{sec:model} presents the model. \Cref{sec:analysis} derives the equilibrium and the over-automation wedge. \Cref{sec:policy} evaluates policy instruments. \Cref{sec:generalizations} extends the model to AI productivity gains, endogenous entry, endogenous wages, capital-income recycling, and richer product-market interaction. \Cref{sec:discussion} discusses implications and limitations.

\section{Model}\label{sec:model}
The baseline isolates the demand consequences of automation in the simplest environment that supports the mechanism: symmetric firms, a single sector, and exogenous wages. We describe the supply side (cost structure and automation choice), then the demand side (how displacement feeds back into revenue), and finally the game firms play. Each assumption is relaxed in \Cref{sec:generalizations}.

Consider a sector with $\nn \geq 2$ symmetric firms, indexed $i = 1, \dots, \nn$. It will later become useful to think of each firm as having a single \emph{owner}, for example the equity holder, who is entitled to the firm’s operating profits.

In the spirit of the task-based framework of \citet{acemoglu2018race}, each firm is endowed with $\LL > 0$ task-positions. Initially all tasks are performed by human \emph{workers}; a new technology shock arrives, for example agentic AI, and each firm must decide how much of its workforce to replace. In particular, firm $i$ chooses an automation rate $\alpha_i \in [0,1]$: tasks $z \in [0,\alpha_i]$ are performed by AI at cost $\cc$ per task, and tasks $z \in (\alpha_i,1]$ remain with human workers at wage $\ww$ per task, with $0 \le \cc \le \ww$. Since each automated task displaces one worker, $\alpha_i$ is simultaneously the automation rate and the fraction of the workforce laid off; we use the two descriptions interchangeably. Wages are exogenous in the baseline; \Cref{sec:endo_wages} endogenizes wages.

In the perfect-substitutes limit of the CES task aggregator in \citet{acemoglu2018race}, each task produces one unit of output regardless of mode, so firm output is $Y_i = \LL$; \Cref{sec:phi} relaxes this to allow AI to not only reduce costs, but also increase firm output. This normalization shuts down productivity and quality margins so that the baseline captures only the spending consequences of labor displacement.

We follow the literature in assuming tasks are ordered by comparative advantage, making the marginal task progressively harder to integrate; we capture this via a convex integration cost $\frac{\kk}{2}\LL\alpha_i^2$ with $\kk \geq 0$, using the standard quadratic adjustment-cost specification \citep{lucas1967adjustment}.
Firm~$i$'s total production cost is therefore
\begin{equation}\label{eq:cost}
  C_i(\alpha_i) = \LL \big ( \alpha_i \cc + (1-\alpha_i)  \ww \big ) +  \tfrac{\kk}{2}\,\LL\,\alpha_i^2.
\end{equation}
Defining the per-task cost saving from automation as $\ss \coloneqq \ww - \cc$,
the cost equation can be rewritten as $C_i = \LL(\ww - \ss\,\alpha_i) + \tfrac{\kk}{2}\,\LL\,\alpha_i^2$: each automated task saves~$\ss$ in labor costs but incurs the integration friction.

On the demand side, workers have a higher marginal propensity to consume (MPC) than owners \citep{kaldor1956alternative,summers2014secular,summers2015demand}; workers spend a fraction $\lambda \in (0, 1]$ of their income on the sector's good, generating the type of cross-firm demand linkage analyzed by \citet{murphy1989industrialization}. Owners, by contrast, spend none of their income in the sector in the baseline (\Cref{sec:recycling} relaxes this). This MPC asymmetry implies that when automation displaces workers, income shifts toward agents with a lower sectoral MPC, reducing aggregate expenditure on the sector. The asymmetry itself has a structural foundation: under non-homothetic preferences in which mass-market goods saturate at high incomes \citep{matsuyama2002rise,cominlashkari2021structural,boppart2014structural}, marginal owner income flows to a separate luxury segment rather than back to the modeled sector, and firms specialized to mass-market production cannot quickly retool to capture the redirected demand \citep{rameyshapiro2001displaced,cooperhaltiwanger2006adjustment}. \IAref{ge:appendix} derives a demand spillover from household optimization and market clearing. It allows output to respond at inherited prices, whereas the baseline here holds output fixed and lets prices adjust.

When firm $j$ automates a fraction $\alpha_j$ of its tasks, $\alpha_j \LL$ workers are displaced. A fraction $\eta \in [0,1]$ of displaced wage income is replaced via reemployment, transfers, or other sources \citep{jacobson1993earnings}; the remainder, $(1-\eta)\ww$ per displaced worker, is lost to the sector.

Across all $\nn$~firms, the total number of displaced workers is $\sum_j \alpha_j \LL$, so total wage income lost to displacement is $(1-\eta)\ww\sum_j \alpha_j \LL$.
Total labor income in the sector is therefore $\ww\LL\nn - (1-\eta)\ww\sum_j \alpha_j \LL$, of which a fraction~$\lambda$ is spent on the sector's good.
Adding autonomous demand $A > 0$ (from outside the sector or from capital income), aggregate sectoral expenditure is
\begin{equation}\label{eq:demand}
  D(\boldsymbol{\alpha}) = A + \lambda \ww\LL\!\left[\nn - (1-\eta)\textstyle\sum_j \alpha_j\right].
\end{equation}
Writing $\abar \coloneqq \frac{1}{\nn}\sum_j \alpha_j$ for the average automation rate, this becomes $D = A + \lambda \ww\LL\nn[1 - (1-\eta)\abar]$.
Defining the effective demand loss per automated task as
\begin{equation}\label{eq:ell}
  \ell \coloneqq \lambda(1-\eta)\ww,
\end{equation}
this simplifies to $D = A + \lambda \ww\LL\nn - \ell \LL\nn\abar$: demand falls linearly in the average automation rate. Unless stated otherwise, the over-automation results below assume $\ell>0$ and $\ss>0$. If $\ell=0$, private and cooperative incentives coincide; if $\ss=0$ and $\kk>0$, neither chooses automation.

Firms sell their output on the product market at a uniform price that equates aggregate supply and demand.
Since all firms produce the same output $Y_i = \LL$, total supply is $\nn\LL$ and the market-clearing price is $p = D / (\nn\LL)$.
Each firm earns revenue $\Rev_i = p \cdot Y_i = D/\nn$, which, after substituting \eqref{eq:demand}, gives
\begin{equation}\label{eq:revenue}
  \Rev_i = \frac{A}{\nn} + \lambda \ww\LL - \ell \LL\abar.
\end{equation}
Firm~$i$'s profit is $\pi_i = \Rev_i - C_i$.
Substituting \eqref{eq:revenue} and \eqref{eq:cost}:
\begin{equation}\label{eq:profit_clean}
  \pi_i = \Piz + \LL\!\left(\ss\,\alpha_i - \ell\abar - \tfrac{\kk}{2}\alpha_i^2\right),
\end{equation}
where $\Piz \coloneqq A/\nn + (\lambda - 1)\ww\LL$ is the per-firm profit when no firm automates.
Writing $\abar = (\alpha_i + \sum_{j \neq i}\alpha_j)/\nn$ to isolate firm~$i$'s own action:
\begin{equation}\label{eq:profit_expanded}
  \pi_i = \Piz + \LL\!\left[\alpha_i\!\left(\ss - \frac{\ell}{\nn}\right) - \frac{\kk}{2}\alpha_i^2 - \frac{\ell}{\nn}\sum_{j \neq i}\alpha_j\right].
\end{equation}

Firms play a one-shot simultaneous-move game, each choosing $\alpha_i$ to maximize~$\pi_i$; the product market then clears mechanically given the automation profile.%
\footnote{An alternative would be a two-stage game in which firms first choose automation rates and then compete on price or quantity. We abstract from this type of more elaborate second-stage product-market competition because those strategic effects are already well studied and would obscure the novel mechanism we isolate here: the demand externality from automation under full transparency. The qualitative results are plausibly robust to richer product-market interaction, but closed-form solutions would become substantially more complex.}
The solution concept is Nash equilibrium.

Define aggregate owner surplus $\mathcal{K}$ and aggregate worker income $\mathcal{W}$:
\begin{align*}
  \mathcal{K} &\coloneqq \sum_i \pi_i \\
  \mathcal{W} &\coloneqq \ww \LL \nn [1-(1-\eta)\abar].
\end{align*} 
We measure over-automation against two benchmarks: the \emph{cooperative optimum}, which maximizes~$\mathcal{K}$, and a \emph{generalized social planner} who maximizes
\[
S(\mu) \;\coloneqq\; \mu\,\mathcal{W} + (1-\mu)\,\mathcal{K}
\]
for a weight $\mu \in [0,1]$ on workers.

Note that the environment assumes full transparency: every firm can directly observe how automation maps into lost worker income and reduced aggregate spending. The question \Cref{sec:analysis} answers is whether this visibility alone is sufficient for firms to curb automation in a competitive setting.

\section{Equilibrium and Over-Automation}\label{sec:analysis}

This section derives the equilibrium, shows firms over-automate relative to the cooperative optimum, and quantifies the resulting surplus loss.
All proofs are collected in \ifRFCombined Appendix~\ref{app:proofs}\else the Internet Appendix\fi.

\subsection{Equilibrium and the Over-Automation Wedge}\label{subsec:ge}

To characterize the equilibrium, consider firm~$i$'s marginal incentive to automate.
Recall from~\eqref{eq:ell} that $\ell(\ww) = \lambda(1-\eta)\ww$ is the demand lost per displaced worker, proportional to the wage because displaced workers' forgone spending scales with their earnings.
(We write simply~$\ell$ when the wage is held fixed; \Cref{sec:endo_wages} endogenizes~$\ww$.)
From~\eqref{eq:profit_expanded}, firm~$i$'s marginal profit from automation is
\[
  \frac{\partial \pi_i}{\partial \alpha_i}
    = \LL\!\left(\ss - \frac{\ell}{\nn} - \kk\alpha_i\right).
\]
A marginal increase in automation saves~$\ss$ in labor costs but incurs friction~$\kk\alpha_i$ and reduces the firm's revenue by~$\ell/\nn$.
The revenue loss is $\ell/\nn$ rather than~$\ell$ because competitive pricing allocates revenue equally across symmetric firms~\eqref{eq:revenue}: firm~$i$'s automation reduces aggregate demand by~$\ell\LL$, but only~$\ell\LL/\nn$ of this falls on firm~$i$ itself.
Each firm therefore underestimates the social cost of its automation, suggesting systematic over-automation in equilibrium.
The following proposition confirms this and quantifies the gap.

\begin{proposition}[Equilibrium and over-automation]\label{prop:alphastar}
Suppose $\kk > 0$ (the frictionless case $\kk = 0$ is treated separately in \Cref{cor:frictionless}). In the model defined in \Cref{sec:model}, define the automation threshold
\begin{equation}\label{eq:Nstar}
  \nn^{*} \coloneqq \frac{\ell}{\ss} = \frac{\lambda(1-\eta)\ww}{\ww - \cc}.
\end{equation}
If $\nn \leq \nn^{*}$, no firm automates ($\alpha^{NE} = 0$).\\
If $\nn > \nn^{*}$ (equivalently, $\ss > \ell/\nn$):
\begin{enumerate}[label=(\roman*),nosep]
  \item Each firm's strictly dominant strategy is $\alpha^{NE} = \min\inparen{ (\ss - \ell/\nn)/\kk, 1 }$;
  \item\label{prop:CO} The cooperative optimum is
  $\alpha^{CO} = \min \inparen{ \max\inparen{ 0,(\ss - \ell)/\kk }, 1 }$;
  \item\label{prop:wedge_N} If $\ell < \ss < \kk + \ell/\nn$ then both $\alpha^{NE}$ and $\alpha^{CO}$ are interior, 
  and the over-automation wedge is
  \[
    \alpha^{NE} - \alpha^{CO} \;=\; \frac{\ell\,(1 - 1/\nn)}{\kk} \;>\; 0.
  \]
  This is strictly increasing in~$\nn$ and~$\ell$, and decreasing in~$\kk$. If instead $\kk+\ell/\nn \le \ss < \kk+\ell$ (with $\ss>\ell$), then $\alpha^{NE}=1$ while $\alpha^{CO}=(\ss-\ell)/\kk$ is interior, and the wedge is $1-(\ss-\ell)/\kk>0$.
\item\label{prop:wedge_boundary1}
  If $\ss \le \ell$, then $\alpha^{CO} = 0$, and so the wedge is $\alpha^{NE}$.
  Thus if $\ss < \kk + \ell/\nn$, then $\alpha^{NE} = (\ss-\ell/\nn)/\kk$, so the wedge is $(\ss-\ell/\nn)/\kk$.
  On the other hand, if $\kk + \ell/\nn \le \ss$, then $\alpha^{NE} = 1$, so the wedge is 1.
\end{enumerate}
\end{proposition}

The proposition follows from the private first-order condition derived above and its cooperative counterpart: a planner setting a common rate for all firms faces the full demand loss~$\ell$ per automated task rather than the~$\ell/\nn$ each firm perceives, yielding $\alpha^{CO} = (\ss - \ell)/\kk$.
Because rivals' rates enter~\eqref{eq:profit_expanded} only through the term $-(\ell/\nn)\sum_{j\neq i}\alpha_j$, which is independent of~$\alpha_i$, the equilibrium rate is a strictly dominant strategy: each firm over-automates even with perfect foresight about every rival's behavior.

The case structure arises because both $\alpha^{NE}$ and $\alpha^{CO}$ lie in $[0,1]$: each can be at no automation, interior, or full automation depending on how the cost saving~$\ss$ compares to the demand-loss and friction parameters. The wedge can be positive whether the cooperative rate is zero or interior, and it shrinks to zero only once cost savings grow large enough ($\ss \geq \kk + \ell$) to pin both rates at full automation. We caution against reading that zero as a benign region: the demand externality is still present in firms' incentives, but the $\alpha^{NE} - \alpha^{CO}$ comparison can no longer detect it once even the profit-only cooperative benchmark ($\mu = 0$) prescribes displacing the entire workforce. A planner who weights worker income may still choose less automation, but also chooses full automation if $\ss\geq\kk+\ell+\mu\ell/[\lambda(1-\mu)]$ for $\mu<1$ (\Cref{prop:surplus}). Parts~\ref{prop:wedge_N}--\ref{prop:wedge_boundary1} enumerate the relevant combinations; the economic force is the same throughout.

The wedge is nondecreasing in~$\nn$, and strictly increasing while the Nash rate is interior: more fragmented sectors can exhibit wider automation gaps.
This runs counter to the standard intuition that competition disciplines firms to act in consumers' interests; here, more competition dilutes each firm's share of the demand loss, weakening the private incentive to restrain.
A monopolist ($\nn = 1$) fully internalizes the externality ($\alpha^{NE} = \alpha^{CO}$); as $\nn \to \infty$, the wedge approaches $\min\{\ss/\kk,1\}-\alpha^{CO}$, equal to $\ell/\kk$ when both limiting rates are interior.

From \Cref{prop:alphastar}, a firm automates only when $\nn > \nn^{*} = \ell/\ss$: the number of competitors must be large enough that each firm's share of the demand loss, $\ell/\nn$, falls below its cost saving~$\ss$.
As AI costs fall ($\cc \to 0$), $\nn^{*} \to \lambda(1-\eta) \leq 1$: the automation region expands to cover virtually any market with $\nn \geq 2$.
For illustrative parameters ($\cc/\ww = 0.30$, $\lambda = 0.5$, $\eta = 0.30$, $\nn \to \infty$), the wedge equals $\ell/\kk = \alpha^{CO}$: firms in competitive markets automate at twice the cooperatively efficient rate. This interior comparison should be read with a caveat: it presumes frictions are not negligible relative to the cost saving ($\kk > \ss$, so that both rates stay interior as $\nn \to \infty$). When AI is both cheap and frictionless to integrate, both rates saturate at full automation and the wedge collapses to the corner noted above, so the factor-of-two statement is specific to the interior regime. The cost decline that carries the equilibrium into that regime still pushes worker income to the floor of universal displacement before the cooperative rate follows; once both rates reach one, further declines raise profits and leave worker income at that floor, even though the $\alpha^{NE} - \alpha^{CO}$ wedge no longer registers the loss.

\Cref{fig:wedge_panels} illustrates these comparative statics. In each panel, the dashed line marks the $\nn = \nn^{*}$ boundary below which no firm automates, and darker shading indicates a larger wedge. The dominant pattern is that the wedge grows with~$\nn$; non-monotonicity in the other dimensions reflects the regime shift at $\ss = \ell$, where the cooperative optimum moves from zero to an interior solution.

\begin{figure}[t]
\centering
\includegraphics[width=0.33\linewidth]{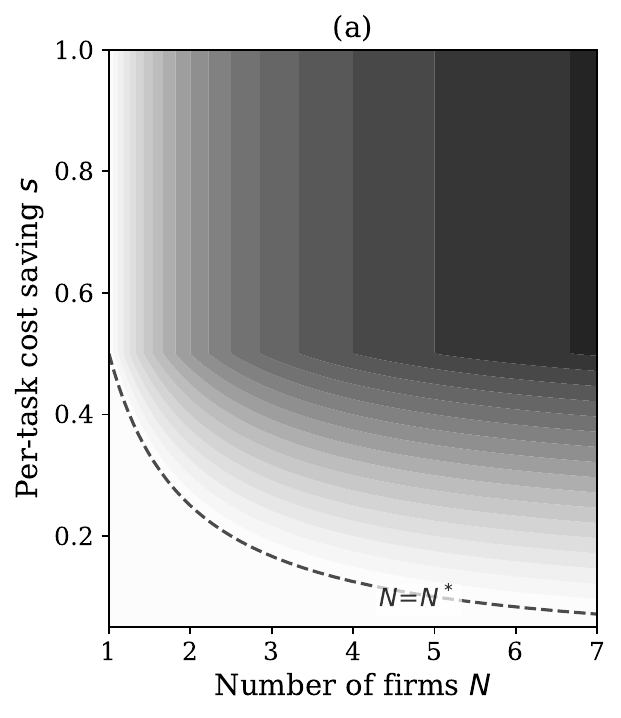}%
\includegraphics[width=0.33\linewidth]{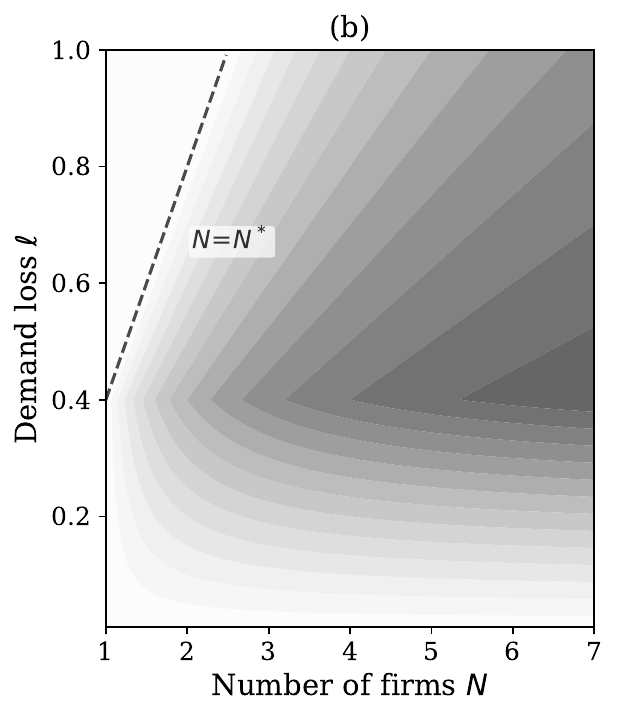}%
\includegraphics[width=0.33\linewidth]{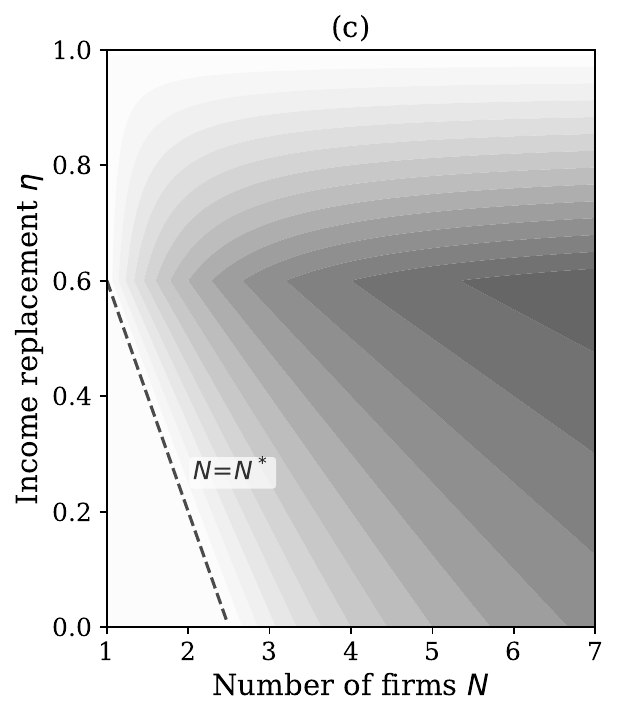}
\caption{The over-automation wedge $\alpha^{NE} - \alpha^{CO}$ across the parameter space. Shading runs from white (zero wedge) to black (wedge $\geq 0.50$). Dashed lines mark the $\nn = \nn^{*}$ boundary below which no firm automates. All panels normalize $\ww=1$.
(a)~Number of firms vs.\ cost saving $\ss = \ww - \cc$ at fixed wage ($\lambda = 0.5$, $\eta = 0$, $\kk = 1$).
(b)~Number of firms vs.\ demand loss $\ell = \lambda(1-\eta)\ww$ ($\cc = 0.60$, $\eta = 0$, $\kk = 1$).
(c)~Number of firms vs.\ income replacement ($\cc = 0.60$, $\lambda = 1$, $\kk = 1$).}\label{fig:wedge_panels}
\begin{minipage}{\linewidth}
\textit{Alt text:} Three grayscale heat maps show the difference between Nash and cooperative automation rates. Across the panels, the wedge generally increases with the number of firms and varies with cost savings, demand loss, or income replacement. Dashed curves mark the threshold below which equilibrium automation is zero.
\end{minipage}
\end{figure}

When frictions are positive, integration costs moderate the equilibrium automation rate. The next subsection shows that when frictions vanish ($\kk \to 0$), this moderating force disappears and the game reduces to a Prisoner's Dilemma: full automation versus none.

\subsection{Frictionless Automation as a Prisoner's Dilemma}\label{sec:PD}

When $\kk = 0$, marginal profit becomes the constant $\LL(\ss - \ell/\nn)$, independent of the automation level, and the outcome is all-or-nothing.
If $\nn < \nn^{*}$, no firm automates; at $\nn=\nn^{*}$, each firm is indifferent over all automation rates.
If $\nn > \nn^{*}$, automating is strictly dominant, and collectively harmful when the cost saving falls short of the demand loss ($\ss < \ell$):

\begin{corollary}[Frictionless limit]\label{cor:frictionless}
Suppose adjustment frictions vanish ($\kk = 0$) and the number of firms exceeds the critical threshold ($\nn > \nn^{*}$).
\begin{enumerate}[label=(\roman*),nosep]
  \item Full automation ($\alpha_i = 1$) is strictly dominant for every firm.
  \item\label{prop:PD} If additionally the cost saving is less than the demand loss per task ($\ss < \ell$), the cooperative optimum is no automation ($\alpha_i = 0$ for all~$i$, yielding per-firm profit~$\Piz$); the equilibrium yields $\Piz + \LL(\ss - \ell) < \Piz$.
  Aggregate owner-surplus loss is $\nn\LL(\ell - \ss)$.
\end{enumerate}
\end{corollary}

Under the condition in part~\ref{prop:PD} ($\ss < \ell$), the Prisoner's Dilemma structure makes the failure of voluntary restraint transparent.
A firm that holds back unilaterally (choosing $\alpha_i = 0$) still suffers the revenue decline from rivals' automation but forgoes the offsetting cost savings; a firm that deviates (choosing $\alpha_i = 1$) captures the savings while imposing only a~$1/\nn$ share of the demand loss on itself.
The resulting payoff matrix has the classic form: mutual restraint yields~$\Piz$ per firm, while mutual automation yields $\Piz + \LL(\ss - \ell) < \Piz$, yet defecting is individually rational regardless of others' choices.
Because automating is strictly dominant (not merely a best response to others' automating), no non-binding agreement can restore efficiency.
Communication is cheap talk: even if all firms acknowledge that collective restraint would raise profits, each firm's individually optimal action remains unchanged.
This distinguishes the automation externality from pure coordination failures (where firms simply need to agree on which equilibrium to play) and motivates the analysis of Coasean bargaining in \Cref{sec:coase}.

\subsection{Income and Profit Losses from Over-Automation}\label{subsec:surplus}

Is the over-automation wedge merely a redistribution from workers to firm owners, or does it reduce total surplus?
Recall the generalized planner introduced in the model section, who maximizes
\begin{equation}\label{eq:S}
  S(\mu) \;=\; \mu\, \mathcal{W} + (1-\mu)\, \mathcal{K}
\end{equation}
for a weight $\mu \in [0,1]$ on workers.

\begin{proposition}[Generalized planner and surplus loss]\label{prop:surplus}
Suppose $\kk > 0$ and $\nn > \nn^{*}$.
\begin{enumerate}[label=(\roman*),nosep]
  \item For $\mu\in[0,1)$, the $\mu$-planner's optimal automation rate is
    \[
      \alpha^{SP}(\mu) \coloneqq \frac{\ss-\ell}{\kk} - \frac{\mu\,\ell}{\lambda(1-\mu)\,\kk},
    \]
    where, as usual, $\alpha^{SP}(\mu)$ is the automation rate, and is thus restricted to the interval $[0,1]$.
    At $\mu = 0$ this reduces to $\alpha^{CO}$ from \Cref{prop:alphastar}; at $\mu=1$, the optimum is zero because worker income strictly decreases with automation.
  \item 
    The surplus loss from the Nash equilibrium relative to the planner's optimum is
    \[
      S\!\left(\mu;\alpha^{SP}\right) - S\!\left(\mu;\alpha^{NE}\right)
        = \frac{(1-\mu)\,\nn\LL\kk}{2}\,
          \bigl[\alpha^{NE} - \alpha^{SP}(\mu)\bigr]^{2},
    \]
    valid when the planner's optimum $\alpha^{SP}(\mu)$ is interior.
  \item \textup{(Pareto dominance.)}
    $\alpha^{NE} > \alpha^{SP}(\mu)$ for every $\mu \in [0,1)$, except when cost savings are large enough to pin both rates at full automation.
    Whenever the cooperative rate is below full automation ($\alpha^{CO} < 1$), the Nash equilibrium is Pareto dominated by the cooperative optimum: workers and firm owners are both strictly worse off.
\end{enumerate}
\end{proposition}

Over-automation is not a transfer from workers to owners: it is, with slight abuse of terminology, a deadweight loss that harms both sides (part~(iii)), the two groups whose income our surplus measure tracks.
Workers lose wage income directly through displacement.
Firm owners, despite cutting costs on each automated task, also lose: collective displacement erodes demand to the point where every firm's equilibrium profit falls below its cooperative-optimum profit.
No redistribution between the two groups can make the Nash outcome efficient. \IAref{ge:appendix} shows that this loss can survive even when every household is counted: in a general-equilibrium example, owners lose in welfare terms, and coordinated restraint with balanced transfers is a Pareto improvement.

When both $\alpha^{NE}$ and $\alpha^{SP}$ are interior, the total wedge between equilibrium and the planner's optimum decomposes into two distinct sources:
\begin{equation}\label{eq:wedge_decomp}
  \underbrace{\alpha^{NE} - \alpha^{SP}(\mu)}_{\text{total wedge}}
    = \underbrace{\frac{\ell(1-1/\nn)}{\kk}}_{\text{demand externality}}
    + \underbrace{\frac{\mu\,\ell}{\lambda(1-\mu)\,\kk}}_{\text{distributional}}.
\end{equation}
The first term is the uninternalized demand externality from \Cref{prop:alphastar}\ref{prop:wedge_N}: it is present even when the planner places zero weight on workers ($\mu = 0$) and cares only about aggregate profit.
It grows with~$\nn$, approaching $\ell/\kk$ as $\nn \to \infty$, so fragmented markets suffer disproportionately.
The second term is a distributional premium: the additional automation reduction a planner who values worker income ($\mu > 0$) would impose beyond the profit-maximizing benchmark.
It is independent of~$\nn$ but grows without bound as $\mu \to 1$; at $\bar\mu \coloneqq \lambda\max\{\ss-\ell,0\}/[\ell+\lambda\max\{\ss-\ell,0\}]$ the planner prohibits automation entirely.
The surplus loss in~(ii) is quadratic in this total wedge and scales with~$\nn\LL$, so both fragmentation and market size amplify the welfare cost.

\begin{figure}[ht!]
  \centering
  \includegraphics[width=0.33\linewidth]{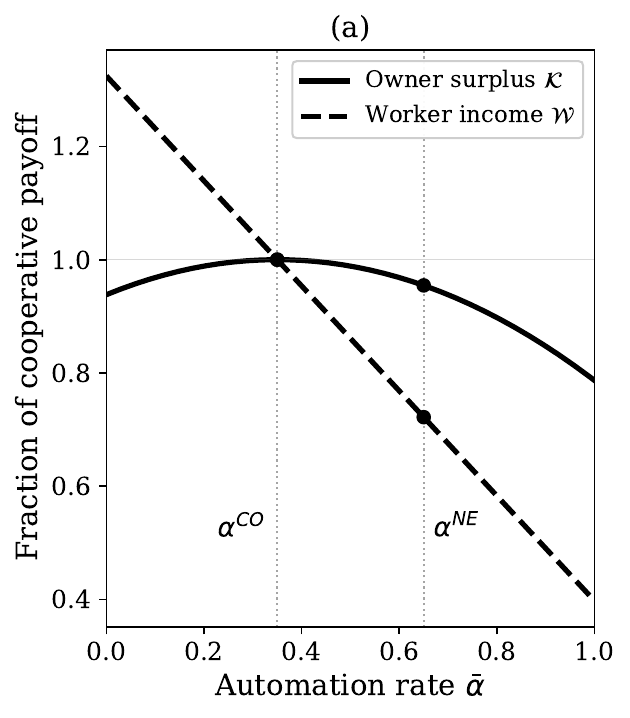}%
  \includegraphics[width=0.33\linewidth]{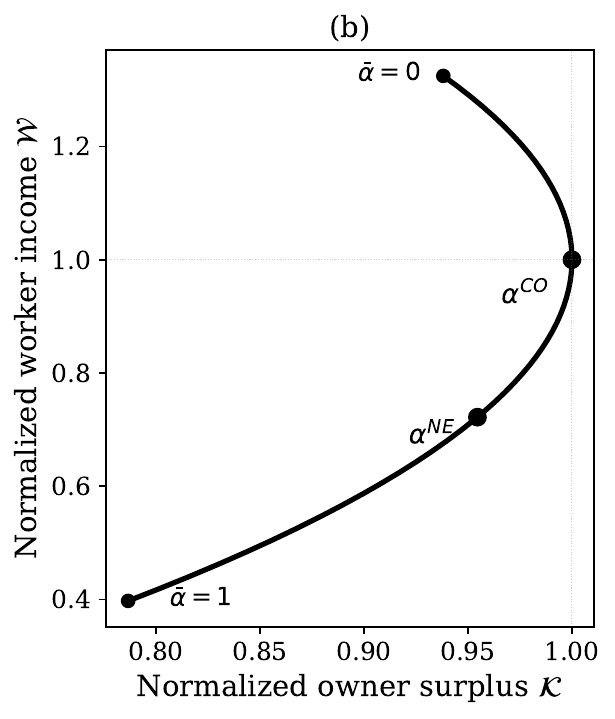}%
  \includegraphics[width=0.33\linewidth]{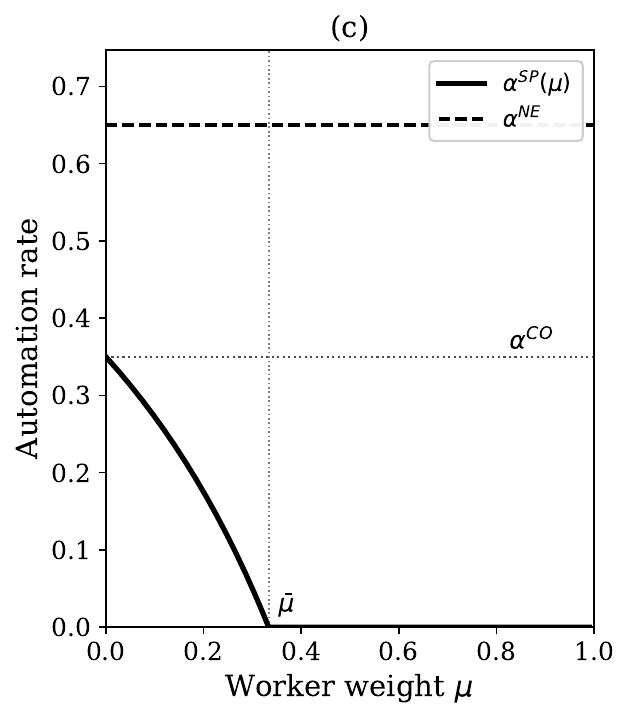}
  \caption{Welfare consequences of over-automation. All panels use $\ww=\LL=1$, $\cc = 0.30$, $\lambda = 0.5$, $\eta = 0.30$, $\nn = 7$, $\kk = 1$, $A = 10$.
  (a)~Normalized owner surplus~$\mathcal{K}$ and worker income~$\mathcal{W}$ vs.\ automation rate ($1$ = cooperative benchmark): both groups lose at $\alpha^{NE}$.
  (b)~Factor payoff locus: $\alpha^{CO}$ at $(1,1)$, $\alpha^{NE}$ strictly southwest (Pareto dominated).
  (c)~Planner's optimum $\alpha^{SP}(\mu)$ vs.\ worker weight~$\mu$ (not normalized): the gap to $\alpha^{NE}$ is the over-automation wedge, positive even at $\mu = 0$.}\label{fig:welfare_panels}
  \begin{minipage}{\linewidth}
  \textit{Alt text:} Three panels compare outcomes under Nash equilibrium and cooperative or social planning. Panel A shows owner surplus and worker income both falling below their cooperative values at the Nash automation rate. Panel B places the Nash outcome southwest of the cooperative outcome in owner-surplus--worker-income space. Panel C shows the planner's automation rate decreasing with the weight on worker welfare while the Nash rate remains fixed.
  \end{minipage}
  \end{figure}

\Cref{fig:welfare_panels} illustrates the Pareto dominance and the decomposition.
To make the losses for the two groups comparable, panels~(a) and~(b) normalize each payoff by its value at~$\alpha^{CO}$, so a value of~$1$ corresponds to the cooperative benchmark.
Panel~(a) plots the normalized payoffs against the common automation rate~$\abar$: both curves peak at or before~$\alpha^{CO}$ and both fall below~$1$ at~$\alpha^{NE}$. The equilibrium rate lies to the right of the aggregate profit peak, so that both owner surplus and worker income are lower than under cooperation. Workers bear the larger loss because their income declines linearly in~$\abar$, while the profit curve is concave and falls more gently.
Panel~(b) re-expresses the same information as a factor payoff locus: each point on the curve corresponds to a different common automation rate, tracing out the $(\mathcal{K}, \mathcal{W})$ pairs as~$\abar$ increases. The cooperative rate sits at~$(1,1)$ and $\alpha^{NE}$ is strictly to the southwest, confirming that moving from equilibrium to the cooperative rate would make both groups better off.
Panel~(c) visualizes the decomposition in~\eqref{eq:wedge_decomp}: the horizontal line marks~$\alpha^{NE}$, and the declining curve is the planner's optimum~$\alpha^{SP}(\mu)$. Even at $\mu = 0$, the gap is positive (the demand-externality term alone), and the required correction grows further as the distributional premium widens with~$\mu$.

Since the over-automation wedge is a structural externality that harms both factor classes, a natural question is whether policy can close it.

\section{Policy Instruments}\label{sec:policy}

Several instruments could in principle address the externality; the question is which ones operate on the right margin.
To answer it, we benchmark against the cooperative optimum~$\alpha^{CO}$, which maximizes aggregate profit without directly weighting worker welfare.
This is deliberately the weakest case for intervention: \Cref{prop:surplus} shows that the demand externality alone reduces both firm profits and worker income, and that placing any positive weight on workers ($\mu > 0$) only widens the wedge.

\Cref{tab:policy} previews the results: only the Pigouvian automation tax fully corrects the distortion; the remaining instruments cushion the losers or partially shrink the wedge, but none eliminates it.\footnote{The comparison concerns the interior automation gap with $\lambda<1$, $\eta<1$, and no enforceable cross-firm automation agreement. UBI and capital-tax neutrality additionally assume their demand contribution is independent of current automation. Full income replacement, the limit of full capital-income recycling, and a binding grand coalition fall outside these conditions. At corners, a rate gap can close even while the marginal demand spillover remains.}

\begin{table}[ht!]
\centering
\caption{Policy instruments and their effects on the over-automation externality.}\label{tab:policy}
\adjustbox{max width=\linewidth}{%

\begin{tabular}{@{}>{\raggedright\arraybackslash}p{0.27\linewidth}>{\centering\arraybackslash}p{0.10\linewidth}>{\centering\arraybackslash}p{0.18\linewidth}>{\centering\arraybackslash}p{0.16\linewidth}>{\centering\arraybackslash}p{0.16\linewidth}@{}}
\toprule
& Section & Changes $\nn^{*}$? & Changes wedge? & Fixes externality? \\
\midrule
Upskilling/retraining ($\eta$) & \ref{sec:eta} & Yes & Yes & Partially \\
Universal Basic Income ($A$) & \ref{sec:ubi} & No & No & No \\
Capital income tax ($t$) & \ref{sec:capital_tax} & No & No & No \\
Worker equity ($\ee$) & \ref{sec:equity} & Yes & Yes & Partially \\
Coasean bargaining ($M < \nn$) & \ref{sec:coase} & Members only & Partially & No \\
Automation tax ($\tau$) & \ref{sec:tax} & Yes & Yes & \textbf{Yes} \\
\bottomrule
\end{tabular}}

\smallskip
\ifRFCombined\footnotesize\fi\textit{Note.} Each row records whether the instrument alters the automation threshold~$\nn^{*}$, the over-automation wedge $\ell(1-1/\nn)/\kk$, and whether it fully corrects the demand externality.
\end{table}
One limitation should be noted: the analysis evaluates each instrument against a single margin, the demand externality identified in \Cref{sec:analysis}, holding all other features of the economy fixed.
In practice, every instrument carries additional costs and benefits outside the model (administrative burden, labor-market distortions, political feasibility) that a full welfare analysis would need to weigh.
Nonetheless, an instrument that does not operate on the externality margin cannot correct the distortion regardless of how it scores on other dimensions; the analysis below separates instruments that can from those that cannot.

\subsection{Displacement vs. Upskilling}\label{sec:eta}

The demand-loss parameter $\ell = \lambda(1-\eta)\ww$ governs the externality's magnitude.
{\emergencystretch=1em In the baseline model, $\eta \in [0,1]$ represents the fraction of displaced wage income recovered through reemployment, transfers, or other sources: higher~$\eta$ shrinks~$\ell$ and thereby, when both rates are interior, the over-automation wedge.\par}

But the parameter extends naturally beyond unity.
When $\eta > 1$, upskilling and reabsorption place displaced workers into higher-paying roles, automation \emph{increases} aggregate labor income, and loss $\ell$ turns negative, which we can interpret as a gain.
This is the scenario invoked by AI optimists, in which technological displacement is a stepping stone to better jobs.
As the following corollary shows, the sign reversal in~$\ell$ flips the externality itself.

\begin{corollary}[Sign of the externality]\label{cor:eta}
In the interior regime where both rates lie strictly in $(0,1)$, the over-automation wedge is
$(\alpha^{NE} - \alpha^{CO}) = \ell\,(1-1/\nn)/\kk$,
which is strictly decreasing in~$\eta$, positive for all $\eta < 1$, zero when $\eta = 1$ (where $\alpha^{NE} = \alpha^{CO} = \ss/\kk$), and negative (under-automation) when $\eta > 1$.
\end{corollary}

The logic is symmetric.
When $\eta < 1$, displacement destroys demand, and each firm bears only $1/\nn$ of the loss, producing over-automation.
When $\eta > 1$, displacement \emph{creates} demand through higher reemployment wages, and each firm captures only $1/\nn$ of the gain, producing under-automation.
In both cases, the distortion grows with~$\nn$: more competition dilutes each firm's share of the externality, whether that externality is negative or positive.
A monopolist ($\nn = 1$) fully internalizes in every case.

The competitive forces are identical; only the sign of the demand externality differs.
As \Cref{sec:tax} will show, the same corrective instrument addresses both cases: a tax when $\eta < 1$, a subsidy when $\eta > 1$.

The case $\eta > 1$ is not merely theoretical.
Historical technological transitions have often eventually reabsorbed displaced workers at higher wages \citep{acemoglu2019automation}, and the current AI buildout offers a concrete channel: the expansion of data centers, energy infrastructure, and AI-adjacent services is creating skilled roles that can pay more than the positions automation displaces.
If this reabsorption is fast enough to push $\eta$ above unity, competitive firms will automate \emph{too slowly}.
However, past displacement episodes point to $\eta < 1$: displaced workers suffer large, persistent earnings losses \citep{jacobson1993earnings}.

The policy implication is that raising~$\eta$ through retraining programs, wage insurance, and incentives for new firm creation is not merely a palliative for displaced workers but a direct lever on the externality: every unit increase in~$\eta$ toward unity shrinks~$\ell$, narrows the over-automation wedge in the interior regime, lowers the corrective tax everywhere, and reduces the burden placed on the corrective instruments analyzed below.
Only at $\eta = 1$ does upskilling close the wedge outright (\Cref{cor:eta}), but that is exactly where the demand loss $\ell$ vanishes, so it marks the edge of the externality rather than a remedy within it.
Pushing $\eta$ past unity would flip the distortion into under-automation, but this is a far less pressing concern: in that regime, displaced workers are already thriving in higher-paying roles.

\subsection{Universal Basic Income}\label{sec:ubi}

Among the most discussed responses to automation-driven displacement is a universal basic income.
In the model, a UBI funded from general revenue maps to an increase in autonomous demand~$A$: because the transfer is unconditional, employed and displaced workers receive the same payment, adding a constant to aggregate spending without altering the marginal income loss from displacement.
This distinguishes UBI from displacement-targeted transfers (wage insurance, severance), which raise the income-replacement rate~$\eta$ and directly shrink~$\ell$; see \Cref{sec:eta}.
The results below concern this modeled object and should not be read as a verdict on all UBI designs.

Because UBI adds a constant to demand, it enters firm profit only through $\Piz = A/\nn + (\lambda - 1)\ww\LL$, the baseline profit when no firm automates.
This term drops out of the first-order condition $\ss - \ell/\nn - \kk\alpha_i = 0$: a higher~$A$ raises the profit floor but changes neither the cost saving~$\ss$ nor the demand loss~$\ell$ that determine the automation rate.
Consequently, UBI alters neither the automation threshold $\nn^{*} = \ell/\ss$ nor the over-automation wedge $\ell(1-1/\nn)/\kk$.
In the language of game theory, UBI changes payoff \emph{levels} but not the payoff \emph{differences} that drive strategic behavior.
More generally, instruments that operate on profit levels can redistribute income but cannot correct the externality; only instruments that change the per-task automation margin can.

Despite not correcting the externality, UBI serves a complementary role.
A higher~$\Piz$ cushions profit losses from over-automation, while the transfer itself raises the floor on workers' living standards, buying time for corrective instruments that operate on the right margin.

Within the model, UBI is a \emph{complement} to the automation tax, not a substitute: a society that relies solely on UBI will over-automate at the same rate, with a higher floor on living standards but the same externality.

This limitation is one of scope, not of magnitude: because UBI is unconditional and its budget does not depend on current automation, no level of payment alters the marginal displacement decision. What UBI does support is the level of demand and of living standards, a role that becomes constructive when the transfer is financed by the automation tax rather than from general revenue, which displacement itself erodes (\Cref{sec:tax}). This complementarity is specific to the regime in which a wedge exists: once automation is collectively optimal, as when AI comes to dominate human labor, the corrective role disappears and only the recycling role remains (\Cref{sec:postlabor}).

\subsection{Capital Income Taxation}\label{sec:capital_tax}

If unconditional transfers do not alter the automation incentive, a natural alternative is to tax the \emph{proceeds} of automation directly.
Consider a proportional tax~$t \in (0,1)$ on capital income (profits), with revenue rebated economy-wide (e.g., as general transfers, public goods, or debt service) so that its contribution to this sector's demand does not depend on the firms' automation choices and is absorbed into the autonomous component~$A$.\footnote{The polar closed-loop case, in which revenue flows immediately to workers who spend a $\lambda$-fraction in the same sector, is isomorphic to mandated worker equity at profit-sharing rate $\ee = t$ and is analyzed in \Cref{sec:equity}: it narrows the wedge but cannot close it whenever $\lambda < 1$.}
Firm~$i$ now maximizes $(1-t)\pi_i$, but because $(1-t)$ is a positive scalar it cancels from the first-order condition: the equilibrium automation rate, the threshold~$\nn^{*}$, and the over-automation wedge are all unchanged.
On the revenue side, $A$ enters per-firm profit only through the constant $\Piz = A/\nn + (\lambda - 1)\ww\LL$, which does not appear in the first-order condition.
If revenue instead funds displacement insurance that raises~$\eta$, the externality shrinks through~$\ell$, but the operative channel is~$\eta$, not the profit tax itself.

The distinction matters because capital income taxes are often conflated with robot taxes in the policy debate.
The robot taxes studied in the literature \citep[e.g.,][]{guerreiro2022should} are taxes on robot purchases, which operate on the per-task margin; a proportional capital income tax is a fundamentally different instrument that scales the entire profit function by $(1-t)$ and cancels from the optimality condition.
The failure is structurally identical to that of UBI (\Cref{sec:ubi}): both instruments shift profit levels rather than operating on the margin where the externality resides.
It follows that a profit tax paired with an unconditional transfer, two level instruments, cannot eliminate the externality either; the corrective leg must be the automation tax, with which such a transfer instead becomes a complement (\Cref{sec:tax}).

\subsection{Worker Equity Participation}\label{sec:equity}

A market-based alternative to taxation, rooted in the profit-sharing literature \citep{weitzman1985sharing}, gives workers a direct stake in the profits that automation generates.
Unlike UBI, which enters demand only through the level term~$A$, profit-sharing flows through the profit function and therefore interacts with automation decisions.
Suppose each firm distributes a fraction $\ee \in [0,1]$ of its profits to workers (through ESOPs, equity grants, or co-determination mandates).
Workers spend a $\lambda$-fraction of this income in the sector, so profit-sharing recycles capital income back into demand.

Aggregate demand now satisfies a fixed-point condition: $D = A + \lambda[\text{wage income} + \ee\sum_i \pi_i]$, where $\sum_i \pi_i = D - \sum_i C_i$.
Because profits depend on~$D$, demand is determined simultaneously with the automation decision; the proof solves this fixed point explicitly.

\begin{proposition}[Worker equity reduces but cannot eliminate the wedge]\label{prop:equity}
Let $\ee \in [0,1)$ and $\kk > 0$, and suppose the equilibrium is interior. Statements at full sharing refer to limits as $\ee\to1$.
Define $\EE \coloneqq \nn - \lambda\ee(\nn-1)$.
\begin{enumerate}[label=(\roman*),nosep]
  \item The cooperative optimum is unchanged: $\alpha^{CO}(\ee) = (\ss - \ell)/\kk$, where, as usual $\alpha^{CO}$ is restricted to $[0,1]$.
  \item The Nash equilibrium automation rate is $\alpha^{NE}(\ee) = (\ss - \ell/\EE)/\kk$ (restricted to $[0,1]$).
  \item When both $\alpha^{NE}$ and $\alpha^{CO}$ are interior solutions, the over-automation wedge
    \[
      \alpha^{NE}(\ee) - \alpha^{CO} = \frac{\ell(\nn-1)(1-\lambda\ee)}{\kk\,\EE}
    \]
    is strictly decreasing in~$\ee$ but strictly positive for all $\ee < 1/\lambda$.
    The wedge vanishes only at $\ee = 1/\lambda$, which requires $\ee > 1$ whenever $\lambda < 1$.
\end{enumerate}
\end{proposition}

Part~(i) is not immediate: profit-sharing changes the demand function by recycling capital income into worker spending, so one might expect it to shift the planner's optimum.
The result follows because the planner already controls all $\nn$ firms and thus fully internalizes the demand externality.
In the planner's first-order condition, the profit-sharing terms cancel: the modified demand-loss parameter $\ell_\ee = \ell - \lambda\ee\ss$ and the demand multiplier $1/(1-\lambda\ee)$ exactly offset, leaving $\kk\alpha = \ss - \ell$ regardless of~$\ee$.

While the cooperative optimum is unaffected, the Nash equilibrium does shift. Intuitively, when workers hold equity, lost sales also cut profits and hence workers' equity income, which depresses spending further; each firm therefore perceives a larger effective demand loss from its own automation than in the baseline, and restrains accordingly. The magnitude of this shift is governed by the compound parameter $\EE = \nn - \lambda\ee(\nn-1)$.
This parameter measures the effective demand-leakage divisor: at $\ee = 0$ it equals~$\nn$, recovering the baseline in which each firm perceives demand loss $\ell/\nn$ per automated task.
As $\ee$ rises, $\EE$ falls toward~$1$, pushing $\alpha^{NE}$ toward the cooperative optimum.

Despite this improvement, the recycling cannot fully close the wedge whenever $\lambda < 1$.
Closing the wedge requires the product $\lambda\ee$ to reach one, i.e., $\ee = 1/\lambda$.
When $\lambda < 1$ this exceeds the feasible range $\ee \in [0,1]$: each unit of profit recycled to workers generates only $\lambda$ units of sectoral demand, so compensating for the leakage would require sharing more than the firm's entire profit.
As $\ee\to1$ (full profit-sharing), the wedge approaches $\ell(\nn-1)(1-\lambda)/[\kk(\nn - \lambda(\nn-1))]$, which remains strictly positive.
(The exception is the knife-edge $\lambda = 1$, $\ee = 1$: there the demand fixed-point expression derived in the \IAname{} has a vanishing divisor $1 - \lambda\ee$, so the wedge closes only in the limit $\lambda\ee \to 1$, and only because $\lambda = 1$ removes the spending leakage that drives the externality. For all $\lambda < 1$ it stays strictly positive.)

The structural limitation is that the externality is fundamentally multilateral: each firm's automation depresses demand for all~$\nn$ firms, and bilateral arrangements between a firm and its own workers cannot reach the demand that leaks to rivals.

A separate question is whether profit-sharing would arise voluntarily.

\begin{corollary}[Voluntary profit-sharing does not arise]\label{cor:equity_voluntary}
Suppose each firm has positive operating profit. If each firm independently chooses its own profit-sharing rate $\ee_i \in [0,1]$ to maximize retained profit $(1-\ee_i)\pi_i$, then $\ee_i = 0$ is a dominant strategy, strictly so when $\lambda<1$.
\end{corollary}

The marginal cost of sharing is~$\pi_i$ (a dollar-for-dollar reduction in retained earnings), while, absent sharing by rivals, the marginal demand benefit is only $\lambda\pi_i/\nn$: workers spend fraction~$\lambda$ of the shared profit in the sector, and firm~$i$ captures $1/\nn$ of the resulting demand increase.
Since $\lambda/\nn < 1$ for any $\nn \geq 2$, the cost strictly exceeds the benefit.
This is a second externality layered on top of the automation externality itself, with the same dominant-strategy structure as \Cref{sec:PD}.

Profit-sharing must therefore be mandated to have any effect, and even then it cannot substitute for a corrective tax: it narrows the wedge but cannot eliminate it, and unlike a corrective tax (\Cref{sec:tax}), does not generate government revenue for retraining programs that would raise~$\eta$.

\subsection{Coasean Bargaining}\label{sec:coase}

None of the instruments considered so far fully closes the over-automation wedge, and worker equity will not arise voluntarily (\Cref{cor:equity_voluntary}). A natural question is whether private ordering could succeed where these instruments have not. By the Coase Theorem \citep{coase1960problem}, if property rights over the externality were well-defined and transaction costs sufficiently low, bargaining could achieve the cooperative optimum without government intervention.
To evaluate this possibility, it helps to separate two questions: can bargaining between a firm and its own workers correct the externality, and can bargaining among firms do so?
As we show below, neither does so outside limiting cases.

\paragraph{Bargaining between a firm and its own workers.}
If displaced workers can bargain for compensation, say a per-task severance payment~$\sigma$, the firm's effective cost saving falls from $\ss$ to $\ss - \sigma$, and the equilibrium automation rate drops.
But the severance also recycles income into demand: displaced workers spend a fraction~$\lambda$ of their compensation in the sector, reducing the effective demand-loss parameter from~$\ell$ to $\ell - \lambda\sigma$.
Severance thus works through the same two channels as the instruments of \Cref{sec:eta,sec:equity}, a smaller cost saving and a smaller spillover, and it compensates the displaced workers for their own loss.

The externality, however, is a different loss.
When firm~$i$ automates, the demand loss $\ell(1-1/\nn)\LL$ reduces revenue at rival firms, lowering their owners' profits (\Cref{eq:revenue}).
Those owners have no seat at firm~$i$'s negotiation with its workers, and firm~$i$'s workers have no claim on their behalf.
Workers at rival firms who retain their positions continue to earn~$\ww$ per task; they are not directly harmed either.
The externality is therefore a firm-to-firm channel running through the product market. Bilateral bargaining reaches it only indirectly, through the recycled severance $\lambda\sigma$, and removes it only at full income replacement, $\sigma=(1-\eta)\ww$, the limiting case of \Cref{sec:eta}.

\paragraph{Firm-to-firm bargaining.}
Since the externality runs across firms, consider a coalition of $M \leq \nn$ firms that jointly choose automation rates to maximize their combined profit, while the remaining $\nn - M$ firms play Nash.

\begin{proposition}[Partial coalitions cannot eliminate the wedge]\label{cor:coase}
Let $\kk>0$ and suppose both the coalition rate and the cooperative optimum are interior.
A coalition of $M$ firms that jointly maximizes its members' combined profit chooses the common automation rate
\[
  \alpha^{M} = \frac{\ss - M\ell/\nn}{\kk}.
\]
The residual over-automation wedge relative to the cooperative optimum is
\[
  \alpha^{M} - \alpha^{CO} = \frac{\ell(1 - M/\nn)}{\kk} > 0 \quad \text{for all } M < \nn.
\]
The wedge vanishes only when $M = \nn$: only the grand coalition replicates the cooperative optimum.
\end{proposition}

The intuition is that a coalition of $M$ firms internalizes $M/\nn$ of the aggregate demand loss; the rest accrues to non-members who enjoy the coalition's restraint without bearing its cost. This is the classic free-rider problem of collective restraint, whose sign-flipped twin is the under-provision of trade credit among rival suppliers, which similarly worsens as the market fragments \citep{chod2019trade}.
Four features of the automation externality prevent the grand coalition from forming.
First, voluntary agreements are not self-enforcing: in the frictionless limit (\Cref{cor:frictionless}), automation is strictly dominant, so a coalition member gains from deviating regardless of whether others honor the agreement.
This is not a coordination failure that communication can resolve; the Prisoner's Dilemma structure means no non-binding arrangement is stable.
With convex costs ($\kk > 0$), the deviation incentive is continuous but still positive. The pattern is familiar from entry-level labor markets that unravel: timing agreements among competing employers collapsed repeatedly because early offers were individually profitable and agreements hard to enforce \citep{roth1994jumping}.
Second, the externality is multilateral and diffuse.
The Coase Theorem's canonical applications involve bilateral or small-number settings; here, each of $\nn$ firms imposes demand losses on all $\nn - 1$ others.
The harm is spread thin: any one rival loses only $\ell\LL/\nn$ per task that firm~$i$ automates, too little to justify a negotiation between the two.
Only the sum over all $\nn-1$ rivals, $\ell\LL(1-1/\nn)$, is worth recovering, and recovering it requires pooling the rivals into a single negotiation, which is the grand coalition that the other three features rule out.
This is precisely the large-numbers setting in which \citet{coase1960problem} acknowledged that private bargaining breaks down.
Third, the automation rate $\alpha_i$ is not contractible among firms: it is an internal organizational choice that rival firms cannot observe or verify, making binding private agreements impractical.
Fourth, automation decisions involve large sunk costs and are substantially irreversible, so a deviator locks in its cost saving before rivals can respond; a firm that delays while rivals proceed loses market share (see \Cref{sec:phi} below); and large $\nn$ makes defection harder to detect and punishment harder to sustain.

To summarize, the demand externality studied here is not a market failure that private ordering can cure. Worker-side bargaining operates on within-firm channels ($\eta$, $\ee$) that reach the cross-firm margin where the externality resides only through recycled income; firm-to-firm bargaining targets the right margin but cannot sustain the grand coalition needed to close the wedge.\footnote{An alternative to bargaining is common control: a merger of $M$ firms implements the coalition of \Cref{cor:coase}, internalizing $M/\nn$ of the externality. It needs no cross-firm contract, but only a merger of all $\nn$ firms closes the wedge, and a merger to monopoly is barred by antitrust law.} The fundamental obstacle is enforcement: automation is an internal choice that rivals cannot observe or contract on, so only self-enforcing agreements are available, and the game's dominant-strategy structure leaves none. The one contractual exception, the grand coalition $M = \nn$ (\Cref{cor:coase}), requires exactly the cross-firm contractibility the maintained regime rules out.
Correcting the externality therefore requires an instrument that does not rely on voluntary agreement but instead changes each firm's marginal automation incentive directly.

\subsection{Pigouvian Automation Tax}\label{sec:tax}

The classic remedy for a negative externality is a \emph{Pigouvian tax}: a per-unit charge set equal to the marginal external cost, so that every agent's private incentive aligns with the social cost \citep{pigou1920}.
In contrast to many textbook externalities, where the harmed parties are outside the firms' market (e.g., pollution), here the harm runs through the firms' own market: automation cuts the income of workers whose spending is the firms' demand, and rival firms bear the lost revenue.
This means the tax rate, its revenue, and its incidence all interact through the same labor-market channel, creating richer policy design questions than the standard case.

\begin{proposition}[Pigouvian automation tax]\label{prop:tax}
Let $\tau \geq 0$ be a per-task automation tax, $\kk > 0$, and suppose the untaxed rates $\alpha^{NE}$ and $\alpha^{CO}$ are both interior.
\begin{enumerate}[label=(\roman*),nosep]
  \item The Nash equilibrium automation rate is $\alpha^{NE}(\tau) = (\ss - \tau - \ell/\nn)/\kk$.
    The rate
    \[
      \tau^{*} = \ell\!\left(1 - \frac{1}{\nn}\right)
    \]
    implements $\alpha^{NE}=\alpha^{CO} = (\ss - \ell)/\kk$.
  \item Under the tax at rate~$\tau^{*}$ without rebate, each firm earns $\pi^{\mathrm{tax}} = \pi^{CO} - \tau^{*}\LL\alpha^{CO}$.
    With an exogenous lump-sum rebate (each firm's receipt is independent of its own automation), each firm achieves exactly~$\pi^{CO}$.
\end{enumerate}
\end{proposition}

The optimal rate has a transparent economic interpretation: each firm already bears $\ell/\nn$ of the demand loss from its own automation; the tax charges it for the remaining $\ell(1-1/\nn)$ imposed on rivals.
For large~$\nn$, $\tau^{*} \approx \ell = \lambda(1-\eta)\ww$, so setting the rate requires only sector-level observables.
Levying the tax, however, requires observing firm-level automation rates. The robot-tax literature assumes that robot purchases are observable \citep{guerreiro2022should}; in practice this is a challenge, though one that may ease as AI adoption generates procurement records.
Unlike rival firms in a Coasean bargaining setting (\Cref{sec:coase}), a tax authority can compel disclosure through mandatory reporting, payroll records, and procurement audits, making approximate measurement feasible even when private verification is not.
Because the welfare loss is quadratic in the wedge (\Cref{prop:surplus}), even an imprecisely targeted tax yields a first-order gain.

\paragraph{Allocation of tax revenue.}
\Cref{prop:tax} pins down the rate; the remaining design question is what to do with the revenue.
Because the externality flows through the labor market, this choice can affect structural parameters governing the distortion, not merely the distribution of gains.

A lump-sum rebate to firms restores cooperative profits exactly (part~(ii)), but returns revenue to the firms that automate while leaving displaced workers, the harmed parties, uncompensated.
The more natural option is to direct revenue toward those workers.
Two channels are available, with different incentive properties.

Direct transfers (wage insurance, severance supplements) raise~$\eta$ mechanically by replacing lost income.
The firm takes~$\eta$ as given: it pays~$\tau$ per automated task regardless of where the revenue goes, so its first-order condition keeps the form in \Cref{prop:tax}. A higher~$\eta$ still lowers $\ell=\lambda(1-\eta)\ww$, shrinking both the wedge and the corrective rate~$\tau^{*}$.
But the standard moral-hazard concern applies: generous income replacement may weaken workers' incentive to retrain or reallocate, sustaining~$\eta$ through transfers rather than through productive reabsorption.

Funding retraining programs also raises~$\eta$, but through human-capital investment rather than income replacement.
This channel is slower and harder to implement, yet it builds the capacity for workers to re-enter the labor market at comparable or higher wages, making gains in~$\eta$ self-sustaining.
In principle, the resulting dynamic is self-reinforcing, complementing the analysis in \Cref{sec:eta}: the tax funds programs that raise~$\eta$, which lowers~$\ell$, which reduces~$\tau^{*}$ in future periods.
To the extent that reabsorption is successful, the required correction shrinks over time and the tax is transitional rather than permanent, echoing the finding in \citet{guerreiro2022should} that the optimal robot tax declines to zero as displaced cohorts retire and new workers choose occupations with full knowledge of automation.

A further use of revenue acts on the demand level rather than on the externality.
Funding an unconditional transfer from a budget fixed in advance does not raise~$\eta$ or enter any firm's first-order condition; it returns revenue to autonomous demand~$A$.\footnote{If the transfer instead pays out current tax revenue, each firm recovers $\lambda/\nn$ of its own tax payment through demand, so the rate that implements the cooperative optimum rises to $\tau=\ell(1-1/\nn)/(1-\lambda/\nn)$.}
The contrast with \Cref{sec:capital_tax} is this: a profit tax paired with such a transfer stacks two level instruments and leaves the wedge intact, whereas the automation tax paired with it corrects the margin and recycles the revenue into demand, the working form of the complement anticipated in \Cref{sec:ubi}.
In a post-labor economy, where the corrective term vanishes, this recycling becomes the entire policy (\Cref{sec:postlabor}).

In practice, a mix of transfers and retraining is likely optimal: short-run transfers to cushion displacement while longer-run retraining builds durable gains in~$\eta$.
The Pigouvian tax therefore has the potential to do double duty: it corrects the externality at the margin, and its revenue can be recycled to shrink the distortion over time.

\subsection{The Post-Labor Limit}\label{sec:postlabor}

What does the model say when AI replaces most human labor? The problem does not vanish; it changes. During the transition, firms automate too much and both workers and owners lose. Once the cost saving is large enough ($\ss\geq\kk+\ell$), automating every task maximizes joint profit as well as each firm's own: the Nash and cooperative rates both equal one, and the over-automation wedge closes (\Cref{prop:alphastar}). Demand does not fall to zero but settles at a floor set by autonomous spending and replaced income. What remains wrong is not too much automation but that the eliminated wages do not return. A planner who values workers ($\mu>0$) can still prefer less than full automation (\Cref{prop:surplus}), and better AI need not change that: the cost saving cannot exceed the wage, and higher output per task leaves this planner's objective unchanged (\Cref{prop:phi}).

For the profit benchmark, the remaining failure is distributional, not allocative, so an automation tax has no margin left to correct. A profit tax, useless for correction during the transition because it cancels from each firm's automation decision (\Cref{sec:capital_tax}), now has no margin left to distort, which makes it the natural instrument for the one task that remains. A universal basic income funded by such a tax returns part of the owners' surplus to households as spending and rebuilds the demand floor. Workers gain directly, but owners recover only the fraction~$\lambda$ of the transfer that households spend in the sector and so bear its net cost. In the transition, UBI could support demand but not correct the margin; here, with no margin left to correct, demand support is all that is needed (\Cref{sec:ubi}). The transfer must be mandated, since firms will not provide it on their own (\Cref{cor:equity_voluntary}). The post-labor case is thus the limit of our result, not an exception. The failure is now distributional rather than allocative, and the corrective instrument is a profit-funded UBI rather than an automation tax.

\section{Extensions}\label{sec:generalizations}

The baseline isolates the demand externality in the simplest environment that supports it. A natural concern is that the result depends on what has been held fixed: endogenous wage adjustment might close the wedge, free entry might discipline the market to an efficient scale, higher AI productivity might resolve the demand problem by expanding the pie, and capital-income recycling might offset the spending lost through displacement. This section takes up each of these objections and shows that the externality is robust to all of them and, in some cases, amplified; it then turns to richer product-market interaction. \Cref{tab:extensions} previews the results.

\begin{table}[ht!]
\centering
\caption{Extensions of the baseline model and effects on the over-automation result.}\label{tab:extensions}
\adjustbox{max width=\linewidth}{%

\begin{tabular}{@{}>{\raggedright\arraybackslash}p{0.27\linewidth}>{\centering\arraybackslash}p{0.10\linewidth}>{\centering\arraybackslash}p{0.18\linewidth}>{\centering\arraybackslash}p{0.16\linewidth}>{\centering\arraybackslash}p{0.16\linewidth}@{}}
\toprule
& Section & Modifies $\nn^{*}$? & Wedge effect & Eliminates externality? \\
\midrule
AI productivity ($\phi > 1$) & \ref{sec:phi} & Yes (lowers) & Widens & No \\
Endogenous entry & \ref{sec:entry} & No & Persists (can widen) & No \\
Endogenous wages & \ref{sec:endo_wages} & Yes (raises) & Narrows (interior) & No \\
Capital-income recycling ($\ehat > 0$) & \ref{sec:recycling} & Yes (raises when $\ell>\ss$) & Narrows & Partially \\
Imperfect competition & \ref{sec:richer_models} & n/a & Can widen or narrow & No \\
\bottomrule
\end{tabular}}

\smallskip
\ifRFCombined\footnotesize\fi\textit{Note.} Each row records how the generalization affects the automation threshold~$\nn^{*}$ (for AI productivity, the formula $\nn^{*} = \ell/\ss$ is unchanged, but the market-share term lowers the firm count at which automation begins), the over-automation wedge, and whether it eliminates the demand externality. The externality survives all five generalizations under empirically plausible parameters.
\end{table}

\subsection{AI Productivity}\label{sec:phi}

In the baseline, AI and human workers produce the same output per task and so the automation incentive is purely cost-driven. In practice, AI can replace humans while also raising output per task (e.g., autonomous agentic coding agents, higher-throughput customer service bots). To capture this, we add a productivity advantage on top of the cost saving. A natural conjecture is that this output channel mitigates the demand problem by making the economy more productive.
We show the opposite is true: higher AI output per task widens the over-automation wedge.

Let an AI-performed task produce $\phi \geq 1$ units of output, while a human-performed task produces~$1$ unit.
With $\phi > 1$, firm~$i$'s output becomes
\[
    Y_i(\alpha_i) = [\phi\alpha_i + (1 - \alpha_i)]\LL = [1 + (\phi - 1)\alpha_i]\LL.
\]
With a uniform market-clearing price, revenue is allocated by output share: $\Rev_i = D \cdot Y_i / \sum_j Y_j$.
At a symmetric profile all firms produce the same output $\bar{Y} = [1 + (\phi-1)\alpha]\LL$, so $\Rev_i = D/\nn$ as in the baseline.
Differentiating with respect to $\alpha_i$ and evaluating at the symmetric profile yields
\begin{align}\label{eq:rev_phi}
  \frac{\partial \Rev_i}{\partial \alpha_i}\bigg|_{\mathrm{sym}} &= \underbrace{-\frac{\ell \LL}{\nn}}_{\text{demand loss}}
  \;+\; \underbrace{\frac{D(\phi - 1)(\nn-1)}{\nn^2[1+(\phi-1)\alpha]}}_{\text{market-share gain}}.
\end{align}
The first term is the baseline demand externality, which depends only on~$\ell$ and~$\nn$ and is therefore independent of~$\phi$.
The second is new: a deviating firm raises its output above rivals and captures a larger share of expenditure.
This market-share gain is positive whenever $\phi > 1$, raising the private incentive to automate above the baseline.

To quantify the effect, we combine cost saving, demand loss, and market-share gain into the first-order condition.
The symmetric equilibrium equates marginal integration cost to the combined benefit:
\[
  \kk\alpha = \ss - \frac{\ell}{\nn} + \frac{D(\alpha)\,(\phi-1)(\nn-1)}{\nn^2[1+(\phi-1)\alpha]\,\LL}.
\]
Because $D(\alpha)$ is linear in~$\alpha$, an interior symmetric equilibrium solves a quadratic; however, the resulting expression is less transparent than the baseline formula, so the comparative statics below are established via a monotone crossing argument.

\begin{proposition}[AI productivity widens the over-automation wedge]\label{prop:phi}
Let $\kk > 0$ and suppose the equilibrium is interior.
\begin{enumerate}[label=(\roman*),nosep]
  \item The Nash equilibrium automation rate is increasing in AI productivity: $\alpha^{NE}(\phi) > \alpha^{NE}(1)$ for $\phi > 1$.
  \item The cooperative optimum and the generalized planner's optimum, however, are unchanged: $\alpha^{CO}(\phi) = \alpha^{CO}(1)$ and $\alpha^{SP}(\mu;\phi) = \alpha^{SP}(\mu;1)$ for all $\mu \in [0,1]$.
  \item Consequently, the over-automation wedge $\alpha^{NE}(\phi) - \alpha^{SP}(\mu;\phi)$ is strictly increasing in~$\phi$ for every $\mu \in [0,1]$.
\end{enumerate}
\end{proposition}

The mechanism is a Red Queen effect: each firm perceives a market-share gain from automating beyond rivals, but at the symmetric equilibrium all firms expand equally, so the gains cancel.
By contrast, the cost saving~$\ss$ enters each firm's profit identically regardless of rivals' choices and carries no strategic component; in the baseline it shifts $\alpha^{NE}$ and $\alpha^{CO}$ equally and leaves the wedge unchanged.

Part~(ii) holds because total sectoral revenue equals total expenditure~$D$ under market clearing, and $D$~\eqref{eq:demand} depends on worker income, not output: higher~$\phi$ raises output but lowers the price in proportion, leaving the planner's objective invariant to~$\phi$.
Together with part~(i), this yields part~(iii): better AI raises the equilibrium automation rate without shifting the efficient benchmark, so the distortion grows with AI capability.

The wider wedge also carries a policy implication. Because the market-share motive adds a second distortion on top of the demand externality, the baseline Pigouvian rate $\tau^{*} = \ell(1-1/\nn)$ no longer suffices: implementing~$\alpha^{CO}$ requires an additional correction equal to the market-share term in~\eqref{eq:rev_phi}, divided by~$\LL$, evaluated at the cooperative rate, which is strictly positive whenever $\phi > 1$.
That said, the lower price means each dollar of spending buys more physical output, so the welfare measure $S(\mu)$, built from nominal flows, understates the real consumption gains from higher~$\phi$.
The proposition identifies a strategic distortion, not a claim that higher AI productivity reduces total welfare on net.

\subsection{Endogenous Entry}\label{sec:entry}

So far the number of firms has been exogenous.
With free entry, one might expect the over-automation problem to be self-correcting: surplus erosion lowers profits, marginal firms exit, and the remaining industry settles at an efficient scale.
Whether this logic goes through, however, depends on how the entry margin interacts with the automation decision.

Consider a two-stage game: firms pay a fixed cost $\kappa \geq 0$ to enter, then simultaneously choose automation rates.
Given~$\nn$ entrants, the stage-2 Nash equilibrium yields per-firm operating profit $\Pi^{*}(\nn)$.
A pure-strategy \emph{free-entry equilibrium} is an integer $\nn \geq 1$ such that
\begin{equation}\label{eq:free_entry}
  \Pi^{*}(\nn) \geq \kappa
  \qquad\text{and}\qquad
  \Pi^{*}(\nn + 1) < \kappa:
\end{equation}
incumbents weakly prefer to remain active, while an additional entrant would not recover the fixed cost.

Because the frictionless ($\kk = 0$) and convex-cost ($\kk > 0$) regimes shape the profit schedule in qualitatively different ways, we treat each in turn. The frictionless case sets $\lambda = 1$ for clean closed forms; the qualitative results hold for any $\lambda \in (0,1]$, as the assumption affects only profit levels, not the structure of the entry regimes.

In the frictionless benchmark ($\kk = 0$), automation is all-or-nothing.
When $\ell > \ss$, the profit schedule drops discretely at~$\nn^{*}$: below the threshold no firm automates, while above it full automation is dominant (\Cref{cor:frictionless}) and per-firm profit falls by $\Delta \coloneqq \LL(\ell - \ss) > 0$; \ifRFCombined Appendix~\ref{app:proofs}\else the Internet Appendix\fi\ illustrates this schedule.
Write $m \coloneqq \lfloor \nn^{*} \rfloor$ for the largest integer not exceeding~$\nn^{*}$.
We assume, as a genericity condition, that the threshold is not itself an integer: $\nn^{*} \notin \mathbb{N}$.
Because the number of firms is a whole number, this places every market size strictly below the threshold (no firm automates) or strictly above it (all do), with no ambiguous case sitting exactly on the boundary.
The condition is innocuous, failing only on the measure-zero set of parameters for which $\ell/\ss$ is exactly a whole number.\footnote{Exactly at an integer threshold, $\ss = \ell/\nn$, so each firm's private gain from automating is zero and every firm is indifferent over its own rate; yet a rival's automation still lowers the firm's payoff by contracting demand~$D$. The stage-2 equilibrium profit is then not single-valued, and pinning it down would require an arbitrary selection rule, which the assumption avoids.}

\begin{proposition}[Endogenous entry in the frictionless benchmark]\label{prop:entry}
Suppose $\kk = 0$, $\lambda = 1$, $\ell > \ss$, $0 < \kappa < A$, and, generically, $\nn^{*} \notin \mathbb{N}$.
The per-firm equilibrium profit schedule
\[
  \Pi^{*}(\nn) =
  \begin{cases}
    A/\nn         & \text{if } \nn \leq \nn^{*}, \\[2pt]
    A/\nn - \Delta & \text{if } \nn > \nn^{*},
  \end{cases}
\]
is strictly decreasing on~$\mathbb{N}$.
The unique pure-strategy free-entry equilibrium has
$\nn^{FE} = \max\{\nn \in \mathbb{N} : \Pi^{*}(\nn) \geq \kappa\}$ firms.
Three regimes arise depending on the entry cost:
\begin{enumerate}[label=(\roman*)]
  \item \emph{Low entry cost} ($\kappa \leq A/(m+1) - \Delta$):\;
    $\nn^{FE} = \lfloor A/(\kappa + \Delta) \rfloor \geq m + 1$.
    Every firm fully automates.
  \item\label{prop:entry:barrier} \emph{Intermediate entry cost} ($A/(m+1) - \Delta < \kappa < A/m$):\;
    $\nn^{FE} = m$.
    No firm automates, yet each earns strictly positive profit $A/m > \kappa$.
    For $\kappa \leq A/(m+1)$, the \emph{threat} of automation is what deters further entry: an additional firm would cover its entry cost absent automation, but its entry would trigger full automation, dropping per-firm profit to $A/(m+1) - \Delta < \kappa$.
  \item \emph{High entry cost} ($\kappa \geq A/m$):\;
    $\nn^{FE} = \lfloor A/\kappa \rfloor \leq m$.
    No firm automates.
    Entry costs, not the automation externality, limit the number of firms.
\end{enumerate}
\end{proposition}

When case~(i) arises, the Prisoner's Dilemma of \Cref{cor:frictionless} materializes under free entry: all firms automate, demand contracts, and every firm would be better off had none automated.
Case~(iii) is the standard free-entry outcome: entry costs are high enough that the market never approaches the automation threshold, and the externality is irrelevant.
Case~(ii) is the most distinctive: in its lower range, the \emph{threat} of automation functions as an endogenous entry barrier, sustaining positive profits without any automation actually occurring, at the cost of sustaining market power.

The convex-cost case is less stark but more robust.
With $\kk > 0$, the automation rate $\alpha^{NE}(\nn) = (\ss - \ell/\nn)/\kk$ varies continuously in~$\nn$, so the profit schedule no longer jumps at~$\nn^{*}$; the entry deterrence of \Cref{prop:entry}\ref{prop:entry:barrier} survives in discrete form, through the drop in $\Pi^{*}(\nn)$ when the first entrant above~$\nn^{*}$ activates automation.

\begin{proposition}[Endogenous entry with convex costs]\label{prop:entry_convex}
Suppose $\kk > 0$, $\ell > \ss$, $\kappa > 0$, and $\Piz(1) > \kappa$ (the market is viable).
A free-entry equilibrium $\nn^{FE}$ satisfying~\eqref{eq:free_entry} exists.
If $\nn^{FE} > \nn^{*}$, then $\alpha^{NE}(\nn^{FE}) > \alpha^{CO}$: over-automation persists under free entry.
\end{proposition}

The market moves past the automation threshold, $\nn^{FE}>\nn^{*}$, if and only if it can support one firm beyond it: $\Pi^{*}(\lfloor\nn^{*}\rfloor+1)\geq\kappa$.
Free entry then pins down the number of firms but does not alter the strategic incentives within the automation subgame: each firm still bears only a fraction~$1/\nn^{FE}$ of the demand loss. Over-automation persists because $\ell>\ss$ forces $\alpha^{CO}=0$, so the realized wedge equals $\alpha^{NE}(\nn^{FE})=\min((\ss-\ell/\nn^{FE})/\kk,1)>0$ (the interior expression $\ell(1-1/\nn^{FE})/\kk$ applies only where the cooperative rate is itself interior).
If $\nn^{FE} \leq \nn^{*}$, no firm automates and automation is efficient, but only because the market is too concentrated for the private automation incentive to activate.

Taken together, the two propositions deliver a common lesson: free entry reshapes the over-automation problem but does not resolve it. If anything, the standard tendency toward excess entry \citep{mankiw1986free} widens the wedge by fragmenting the market further.

The entry margin reveals an unintended side effect of UBI (\Cref{sec:ubi}). By raising autonomous demand~$A$, UBI increases per-firm profit $\Piz = A/\nn + (\lambda-1)\ww\LL$ at any given~$\nn$, attracting additional entrants until the zero-profit condition~\eqref{eq:free_entry} binds at a larger~$\nn^{FE}$. Since the realized over-automation rate $\alpha^{NE}=\min((\ss-\ell/\nn)/\kk,1)$ is increasing in~$\nn$ under the maintained $\ell>\ss$ (where $\alpha^{CO}=0$), a policy designed to cushion displacement can paradoxically widen the very externality that causes it.

\subsection{Endogenous Wages}\label{sec:endo_wages}

A central insight of \citet{acemoglu2018race} is that endogenous wage adjustment can stabilize the automation path: as firms automate, displaced workers increase labor supply, pushing wages down; lower wages narrow the cost saving from automation and discourage further displacement.
This self-correcting feedback is a natural candidate for resolving the demand externality identified above.
We show that it weakens the incentive to automate but does not close the wedge: if automation pays at the pre-automation wage, it persists at the equilibrium wage, and the Nash rate still exceeds the cooperative rate. What narrowing occurs comes through lower wages.

In \citet{acemoglu2018race}, wages are determined by labor-market clearing in a full general equilibrium; we adopt a reduced-form representation that captures the key qualitative feature of their mechanism.
Let the wage depend on the aggregate automation rate: $\ww(\abar)$ with $\ww(\abar) > \cc$ and $\ww'(\abar) \leq 0$; firms take the prevailing wage as given when choosing~$\alpha_i$. This specification requires only that wages fall when aggregate labor-market slack increases, a property shared by efficiency-wage models, where the no-shirking wage declines with unemployment \citep{shapiro1984equilibrium}. In our setting, automation displaces workers into the labor pool, generating exactly this type of slack.

Both the cost saving $\ss(\ww) = \ww - \cc$ and the demand-loss parameter $\ell(\ww) = \lambda(1-\eta)\ww$ are increasing in the wage, so falling wages affect both sides of the automation margin: they shrink the private incentive to automate (the self-correcting channel) and reduce the demand loss per automated task.
The equilibrium automation rate is therefore a fixed point in which automation, wages, and the externality are jointly determined.
Despite this richer feedback, the threshold $\nn^{*} = \ell/\ss = \lambda(1-\eta)\ww/(\ww - \cc)$ rises as wages fall, because the cost saving $\ss = \ww - \cc$ contracts faster than the demand loss $\ell = \lambda(1-\eta)\ww$.
The structural source of the distortion, however, is unaffected.

\begin{proposition}[Robustness to wage adjustment]\label{prop:endo_wages}
Let $\kk > 0$, let $\ww: [0,1] \to (\cc, \infty)$ be differentiable with $\ww'(\abar) \leq 0$, and suppose firms are wage-takers in both the Nash and the cooperative comparison.\footnote{The cooperative benchmark $\alpha^{CO}$ is the common automation rate that maximizes aggregate firm profit when firms remain price-takers in the labor market, so the wage is taken parametrically and only evaluated at $\ww(\abar)$ at the symmetric equilibrium. This isolates the cross-firm demand externality from a separate monopsonistic-incidence channel that would arise if a coalition exploited the wage schedule $\ww(\abar)$ directly. This is a restricted wage-taking benchmark, not the optimum of a planner who internalizes the wage response.}
\begin{enumerate}[label=(\roman*),nosep]
  \item At any symmetric equilibrium with $\nn > \nn^{*}(\ww(\abar))$, the Nash automation rate exceeds the cooperative optimum: $\alpha^{NE} > \alpha^{CO}$, unless the cost saving is so large that both rates pin at full automation ($\ss(\abar) \geq \kk + \ell(\abar)$), where the strict inequality lapses.
  \item Endogenous wage adjustment raises the threshold: $\nn^{*}(\ww(\abar)) \geq \nn^{*}(\ww(0))$ for all $\abar \in [0,1]$, with strict inequality whenever $\ww(\abar) < \ww(0)$, $\cc > 0$, and $\eta < 1$ (at $\cc = 0$ the threshold $\nn^{*} = \lambda(1-\eta)$, and at $\eta = 1$ the threshold $\nn^{*} = 0$, are independent of the wage, so wage decline does not move them).
\end{enumerate}
\end{proposition}

In both benchmarks the wage adjusts with automation, but firms take it as given. The cooperative rate therefore differs from the Nash rate only in that firms account for the sales their automation destroys at rivals.\footnote{Firms that coordinated to push wages down could prefer a different rate, but that gain would come at workers' expense, a transfer rather than an efficiency gain, so it is not the right target for policy.}

Competitive pricing allocates revenue as $\Rev_i = D/\nn$ at any wage level, so each firm bears only a fraction of the demand destruction its automation causes regardless of whether~$\ww$ is high or low. Wage adjustment changes the \emph{magnitude} of~$\ell$ but not the \emph{fraction} each firm internalizes; that fraction is a property of market structure, not of factor prices.

Could wages fall far enough to shut the race down? At a given wage, no firm automates once $\ww[1-\lambda(1-\eta)/\nn]\leq\cc$, a cutoff strictly above~$\cc$ when $\cc>0$ and $\eta<1$. In the model, however, wages fall only as automation proceeds, so if automation pays at the pre-automation wage~$\ww(0)$, the equilibrium rate stays positive and wage adjustment does not close the gap. And where other forces do push wages that low, the resolution is Pyrrhic: wages are pushed toward the cost of the machines that would replace the workers, and the demand problem reappears as depressed wages rather than displacement.

The analysis above sets $\mu = 0$, so the planner cares only about firm profits. A planner who also values worker welfare ($\mu > 0$) demands a larger correction (\cref{eq:wedge_decomp}), and wage adjustment shrinks both parts of that wedge without closing either. The \IAname{} confirms this: whenever the $\mu$-planner's marginal benefit~$g_\mu$ is strictly decreasing (which holds for all $\mu$ up to a threshold $\bar\mu$, and beyond it when integration frictions dominate wage sensitivity), the over-automation result extends to any $\mu \in [0,1)$ under endogenous wages.

\subsection{Capital Income Recycling}\label{sec:recycling}

\Cref{sec:eta} showed that raising~$\eta$, the fraction of displaced income recovered by workers, shrinks the demand-loss parameter~$\ell$ and narrows the over-automation wedge. A natural counterpart on the capital side is that owners spend their profits: if their consumption offsets the spending lost through displacement, the demand externality might disappear. We show that recycling narrows the wedge but cannot close it under empirically plausible parameters.

Suppose capital owners consume a fraction $\ehat \in [0, 1)$ of their capital income in the sector.
First set $\kk=0$. Total sector profit is $\Pi = D - \nn\LL(\ww - \ss\abar)$.
Adding capital consumption~$\ehat\Pi$ to aggregate demand and solving for~$D$ yields
\begin{equation}\label{eq:D_delta}
  D = \frac{A + (\lambda-\ehat)\ww\LL\nn}{1 - \ehat} - \frac{\ell_{\ehat} \LL\nn}{1 - \ehat}\,\abar,
\end{equation}
where
\[
  \ell_{\ehat} = \ell - \ehat\ss
\]
is the effective demand-loss parameter: each automated task loses~$\ell$ in worker spending, but owners recycle~$\ehat$ of the per-task saving~$\ss$ back into demand.
When $\ehat = 0$, $\ell_{\ehat} = \ell$ and \cref{eq:D_delta} reduces to \cref{eq:demand}.
This is the frictionless ($\kk = 0$) specialization; for $\kk > 0$ the fixed point carries an additional minus $(\ehat\kk/[2(1-\ehat)]) \LL \sum_j \alpha_j^2$ term (derived in the proof of \Cref{prop:recycling}).

Competitive pricing still gives $\Rev_i = D/\nn$.
The first-order condition becomes $\LL(\ss - \ell_{\ehat}/[\nn(1-\ehat)])$, giving a modified threshold
\[
  \nntil \coloneqq \frac{\ell_{\ehat}}{\ss(1-\ehat)}.
\]

\begin{proposition}[Capital income recycling]\label{prop:recycling}
Suppose $\kk = 0$.
When there is capital income recycling at rate $\ehat$,
\begin{enumerate}[label=(\roman*),nosep]
  \item When $\ell_{\ehat} > 0$, full automation is strictly dominant if and only if $\nn > \nntil$.
  \item The demand loss vanishes or reverses ($\ell_{\ehat}\leq0$) if and only if $\ehat\geq\ell/\ss$.
\end{enumerate}
\end{proposition}

Part~(ii) requires $\ehat \geq \ell/\ss = \lambda(1-\eta)\ww/(\ww-\cc)$: owners must recycle enough of each task's cost saving to replace the demand that displaced workers would have generated. When $\ell > \ss$, the required rate exceeds one, so recycling cannot remove the demand loss precisely where the externality is most harmful, since $\ell > \ss$ implies $\alpha^{CO} = 0$ and firms automate when the planner would prefer none. When $\ell < \ss$, elimination is feasible in principle, but the planner already prefers positive automation (\Cref{prop:alphastar}) and the wedge is quantitatively smaller.

The frictionless case gives the sharpest result, but the structure carries over when frictions are positive. The proof of \Cref{prop:recycling} extends the result to $\kk > 0$. The Nash equilibrium generalizes to $\alpha^{NE} = \min\{1,\max\{0,(\ss - \ell/\nnhat)/\kk\}\}$, where $\nnhat = \nn(1-\ehat)+\ehat$ is an effective market size that interpolates between~$\nn$ (no recycling) and~$1$ (full recycling), making each firm behave as though it faced fewer competitors. The cooperative optimum, however, is unchanged: the $1/(1-\ehat)$ multiplier scales total profit without shifting the optimizer.

The upshot parallels \Cref{sec:eta}: for any recycling rate below one, recycling raises the fraction of demand loss each firm internalizes from $1/\nn$ to $1/\nnhat$, but cannot push it to one. Addressing how income is spent narrows the wedge but does not close it when both rates are interior, because the underlying dilution across firms persists. When the planner prefers no automation, recycling short of one can halt private automation altogether (see the proof of \Cref{prop:recycling}).

\subsection{Imperfect Product-Market Competition and Task Complementarity}\label{sec:richer_models}

The baseline assumes competitive pricing and perfect substitution across tasks. Relaxing either changes the size of the demand externality but not its source, because neither alters the one fact the externality rests on: automation reduces the \emph{level} of aggregate expenditure by displacing workers, and each firm bears only a fraction of that loss. We give the argument informally; a formal treatment is left for future work.

\paragraph{Second-stage price or quantity competition.}\leavevmode\allowbreak
Suppose firms compete on quantities or prices after choosing automation. This changes how a given level of expenditure is \emph{divided} among firms and rewards a firm that automates more than its rivals with a larger market share. Neither force removes the externality. However expenditure is split, each firm still bears only a fraction of the demand its own automation destroys, so the externality persists whenever a firm lacks full market power; with differentiated products the uninternalized loss scales with the share of the market the firm does \emph{not} capture. The market-share gains cancel across firms at the symmetric equilibrium, echoing the Red Queen effect of \Cref{sec:phi}, and under Cournot the motive is partly self-correcting, since a firm that gains share also absorbs more of the demand loss it causes. As long as firms lack full market power, market structure moves the size of the externality without removing it; only in the monopoly limit, where a single firm internalizes the entire demand loss it causes, does the externality disappear.

\paragraph{CES task aggregation.}\leavevmode\allowbreak
The baseline takes the perfect-substitutes limit of the CES task aggregator in \citet{acemoglu2018race}. Imperfect substitution does not change how displacement maps into expenditure; the supply-side restraint from ordering tasks by comparative advantage is already in the friction~$\kk$. The demand channel is untouched, so the externality remains positive as long as displaced workers lose income.

\medskip
In every case the source of the distortion is the same and survives: firms do not internalize the demand they destroy. Richer modeling moves the magnitude, not the mechanism, which the baseline isolates in its starkest form.

\section{Discussion}\label{sec:discussion}

This paper develops a simple model with a stark insight. Even as AI-driven layoffs sweep across industries, and even as every firm recognizes that vanishing paychecks mean vanishing customers, not one of them will stop. Each firm reaps the full savings of replacing its own workers yet bears only a sliver of the demand it destroys; the rest lands on rivals. No firm can afford to be the one that holds back. This is the trap: an automation arms race that can intensify as AI improves, that leaves workers and firm owners alike worse off, and that no firm can escape on its own. We close by discussing implications for empirics and policy, and then the scope and limitations of the analysis.

\paragraph{Empirical implications.}
Anthropic CEO Dario Amodei has warned that AI-driven displacement will be ``unusually painful,'' ``much broader'' and ``much faster'' than previous technological shocks \citep{cnbc2026amodei}. If that assessment proves correct and income replacement remains incomplete, the model points, perhaps counterintuitively, to where the problem is most severe: not dominant technology firms but fragmented industries deploying the most capable AI (\Cref{prop:alphastar,prop:phi}). One empirical signature would be profit erosion. Standard competitive models predict that cost-reducing technology raises profits; profit erosion that coincides with mass layoffs and falling sector spending, rather than with cost pass-through to prices or an outside demand shock, would be difficult to rationalize without the externality. The model delivers it when the cost saving is small relative to the demand loss (\Cref{cor:frictionless}\ref{prop:PD}); where the saving is large enough, firms over-automate and profits still rise, so the externality is present even without this signature. That said, this signature requires displacement at a scale and speed beyond what has materialized so far. If reabsorption keeps pace with automation, the externality may remain too small to detect, and the paper's contribution is identifying a structural vulnerability rather than diagnosing an active crisis. Three settings where AI-driven displacement is already under way offer concrete starting points: customer support, where Salesforce reported cutting 4,000 roles as it adopted AI \citep{cnbc2025salesforce}; software services, where Goldman Sachs has piloted autonomous coding tools \citep{cnbc2025devin}; and back-office operations across competing financial institutions, where regulatory reporting makes revenue outcomes unusually transparent.

\paragraph{Policy implications.}
Much of the policy debate around AI-driven displacement focuses on how to respond after the fact, through retraining, income support, or regulation. Our results reframe the question: do competitive incentives drive firms to automate beyond what is collectively optimal? Even a planner who places zero weight on worker welfare would reduce the automation rate below the equilibrium level (\Cref{prop:surplus}). Because over-automation leaves both firms and workers worse off, correcting it is a matter of eliminating waste, not of redistributing gains between them. Universal basic income, perhaps the most widely discussed response, raises living standards but does not change a single firm's incentive to automate (\Cref{sec:ubi}). Collective bargaining faces the same wall: because automation is a dominant strategy, no voluntary agreement among firms to restrain layoffs is self-enforcing (\Cref{sec:coase}). By Tinbergen's principle, a distinct market failure requires a distinct instrument; among the instruments we evaluate, only a Pigouvian automation tax supplies it (\Cref{tab:policy}). This ordering reverses in the post-labor limit, where the arms race has nothing left to correct and a profit-funded basic income becomes the operative policy rather than a palliative (\Cref{sec:postlabor}).

One practical consideration bears on implementation: the model is a closed-sector game, and a unilateral automation tax could push adoption offshore, strengthening the case for multilateral coordination or border-adjustment mechanisms analogous to those used in carbon policy.

\paragraph{Scope, limitations, and future directions.}
The model is deliberately simple: one sector, one period, symmetric firms. Several of these choices are conservative, meaning the real problem is likely worse than what we show.

A single sector can understate the externality. In a multi-sector economy, layoffs in one sector reduce spending on every sector's output, creating reinforcing demand spirals. Platform ecosystems make the point concrete: when a platform automates seller support, gig logistics, or content moderation, the lost spending cascades across an entire ecosystem of complementors.

\IAref{ge:appendix} establishes an owner-welfare loss in a solved example of a constrained short-run general equilibrium in which supplier income, household saving, and consumption prices are all accounted for. Output there falls with automation at an interior coordinated optimum when the wage spending a displaced worker withdraws from the sector exceeds the spending on automation services that replaces it ($\theta_W\ww>\cc$, with $\theta_W$ the worker spending share that plays the role of~$\lambda$). This is the GE counterpart of $\ell>0$, with the AI cost netted out because automation services are purchased within the sector. The appendix identifies the contracts, financial constraints, and production opportunities on which the result relies; it does not establish the optimality of the automation tax in GE.

The static setting misses two dynamic forces that pull in opposite directions. AI investments are largely irreversible, and \Cref{prop:entry} shows that even the \emph{threat} of automation can reshape market structure before any displacement occurs, strengthening the case for early policy action. Working the other way, the income-replacement rate~$\eta$ rises over time as displaced workers retrain and new occupations emerge \citep{acemoglu2019automation}, so the optimal tax should shrink as the economy adjusts (\Cref{sec:tax}).

Symmetry rules out heterogeneity across firms and workers, and endogenizing AI development could compound the problem: firms racing to automate may invest disproportionately in labor-replacing AI rather than labor-augmenting AI \citep{acemoglu2018race}, feeding the very arms race the model identifies.

Several of these extensions point toward a larger problem; the main self-correcting force is gradual reabsorption (a rising income-replacement rate~$\eta$), which shrinks the corrective tax over time and removes the externality only when displaced workers recover their full income ($\eta=1$). A further connection worth exploring is the interaction with algorithmic pricing: when AI systems that automate tasks also set prices, collusive markups may partially internalize the demand externality but simultaneously strengthen the private incentive to automate \citep{banchio2022artificial,keppo2026fragility}. Pursuing these and other extensions, along with the empirical tests outlined above, are promising directions for future work.

\section*{Disclosure of AI use}
The original manuscript and main model were written without significant use of AI assistance. The revised manuscript leveraged Refine.ink for verification purposes and Claude Fable 5.1 to address the issues flagged and streamline the exposition. The general equilibrium results in \IAref{ge:appendix} were derived with the assistance of Claude Fable 5.1 and verified independently.

\section*{Data availability}
The authors did not generate or analyze new data for this research.

\bibliographystyle{plainnat}
\ifRFCombined
\bibliography{refs_final}
\else
\bibliography{refs_final}
\fi

\fi
\ifRFAppendix
\ifRFCombined\clearpage\fi
\appendix
\section{Proofs}\label{app:proofs}

\begin{proof}[\textbf{Proof of \Cref{prop:alphastar}}]
~\\
\emph{Preamble and part~(i): Nash equilibrium.}
From \cref{eq:profit_expanded}, we have
\[
\pi_i = \Piz + \LL\!\left[\alpha_i\!\left(\ss - \frac{\ell}{\nn}\right) - \frac{\kk}{2}\alpha_i^2 - \frac{\ell}{\nn}\sum_{j \neq i}\alpha_j\right]
\]

Differentiating with respect to~$\alpha_i$ (only the terms involving~$\alpha_i$ contribute, since the rivals' sum $\sum_{j \neq i}\alpha_j$ is a constant from firm~$i$'s perspective):
\begin{equation}\label{eq:FOCalpha}
    \frac{ \partial \pi_i }{\partial \alpha_i } = \LL \inbrak{ \inparen{ \ss - \frac{\ell}{\nn}} - \kk \alpha_i }.
\end{equation}
The second derivative is
\[
    \frac{ \partial^2 \pi_i }{\partial \alpha_i^2 } = -\kk\LL < 0,
\]
so $\pi_i$ is strictly concave in~$\alpha_i$.

Setting the first-order condition to zero:
\[
  \LL\!\left(\ss - \frac{\ell}{\nn} - \kk\alpha_i\right) = 0
  \qquad\Longrightarrow\qquad
  \alpha_i^* = \frac{\ss - \ell/\nn}{\kk}.
\]

Now, $\alpha_i$ is restricted to the range $[0,1]$ (since it is the fraction of work automated).

If $\alpha_i^* > 1$, then since $\frac{ \partial^2 \pi_i }{\partial \alpha_i^2 } < 0$, $\frac{ \partial \pi_i }{\partial \alpha_i } > 0$ for $\alpha_i < \alpha_i^*$, so in this case, the maximum value of $\pi_i(\alpha_i)$ occurs at $\alpha_i = 1$.

Similarly, if $\alpha_i^* < 0$, then $\frac{ \partial \pi_i }{\partial \alpha_i } < 0$ for $\alpha_i > \alpha_i^*$, so in this case, the maximum value of $\pi_i(\alpha_i)$ occurs at $\alpha_i = 0$.
Note that $\alpha_i^* < 0$ if and only if $\ss < \ell/\nn$, i.e., $\nn < \ell/\ss = \nn^{*}$; this establishes the preamble claim that no firm automates when $\nn \leq \nn^{*}$.

Note that $\alpha_i^*$ does not depend on any rival's choice $\alpha_j$ ($j \neq i$): the rivals' automation levels enter only through the additive term $-(\ell/\nn)\sum_{j \neq i}\alpha_j$, which does not affect the first-order condition.
Hence firm~$i$'s optimal strategy is independent of~$\alpha_{-i}$ (a dominant strategy), and is to set
\[
  \alpha_i = \frac{\ss - \ell/\nn}{\kk},
\]
provided the right-hand side lies in $[0,1]$; otherwise $\alpha_i = 0$ or $\alpha_i = 1$ at the respective boundary.
Since every firm solves the same problem, the unique symmetric Nash equilibrium has $\alpha_i = \alpha^{NE}$ for all~$i$.

\medskip
\emph{Part~(ii): Cooperative optimum.}
    Summing the per-firm profit~\eqref{eq:profit_clean} over all $\nn$ firms, the sector's total profit is
    \begin{align*}
        \pi_{tot} &\coloneqq
        \sum_{i=1}^\nn \pi_i \\
        &= \sum_{i=1}^\nn \inbrak{ \Piz + \LL \inparen{ \ss \alpha_i - \ell \abar - \frac{\kk}{2} \alpha_i^2}} \\
        &= \nn \Piz + \LL\ss\sum_i \alpha_i - \LL\ell\nn\abar - \frac{\kk\LL}{2} \sum_{i=1}^\nn \alpha_i^2.
    \end{align*}
Using $\sum_i \alpha_i = \nn\abar$, this simplifies to
\[
    \pi_{tot} = \nn \Piz + \nn\LL \abar(\ss - \ell) - \frac{\kk\LL}{2} \sum_{i=1}^\nn \alpha_i^2.
\]
The first two terms depend on $\boldsymbol{\alpha}$ only through $\abar$; the last term, $-\tfrac{\kk\LL}{2}\sum_i \alpha_i^2$, is maximized (least negative) when the $\alpha_i$ are as equal as possible, by convexity of~$x^2$.
Formally, for any fixed $\abar$, $\sum_i \alpha_i^2 \geq \nn\abar^2$ with equality if and only if $\alpha_i = \abar$ for all~$i$ (by the QM--AM inequality).
Hence the optimum is symmetric: $\alpha_i = \alpha$ for all~$i$, and $\abar = \alpha$.
Substituting:
\begin{equation}\label{eq:pitotsym}
     \pi_{tot} = \nn \Piz + \nn \LL \alpha(\ss-\ell) - \frac{\kk\nn\LL \alpha^2}{2}.
\end{equation}
This is a concave quadratic in~$\alpha$:
\[
    \frac{ \partial \pi_{tot} }{\partial \alpha} = \nn\LL(\ss-\ell) - \kk\nn\LL\alpha,
    \qquad
    \frac{ \partial^2 \pi_{tot} }{\partial \alpha^2} = -\kk\nn\LL < 0.
\]
Setting the first-order condition to zero:
\[
   \nn\LL(\ss-\ell) - \kk\nn\LL\alpha = 0 \quad \Rightarrow \quad \alpha^* = \frac{\ss-\ell}{\kk}.
\]
Since $\pi_{tot}$ is strictly concave, this is the unique global maximum on~$\mathbb{R}$.
Restricting to $\alpha \in [0,1]$:
if $(\ss-\ell)/\kk > 1$, strict concavity implies $\partial\pi_{tot}/\partial\alpha > 0$ on $[0,1]$, so the constrained maximum is $\alpha = 1$;
if $(\ss-\ell)/\kk < 0$, then $\partial\pi_{tot}/\partial\alpha < 0$ on $[0,1]$, so the constrained maximum is $\alpha = 0$.
Thus
\[
    \alpha^{CO} = \frac{\ss-\ell}{\kk},
\]
with boundary values $0$ or $1$ when the expression falls outside $[0,1]$.

\medskip
\emph{Part~(iii): Over-automation wedge.}
Suppose $\nn > \nn^{*}$, so $\alpha^{NE} = \min\{(\ss - \ell/\nn)/\kk,1\} > 0$ by~(i).

\emph{Case $\ss > \ell$.}
Restrict to $\ell < \ss < \kk + \ell/\nn$, so both $\alpha^{NE} = (\ss - \ell/\nn)/\kk$ and $\alpha^{CO} = (\ss - \ell)/\kk$ are interior.
Subtracting:
\[
  \alpha^{NE} - \alpha^{CO}
    = \frac{\ss - \ell/\nn}{\kk} - \frac{\ss - \ell}{\kk}
    = \frac{\ell - \ell/\nn}{\kk}
    = \frac{\ell\,(1 - 1/\nn)}{\kk}.
\]
This is strictly positive for $\nn \geq 2$ because $\ell > 0$ and $1 - 1/\nn > 0$.

For $\kk + \ell/\nn \le \ss < \kk + \ell$, $\alpha^{NE} = 1$ while $\alpha^{CO} = (\ss - \ell)/\kk \in (0,1)$, so the wedge is $1 - (\ss - \ell)/\kk$; this equals the plateau $\ell(1 - 1/\nn)/\kk$ at $\ss = \kk + \ell/\nn$ (when $\kk > \ell(1-1/\nn)$, so that the both-interior interval is nonempty; otherwise the wedge equals $1$ up to $\ss=\ell$ by the corner case and the descending branch starts there) and decreases to $0$ as $\ss$ approaches $\kk + \ell$. For $\ss \ge \kk + \ell$ both rates equal $1$ and the wedge is $0$.

For the comparative statics, write the wedge as $W = \ell(1-1/\nn)/\kk$:
\begin{align*}
  \frac{\partial W}{\partial \nn} &= \frac{\ell}{\kk\nn^2} > 0
    &&\text{(increasing in~$\nn$)}, \\
  \frac{\partial W}{\partial \ell} &= \frac{1-1/\nn}{\kk} > 0
    &&\text{(increasing in~$\ell$)}, \\
  \frac{\partial W}{\partial \kk} &= -\frac{\ell(1-1/\nn)}{\kk^2} < 0
    &&\text{(decreasing in~$\kk$)}.
\end{align*}

\emph{Part~(iv): Boundary cases.}

When $\ss \le \ell$, then $(\ss - \ell)/\kk \leq 0$, so $\alpha^{CO} = 0$.
So the wedge is $\alpha^{NE} - \alpha^{CO} = \alpha^{NE}$.
Since we are in the setting where $\nn > \nn^{*}$, we have $(\ss - \ell/\nn)/\kk > 0$.
It is still possible, however, that $(\ss - \ell/\nn)/\kk > 1$.

If $\ss < \kk + \ell/\nn$, then $\alpha^{NE} = (\ss-\ell/\nn)/\kk$, and this is the wedge.
If $\ss \ge \kk + \ell/\nn$, then $\alpha^{NE}$ is a corner case, $\alpha^{NE} = 1$, and this is the wedge.
\end{proof}

\begin{lemma}[Boundary cases]\label{lem:boundary}
  For $\mu\in[0,1)$, write $\delta(\mu) \coloneqq \frac{\mu\,\ell}{\lambda(1-\mu)}$ for the additional planner correction beyond $\alpha^{CO}$.
  Then, for $\kk > 0$
  \begin{enumerate}[label=(\roman*),nosep]
    \item $\alpha^{NE}$ is interior if and only if $\ell/\nn < \ss < \kk + \ell/\nn$;
      Note that $\nn > \nn^{*}$ is equivalent to $\ss > \ell/\nn$.
    \item $\alpha^{CO}$ is interior if and only if $\ell < \ss < \kk + \ell$;
    \item $\alpha^{SP}(\mu)$ is interior if and only if $\ell + \delta(\mu) < \ss < \kk + \ell + \delta(\mu)$.
  \end{enumerate}
  Each interval has width~$\kk$.
  Because $\ell/\nn < \ell < \ell + \delta(\mu)$ for $\mu > 0$, the three windows are progressively shifted to the right: interior automation arises at the lowest cost savings under Nash behavior, at intermediate savings under cooperation, and only at the highest savings under the social planner.
\end{lemma}

\begin{figure}[htbp]
\centering
\begin{tikzpicture}[>=stealth, thick,
    NE/.style  ={fill=blue!15},
    CO/.style  ={fill=red!15},
    SP/.style  ={fill=green!15}]
  \def\lN{1.0}%
  \def\el{3.5}%
  \def\elD{6.0}%
  \def\kv{4.0}%
  \pgfmathsetmacro{\rNE}{\lN+\kv}%
  \pgfmathsetmacro{\rCO}{\el+\kv}%
  \pgfmathsetmacro{\rSP}{\elD+\kv}%

  \draw[->] (-0.3,0) -- (10.5,0) node[right]{$\ss$};

  \fill[NE] (\lN,-0.70) rectangle (\rNE,-0.30);
  \draw (\lN,-0.70) rectangle (\rNE,-0.30);
  \node[left] at (-0.3,-0.50) {\small $\alpha^{NE}$};

  \fill[CO] (\el,-1.30) rectangle (\rCO,-0.90);
  \draw (\el,-1.30) rectangle (\rCO,-0.90);
  \node[left] at (-0.3,-1.10) {\small $\alpha^{CO}$};

  \fill[SP] (\elD,-1.90) rectangle (\rSP,-1.50);
  \draw (\elD,-1.90) rectangle (\rSP,-1.50);
  \node[left] at (-0.3,-1.70) {\small $\alpha^{SP}$};
  \draw [<->] (\el,-1.70) to node[below] {$\delta$} (\elD,-1.70);

  \foreach \x/\lab in {%
    \lN/{$\frac{\ell}{\nn}$},%
    \el/{$\ell$},%
    \elD/{$\ell\!+\!\delta$},%
    \rNE/{$\kk\!+\!\frac{\ell}{\nn}$},%
    \rCO/{$\kk\!+\!\ell$},%
    \rSP/{$\kk\!+\!\ell\!+\!\delta$}} {
      \draw (\x,0.10) -- (\x,-0.10);
      \node[above] at (\x,0.12) {\small\lab};
  }

  \draw[<->] (\lN,-2.50) to node[below] {\small width$=\kk$} (\rNE,-2.50);
\end{tikzpicture}
\caption{Interior regions for $\alpha^{NE}$, $\alpha^{CO}$, and $\alpha^{SP}(\mu)$ as functions of the cost saving~$\ss$, where $\delta = \mu\ell/[\lambda(1-\mu)]$.
Each shaded bar spans the values of~$\ss$ for which the corresponding automation rate is strictly interior.
All three intervals have the same width~$\kk$; they are shifted rightward by the progressively larger demand-loss terms that each objective internalizes.}
\label{fig:boundary}
\end{figure}
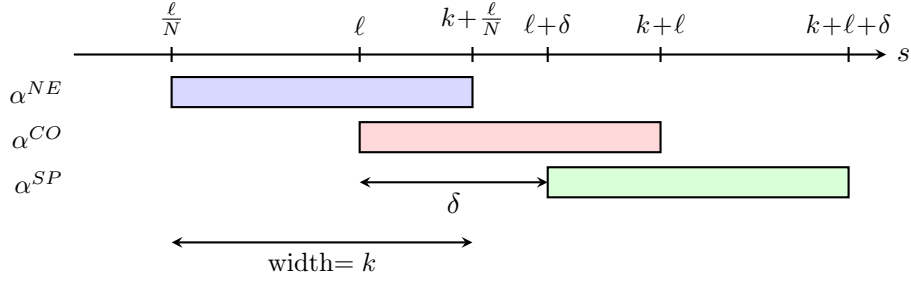

\begin{proof}[\textbf{Proof of \Cref{lem:boundary}}]
Each claim follows by checking when the raw formula lies strictly in $(0,1)$.

\emph{(i) Nash equilibrium.}
From \Cref{prop:alphastar}(i), $\alpha^{NE} = (\ss - \ell/\nn)/\kk$.
The raw formula $(\ss - \ell/\nn)/\kk$ is strictly positive iff $\ss > \ell/\nn$ (i.e., when $\nn > \nn^{*}$), and strictly less than one iff $\ss < \kk + \ell/\nn$.

\emph{(ii) Cooperative optimum.}
From \Cref{prop:alphastar}\ref{prop:CO}, $\alpha^{CO} = (\ss - \ell)/\kk$.
The raw expression $(\ss - \ell)/\kk$ is strictly positive iff $\ss > \ell$, and strictly less than one iff $\ss < \kk + \ell$.

\emph{(iii) Social planner.}
From \Cref{prop:surplus}(i), $\alpha^{SP}(\mu) = (\ss - \ell)/\kk - \mu\ell/[\lambda(1-\mu)\kk]$.
Collecting terms, $\alpha^{SP}(\mu) = \bigl(\ss - \ell - \delta(\mu)\bigr)/\kk$,
which is strictly positive iff $\ss > \ell + \delta(\mu)$ and strictly less than one iff $\ss < \kk + \ell + \delta(\mu)$.

Finally, $\ell/\nn < \ell$ holds because $\nn \geq 2$ and $\ell > 0$, and $\ell < \ell + \delta(\mu)$ for every $\mu \in (0,1)$, since then $\delta(\mu) = \mu\ell/[\lambda(1-\mu)] > 0$ (at $\eta = 1$ the demand loss vanishes, $\ell = 0$, so $\delta(\mu) = 0$ and the windows coincide; $\alpha^{SP}(\mu)$ is defined only for $\mu < 1$, where $\delta(\mu)$ is finite).
\end{proof}

\begin{proof}[\textbf{Proof of \Cref{cor:frictionless}} (Frictionless limit, corollary of \Cref{prop:alphastar})]
~\\
\begin{enumerate}[label=(\roman*),nosep]
\item With $\kk = 0$, firm~$i$'s profit from~\eqref{eq:profit_expanded} becomes
\[
    \pi_i = \Piz + \LL \inbrak{ \alpha_i \inparen{ \ss - \frac{\ell}{\nn}} - \frac{\ell}{\nn} \sum_{j \ne i} \alpha_j}.
\]
This is affine in~$\alpha_i$, with slope $\LL(\ss - \ell/\nn)$.
Crucially, the slope does not depend on rivals' choices $\alpha_{-i}$: rival automation affects the level of firm~$i$'s profit (through the last term) but not the marginal return to~$\alpha_i$.
The optimal $\alpha_i$ therefore does not depend on $\alpha_{-i}$, making it a dominant strategy.

When $\nn > \nn^{*}$ (equivalently $\ss > \ell/\nn$), the slope is strictly positive.
Since $\pi_i$ is linear and increasing in~$\alpha_i$ on $[0,1]$, the unique optimum is the upper boundary $\alpha_i^{*} = 1$.
This holds for every firm simultaneously, so $\alpha_i = 1$ for all~$i$ is the unique Nash equilibrium.

\item Assume additionally that the cost saving is less than the demand loss: $\ss < \ell$.
With $\kk = 0$, aggregate profit from~\eqref{eq:profit_clean} is linear in~$\abar$:
$\sum_i \pi_i = \nn\Piz + \nn\LL(\ss - \ell)\abar$.
Since $\ss - \ell < 0$, this is strictly decreasing in~$\abar$, so the profit-maximizing cooperative outcome is $\abar = 0$ (no automation), yielding per-firm profit~$\Piz$.

At the Nash equilibrium ($\abar = 1$), per-firm profit is $\Piz + \LL(\ss - \ell)$.
Since $\ss - \ell < 0$, we have $\Piz + \LL(\ss - \ell) < \Piz$: every firm earns strictly less than under cooperation.
The per-firm profit loss is $\Piz - [\Piz + \LL(\ss - \ell)] = \LL(\ell - \ss)$.
Across $\nn$ firms, the aggregate owner-surplus loss is $\nn\LL(\ell - \ss)$.

For the demand loss, from \cref{eq:demand}:
\[
  D = A + \lambda \ww\LL\nn - \ell \LL\nn\,\abar.
\]
At $\abar = 0$: $D^{CO} = A + \lambda \ww\LL\nn$.
At $\abar = 1$: $D^{NE} = A + \lambda \ww\LL\nn - \ell \LL\nn$.
Hence $D^{CO} - D^{NE} = \ell \LL\nn$.
\end{enumerate}
\end{proof}

\begin{proof}[\textbf{Proof of \Cref{prop:surplus}}]
~\\
\begin{enumerate}[label=(\roman*),nosep]
\item The $\mu$-planner chooses $\abar$ to maximize $S(\mu) = \mu\,\mathcal{W} + (1-\mu)\,\mathcal{K}$ from~\eqref{eq:S}.
    We compute each derivative in turn.
    Recall that $\mathcal{W} = \ww\LL\nn[1-(1-\eta)\abar]$, so
    \begin{equation}\label{eq:dWdabar}
      \frac{d\mathcal{W}}{d\abar} = -\ww\nn\LL(1-\eta).
    \end{equation}
    Since $\ell = \lambda(1-\eta)\ww$, this equals $-\ell\,\nn\LL/\lambda$.
    Recall that $\mathcal{K} = \nn[\Piz + \LL((\ss-\ell)\abar - \tfrac{\kk}{2}\abar^{2})]$ (e.g. \cref{eq:pitotsym}), so
    \begin{equation}\label{eq:dKdabar}
      \frac{d\mathcal{K}}{d\abar} = \nn\LL\bigl[(\ss-\ell) - \kk\abar\bigr].
    \end{equation}
    Combining:
    \[
      \frac{dS}{d\abar}
        = \mu\,\frac{d\mathcal{W}}{d\abar} + (1-\mu)\,\frac{d\mathcal{K}}{d\abar}
        = -\frac{\mu\,\ell\,\nn\LL}{\lambda}
          + (1-\mu)\,\nn\LL\bigl[(\ss-\ell) - \kk\abar\bigr].
    \]
    Setting this to zero and dividing by $\nn\LL > 0$:
    \[
      (1-\mu)\bigl[(\ss-\ell) - \kk\abar\bigr] = \frac{\mu\,\ell}{\lambda}.
    \]
    Dividing both sides by $(1-\mu) > 0$ and isolating~$\kk\abar$:
    \[
      (\ss-\ell) - \kk\abar = \frac{\mu\,\ell}{\lambda(1-\mu)},
      \qquad\Longrightarrow\qquad
      \kk\abar = (\ss-\ell) - \frac{\mu\,\ell}{\lambda(1-\mu)}.
    \]
    Dividing by~$\kk$, we find
    \begin{equation}\label{eq:alphasp1}
      \alpha^{SP}(\mu) = \frac{\ss-\ell}{\kk} - \frac{\mu\,\ell}{\lambda\cdot(1-\mu)\cdot \kk}
    \end{equation}
    This optimum will be valid as long as $0 \le \alpha^{SP} \le 1$.
    When $\ss>\ell$, define $\bar\mu$ to be the value where
    $\alpha^{SP}(\bar\mu) = 0$.  Then, solving for $\bar\mu$, we have
    \begin{align*}
      \ss-\ell &= \frac{ \bar\mu \ell}{\lambda \cdot (1-\bar\mu)} \\
              &\Downarrow \\
      \bar\mu &= \frac{ \lambda (\ss-\ell)}{\ell + \lambda(\ss-\ell)}
    \end{align*}
    From \cref{eq:alphasp1}, we have that $\alpha^{SP}$ is decreasing in $\mu$, so if
    $\mu > \bar\mu$, then \cref{eq:alphasp1} will be negative.

    When $\ss>\ell+\kk$, setting $\underline{\mu}$ to be the value where $\alpha^{SP}(\underline{\mu}) = 1$, we have
    \begin{align*}
      \ss-\ell - \kk &= \frac{ \underline\mu \ell}{\lambda \cdot (1-\underline\mu)} \\
              &\Downarrow \\
      \underline\mu &= \frac{ \lambda (\ss-\ell-\kk)}{\ell + \lambda(\ss-\ell-\kk)}
    \end{align*}

    For $\mu\leq\underline\mu$, the raw rate is at least one. If $\ss\leq\ell$, the constrained optimum is zero for every $\mu<1$; if $\ell<\ss<\ell+\kk$, it starts interior at $\mu=0$ and reaches zero at $\bar\mu$. At $\mu=1$, worker income is maximized by zero automation.

\item 
    Recall 
    \begin{align*}
      S &= \mu\mathcal{W}+(1-\mu)\mathcal{K} \\
        &= \nn\!\left[\mu\ww\LL[1-(1-\eta)\abar]
           +(1-\mu)\left[\Piz+\LL(\ss-\ell)\abar-\tfrac{\kk\LL}{2}\abar^2\right]\right] \\
        &= \underbrace{\nn[\mu\ww\LL+(1-\mu)\Piz]}_{a}
           +\underbrace{\nn\LL\left[-\mu\frac{\ell}{\lambda}+(1-\mu)(\ss-\ell)\right]}_{b}\abar \\
        &\quad +\underbrace{\left[-\frac{(1-\mu)\nn\LL\kk}{2}\right]}_{\tfrac12\gamma}\abar^2.
    \end{align*}
    where in the linear term we used $\ww(1-\eta) = \ell/\lambda$ from~\eqref{eq:ell}.
    This is a quadratic $S = a + b\,\abar + \tfrac{1}{2}\gamma\,\abar^{2}$, with
    $\gamma = -(1-\mu)\,\nn\LL\kk < 0$
    strictly negative.
    Completing the square around the maximum $x^{*} = -b/\gamma$:
    \[
      f(x) = a + bx + \tfrac{1}{2}\gamma\,x^{2}
        = f(x^{*}) + \tfrac{1}{2}\gamma\,(x - x^{*})^{2},
    \]
    and therefore
    \[
      f(x^{*}) - f(x) = -\tfrac{1}{2}\gamma\,(x - x^{*})^{2}
        = \tfrac{1}{2}|\gamma|\,(x - x^{*})^{2},
    \]
    where the last equality uses $\gamma \le 0$.
    Applying this with $x^{*} = \alpha^{SP}(\mu)$ (the planner's optimum, which coincides with the unconstrained vertex when it lies in $[0,1]$), $x = \alpha^{NE}$, and $|\gamma| = (1-\mu)\,\nn\LL\kk$:
    \[
      S(\alpha^{SP}) - S(\alpha^{NE})
        = \frac{(1-\mu)\,\nn\LL\kk}{2}\,
          [\alpha^{NE} - \alpha^{SP}(\mu)]^{2}.
    \]

\item 
    Recall from \eqref{eq:wedge_decomp}, the wedge decomposes as:
    \[
      \alpha^{NE} - \alpha^{SP}(\mu)
        = \frac{\ell(1-1/\nn)}{\kk} + \frac{\mu\,\ell}{\lambda(1-\mu)\,\kk}.
    \]
    The first term, $\ell(1-1/\nn)/\kk$, is strictly positive because $\ell > 0$ and $\nn \geq 2$.
    The second term, $\mu\,\ell/[\lambda(1-\mu)\,\kk]$, is strictly positive for any $\mu \in (0,1)$ and equals zero at $\mu = 0$.
    These identities compare the untruncated rates; clipping them to $[0,1]$ preserves the weak ordering. Since $\nn>\nn^*$, the Nash rate is positive, and it strictly exceeds the planner's rate unless both equal one.

    For worker income: from~\cref{eq:dWdabar}, $\mathcal{W} = \ww\LL\nn[1-(1-\eta)\abar]$ is affine and strictly decreasing in~$\abar$ (since $(1-\eta) > 0$ under the maintained assumption $\eta < 1$).
    Therefore $\alpha^{NE} > \alpha^{CO}$ implies $\mathcal{W}(\alpha^{NE}) < \mathcal{W}(\alpha^{CO})$.

    For owner surplus: from~\cref{eq:dKdabar}, $d\mathcal{K}/d\abar = \nn\LL[(\ss-\ell) - \kk\abar]$, which is strictly negative for $\abar>\alpha^{CO}$, including when the constrained cooperative optimum is zero.
    Hence $\mathcal{K}$ is strictly decreasing on $(\alpha^{CO},\,\alpha^{NE}]$, so $\mathcal{K}(\alpha^{NE}) < \mathcal{K}(\alpha^{CO})$.

    Provided $\ss < \kk + \ell$, so that $\alpha^{CO} < 1$ and the interval $(\alpha^{CO},\,\alpha^{NE}]$ is nonempty, both $\mathcal{W}$ and $\mathcal{K}$ are strictly lower at the Nash equilibrium than at the cooperative optimum, so neither class can be made better off by moving from $\alpha^{CO}$ to $\alpha^{NE}$: the Nash equilibrium is Pareto dominated. At $\ss \geq \kk + \ell$ both rates equal~$1$ and the two profiles coincide, so the dominance is weak.
\end{enumerate}
\end{proof}

\begin{proof}[\textbf{Proof of \Cref{cor:eta}} (Sign of the externality)]
The wedge formula $\alpha^{NE} - \alpha^{CO} = \ell\,(1-1/\nn)/\kk$ follows from the same computation as \Cref{prop:alphastar}\ref{prop:wedge_N}, now with $\ell = \lambda(1-\eta)\ww$ of either sign; it holds whenever both rates are interior (\Cref{lem:boundary}).
The three sign cases follow from the sign of~$\ell$: $\eta < 1$ gives $\ell > 0$ (positive wedge), $\eta = 1$ gives $\ell = 0$ (zero wedge, and both rates reduce to $\ss/\kk$), and $\eta > 1$ gives $\ell < 0$ (negative wedge).
On the interior regime the wedge is strictly decreasing in~$\eta$ because $\ell$ is. When $\ss\leq\ell$ the cooperative rate is zero and the gap equals the Nash rate (\Cref{prop:alphastar}), which rises with~$\eta$; over all~$\eta$ the gap therefore peaks at $\ell=\ss$ when $\lambda\ww>\ss$ and at $\eta=0$ otherwise.
\end{proof}

\begin{proof}[\textbf{Proof of \Cref{prop:equity}} (Worker equity)]
With profit-sharing~$\ee$, aggregate demand satisfies
\[
  D = A + \lambda\bigl[\underbrace{\ww\LL\nn(1-(1-\eta)\abar)}_{\text{wage income}} + \;\ee\,\Pi\bigr],
\]
where total profit is $\Pi=D-\nn\LL(\ww-\ss\abar)-\tfrac{\kk}{2}\LL\sum_j\alpha_j^2$. At a symmetric profile this becomes $D-\nn\LL(\ww-\ss\abar+\tfrac{\kk}{2}\abar^2)$. We maintain that operating profit is nonnegative at the profiles considered; every term in~$D$ is then nonnegative, so $D>0$ and the recycled profit is spending rather than a loss.
Substituting the expression for~$\Pi$ into the demand equation and collecting~$D$ terms on the left-hand side:
\[
  D(1 - \lambda\ee) = A + \lambda\,\ww\LL\nn(1-(1-\eta)\abar) - \lambda\ee\,\nn\LL\!\left(\ww - \ss\abar + \tfrac{\kk}{2}\abar^2\right).
\]
Expanding the right-hand side and grouping by powers of~$\abar$:
\begin{equation}\label{eq:D_equity}
  D = \frac{A + \lambda(1-\ee)\ww\LL\nn - \ell_\ee\LL\nn\abar - \lambda\ee\tfrac{\kk}{2}\nn\LL\abar^2}{1-\lambda\ee},
\end{equation}
where $\ell_\ee \coloneqq \ell - \lambda\ee\ss = \lambda(1-\eta)\ww - \lambda\ee(\ww - \cc)$.

\emph{(i) Cooperative optimum.}
The planner maximizes total profit $\Pi = D - \sum_i C_i$ by choosing the common automation rate~$\alpha$ at a symmetric profile.
Differentiating~$D$ (from the expression above) with respect to~$\alpha$:
\[
  \frac{\partial D}{\partial \alpha} = \frac{-\ell_\ee\,\LL\nn - \lambda\ee\kk\nn\LL\alpha}{1-\lambda\ee}.
\]
Differentiating total cost $\sum_i C_i = \nn\LL(\ww - \ss\alpha + \tfrac{\kk}{2}\alpha^2)$:
\[
  \frac{\partial(\sum_i C_i)}{\partial \alpha} = \nn\LL(-\ss + \kk\alpha).
\]
Setting $\partial\Pi/\partial\alpha = \partial D/\partial\alpha - \partial(\sum_i C_i)/\partial\alpha = 0$ and rearranging:
\[
  \nn\LL(\ss - \kk\alpha) = \frac{\ell_\ee\nn\LL + \lambda\ee\kk\nn\LL\alpha}{1-\lambda\ee}.
\]
Multiplying both sides by $(1-\lambda\ee)$:
\[
  (\ss - \kk\alpha)(1-\lambda\ee) = \ell_\ee + \lambda\ee\kk\alpha.
\]
Expanding the left-hand side and collecting~$\alpha$ terms:
\[
  \kk\alpha\underbrace{\bigl[(1-\lambda\ee) + \lambda\ee\bigr]}_{=\,1} = \ss(1-\lambda\ee) - \ell_\ee.
\]
Substituting $\ell_\ee = \ell - \lambda\ee\ss$, the right-hand side becomes $\ss - \ss\lambda\ee - \ell + \lambda\ee\ss = \ss - \ell$.
Hence $\kk\alpha = \ss - \ell$, giving $\alpha^{CO}(\ee) = (\ss - \ell)/\kk$, independent of~$\ee$.

\emph{(ii) Nash equilibrium.}
Competitive pricing gives $\Rev_i = D/\nn$.
The owner maximizes $(1-\ee)\pi_i$. For $\ee \in [0,1)$ we have $(1-\ee)>0$, so the FOC reduces to $\partial\pi_i/\partial\alpha_i=0$. At $\ee=1$, retained profit is identically zero and the owner is indifferent over automation rates; the formula describes the limit as $\ee\to1$, not a unique endpoint equilibrium. If also $\lambda=1$, the demand fixed point is singular.
Differentiating the demand fixed point with the general cost term $\tfrac{\kk}{2}\LL\sum_j\alpha_j^2$, then imposing symmetry, yields
\[
  \frac{\partial\Rev_i}{\partial\alpha_i} = \frac{1}{\nn}\cdot\frac{\partial D}{\partial\alpha_i}
  = \frac{-\ell_\ee\LL - \lambda\ee\kk\LL\alpha}{\nn(1-\lambda\ee)}.
\]
Combining with the cost derivative $\partial C_i/\partial\alpha_i = \LL(-\ss + \kk\alpha)$, the symmetric FOC $\partial\Rev_i/\partial\alpha_i - \partial C_i/\partial\alpha_i = 0$ becomes:
\[
  \ss - \kk\alpha - \frac{\ell_\ee + \lambda\ee\kk\alpha}{\nn(1-\lambda\ee)} = 0.
\]
Define $\EE \coloneqq \nn(1-\lambda\ee) + \lambda\ee = \nn - \lambda\ee(\nn-1)$.
Collecting the $\alpha$ terms on the left-hand side:
\[
  \kk\alpha\underbrace{\left[\frac{\nn(1-\lambda\ee) + \lambda\ee}{\nn(1-\lambda\ee)}\right]}_{=\,\EE/[\nn(1-\lambda\ee)]}
  = \ss - \frac{\ell_\ee}{\nn(1-\lambda\ee)}.
\]
Multiplying both sides by $\nn(1-\lambda\ee)/\EE$:
\[
  \kk\alpha = \frac{\ss\,\nn(1-\lambda\ee) - \ell_\ee}{\EE}.
\]
Substituting $\ell_\ee = \ell - \lambda\ee\ss$ into the numerator:
\[
  \ss\,\nn(1-\lambda\ee) - (\ell - \lambda\ee\ss)
  = \ss\bigl[\nn(1-\lambda\ee) + \lambda\ee\bigr] - \ell
  = \ss\,\EE - \ell.
\]
Hence $\kk\alpha = (\ss\EE - \ell)/\EE = \ss - \ell/\EE$, giving $\alpha^{NE}(\ee) = (\ss - \ell/\EE)/\kk$.

\emph{(iii) Wedge.}
Subtracting the cooperative optimum from the Nash rate:
\[
  \alpha^{NE}(\ee) - \alpha^{CO}
  = \frac{\ss - \ell/\EE}{\kk} - \frac{\ss - \ell}{\kk}
  = \frac{\ell - \ell/\EE}{\kk}
  = \frac{\ell(\EE - 1)}{\kk\,\EE}.
\]
Since $\EE = \nn - \lambda\ee(\nn-1)$, we have $\EE - 1 = (\nn - 1) - \lambda\ee(\nn-1) = (\nn - 1)(1 - \lambda\ee)$, so the wedge becomes
\[
  W(\ee) = \frac{\ell(\nn-1)(1-\lambda\ee)}{\kk\,\EE}.
\]
To show the wedge is strictly decreasing in~$\ee$, write $W = \ell(\nn-1)(1-\lambda\ee)/[\kk(\nn - \lambda\ee(\nn-1))]$ and apply the quotient rule.
The numerator is $f(\ee) = 1 - \lambda\ee$ with $f' = -\lambda$, and the denominator is $\tilde{g}(\ee) = \nn - \lambda\ee(\nn-1)$ with $\tilde{g}' = -\lambda(\nn-1)$.
By the quotient rule:
\[
  \frac{d}{d\ee}\!\left(\frac{f}{\tilde{g}}\right)
  = \frac{f'\tilde{g} - f\tilde{g}'}{\tilde{g}^2}
  = \frac{-\lambda[\nn - \lambda\ee(\nn-1)] + \lambda(\nn-1)(1-\lambda\ee)}{\EE^2}
  = \frac{-\lambda}{\EE^2} < 0.
\]
Hence $\partial W/\partial\ee = -\lambda\ell(\nn-1)/(\kk\EE^2) < 0$: the wedge is strictly decreasing.
Setting $W = 0$ requires $1 - \lambda\ee = 0$, i.e., $\ee = 1/\lambda$, which satisfies $\ee \geq 1$ whenever $\lambda \leq 1$.
\end{proof}

\begin{proof}[\textbf{Proof of \Cref{cor:equity_voluntary}} (Voluntary profit-sharing)]
Fix all automation choices and rivals' sharing rates. Let $\hat\pi_i>0$ denote profit at $\ee_i=0$, including rivals' demand feedback. Writing $M_i \coloneqq \nn - \lambda\sum_{j \neq i}\ee_j$, the own-sharing channel gives $\pi_i = \hat\pi_i + \lambda\ee_i\pi_i/M_i$, hence $\pi_i(\ee_i) = \hat\pi_i/(1-\lambda\ee_i/M_i)$ and
\[
  R_i(\ee_i) \coloneqq (1-\ee_i)\pi_i(\ee_i) = \frac{M_i\,\hat\pi_i\,(1-\ee_i)}{M_i-\lambda\ee_i},
  \qquad
  R_i'(\ee_i) = -\,\frac{M_i\,\hat\pi_i\,(M_i-\lambda)}{(M_i-\lambda\ee_i)^2}.
\]
Since $\sum_{j \neq i}\ee_j \leq \nn-1$, we have $M_i \geq \nn - \lambda(\nn-1) = \nn(1-\lambda) + \lambda$, so for $\lambda < 1$, $M_i - \lambda \geq \nn(1-\lambda) > 0$ and $M_i - \lambda\ee_i \geq M_i - \lambda > 0$; thus $R_i'(\ee_i) < 0$ for every $\ee_i \in [0,1]$ and every rival profile. Retained profit is therefore strictly decreasing on all of $[0,1]$, so $\ee_i = 0$ is the unique optimum regardless of the other firms' choices, a strictly dominant strategy. (Only at the knife-edge $\lambda = 1$ with all rivals fully sharing, where $M_i = 1$, does $R_i' \equiv 0$ and $\ee_i = 0$ becomes weakly optimal, the same $\lambda = 1$ degeneracy noted above.)
\end{proof}

\begin{proof}[\textbf{Proof of \Cref{cor:coase}}]
Partition the $\nn$ firms into a coalition $\mathcal{M}$ of size $M$ and a fringe of $\nn - M$ non-members.
For $M<\nn$, let
\begin{equation}
  \alpha^F \coloneqq  \frac{1}{\nn-M} \cdot \sum_{i \not \in \mathcal{M} } \alpha_i
\end{equation}
be the average automation rate of fringe firms; when $M=\nn$, omit the fringe term.
At a symmetric profile where all coalition members choose $\alpha^{M}$, the average automation rate is
\[
  \abar = \frac{M\alpha^{M} + (\nn - M)\alpha^{F}}{\nn}.
\]
From~\eqref{eq:profit_expanded}, a coalition member~$i$'s profit is
\[
  \pi_i = \Piz + \LL\!\left[\alpha_i\!\left(\ss - \frac{\ell}{\nn}\right) - \frac{\kk}{2}\alpha_i^2 - \frac{\ell}{\nn}\sum_{j \neq i}\alpha_j\right].
\]
Then
\begin{align*}
  \pi_i &= \Piz + \LL\!\left[\alpha_i\!\left(\ss - \frac{\ell}{\nn}\right) - \frac{\kk}{2}\alpha_i^2 - \frac{\ell}{\nn}\inparen{ \sum_{j \in \mathcal{M} \setminus \{i\} }\alpha_j + \sum_{j \notin \mathcal{M}} \alpha_j } \right] \\
  &= \Piz + \LL\!\left[\alpha_i\!\left(\ss - \frac{\ell}{\nn}\right) - \frac{\kk}{2}\alpha_i^2 - \frac{\ell}{\nn}\inparen{ \sum_{j \in \mathcal{M} \setminus \{i\} }\alpha_j + (\nn-M)\alpha^F } \right].
\end{align*}
Under a symmetric strategy, where all firms in the coalition follow the same strategy $\alpha_i = \alpha^M$, we have
\begin{align*}
  \pi_i &= \Piz + \LL\!\left[\alpha^M\!\left(\ss - \frac{\ell}{\nn}\right) - \frac{\kk}{2}\inparen{\alpha^M}^2 - \frac{\ell}{\nn}\inparen{ (M-1)\alpha^M + (\nn-M)\alpha^F } \right] \\
        &= \Piz + \LL\!\left[\alpha^M\!\left(\ss - \frac{\ell M}{\nn}\right) - \frac{\kk}{2}\inparen{\alpha^M}^2 - \frac{\ell}{\nn}(\nn-M)\alpha^F \right] \\
        &= \Piz - \LL\frac{\ell(\nn-M)\alpha^F}{\nn}+ \LL\!\left[\alpha^M\!\left(\ss - \frac{\ell M}{\nn}\right) - \frac{\kk}{2}\inparen{\alpha^M}^2 \right]
\end{align*}
The coalition chooses $\alpha^M$ to maximize $\sum_{i \in \mathcal{M}} \pi_i$.
At the symmetric equilibrium, $\sum_{ i \in \mathcal{M} } \pi_i = M \pi_i$, so maximizing the individual firm's profit in the coalition also maximizes the coalition profits.
Taking the derivative, we find
\[
  \frac{\partial \pi_i }{\partial \alpha^M} = \LL \inparen{ \ss - \frac{M\ell}{\nn} - \kk \alpha^M} = \LL \inparen{ \ss - \frac{\ell}{\nn} - \kk \alpha^M - \frac{(M-1)\ell}{\nn} }
\]
The additional term $-(M-1)\ell/\nn$ (relative to \cref{eq:FOCalpha}) arises because each coalition member internalizes the demand loss its automation imposes on the other $M-1$ members, each of whom loses $\ell\LL/\nn$ in revenue.
Setting the first-order condition to zero:
\[
  \alpha^{M} = \frac{\ss - M\ell/\nn}{\kk},
\]
clamped to $[0,1]$.
At $M = 1$, this reduces to $\alpha^{NE} = (\ss - \ell/\nn)/\kk$.
At $M = \nn$, it reduces to $\alpha^{CO} = (\ss - \ell)/\kk$.

When both $\alpha^M$ and $\alpha^{CO}$ are interior, the residual wedge is
\[
  \alpha^{M} - \alpha^{CO} = \frac{\ss - M\ell/\nn}{\kk} - \frac{\ss - \ell}{\kk} = \frac{\ell(1 - M/\nn)}{\kk},
\]
which is strictly positive for $M < \nn$ and zero only at $M = \nn$.
\end{proof}

\begin{proof}[\textbf{Proof of \Cref{prop:tax}} (Pigouvian automation tax)]
~\\
\begin{enumerate}[label=(\roman*),nosep]
\item Under the tax, firm~$i$'s profit from~\eqref{eq:profit_expanded} becomes
\[
  \pi_i = \Piz + \LL\!\left[\alpha_i\!\left(\ss - \tau - \frac{\ell}{\nn}\right) - \frac{\kk}{2}\alpha_i^2 - \frac{\ell}{\nn}\sum_{j \neq i}\alpha_j\right].
\]
The first-order condition is $\ss - \tau - \ell/\nn - \kk\alpha_i = 0$, giving $\alpha^{NE}(\tau) = (\ss - \tau - \ell/\nn)/\kk$.
Setting $\alpha^{NE}(\tau) = \alpha^{CO} = (\ss - \ell)/\kk$ yields $\tau^{*} = \ell - \ell/\nn = \ell(1-1/\nn)$. Relative to the untaxed equilibrium, a rate~$\tau$ with own-choice-independent rebates raises aggregate profit by $\nn\LL(2\tau^{*}\tau-\tau^{2})/(2\kk)$ when all rates stay interior, so any $\tau\in(0,2\tau^{*})$ improves on Nash.

\item At $\tau = \tau^{*}$, all firms choose $\alpha^{CO}$.
Per-firm profit is $\pi^{\mathrm{tax}} = \pi(\alpha^{CO}) - \tau^{*}\LL\alpha^{CO} = \pi^{CO} - \tau^{*}\LL\alpha^{CO}$.
Total tax revenue is $\tau^{*}\LL\nn\alpha^{CO}$. Rebated as an exogenous lump sum (each firm's receipt is fixed independently of its own automation choice, so it does not re-enter the first-order condition), each firm receives $\tau^{*}\LL\alpha^{CO}$, restoring profit to~$\pi^{CO}$ without altering the marginal incentive that pins down~$\tau^{*}$.
\end{enumerate}
\end{proof}

\begin{proof}[\textbf{Proof of \Cref{prop:phi}} (AI productivity widens the over-automation wedge)]
~\\
\begin{enumerate}[label=(\roman*),nosep]
\item With $\phi > 1$, a uniform market-clearing price allocates revenue by output share: $\Rev_i = D \cdot Y_i / \sum_j Y_j$, which reduces to $D/\nn$ only once the symmetric profile $\alpha_i = \alpha$ is imposed. The derivative below is taken on this general form, so that the dependence of the denominator $\sum_j Y_j$ on~$\alpha_i$ (the source of the market-share gain) is retained before symmetry is substituted.
The firm's first-order condition equates the marginal benefit of automation to its marginal cost: $\partial\Rev_i/\partial\alpha_i + \ss\LL - \kk\LL\alpha_i = 0$, where $\ss\LL$ is the per-unit cost saving from replacing a worker with AI and $\kk\LL\alpha_i$ is the marginal integration friction.
Substituting the marginal revenue from~\eqref{eq:rev_phi} and rearranging, the symmetric first-order condition is
\[
\kk\alpha = \ss - \frac{\ell}{\nn} + \frac{D(\alpha)\,(\phi-1)(\nn-1)}{\nn^2[1+(\phi-1)\alpha]\,\LL}.
\]
Define $\mathrm{LHS}(\alpha) = \kk\alpha$ and let $\mathrm{RHS}(\alpha)$ denote the right-hand side.
The left-hand side is strictly increasing (slope~$\kk$).
The right-hand side is strictly decreasing: the market-share term has numerator proportional to $D(\alpha) = A + \lambda \ww\LL\nn - \ell \LL\nn\alpha$ (decreasing in~$\alpha$) and denominator factor $1+(\phi-1)\alpha$ (increasing in~$\alpha$).
Under the maintained interiority condition, this gives a unique symmetric crossing. It is a best response: writing $q=\phi-1$, $B=\sum_{j\ne i}[1+q\alpha_j]$, and $T=1+q\alpha_i+B$, the unilateral revenue curvature is
\[
 \frac{\partial^2\Rev_i}{\partial\alpha_i^2}
 =-\frac{2qB(\ell\LL T+qD)}{T^3}\leq0.
\]
Subtracting the cost curvature $\kk\LL>0$ makes profit strictly concave.

To show $\alpha^{NE}(\phi)>\alpha^{NE}(1)$, first suppose the baseline rate is interior and evaluate both sides at $\alpha = \alpha^{NE}(1) = (\ss - \ell/\nn)/\kk$:
$\mathrm{LHS} = \ss - \ell/\nn$, while $\mathrm{RHS} = \ss - \ell/\nn + {}$(positive market-share term)${} > \mathrm{LHS}$.
Since $\mathrm{LHS}$ is increasing and $\mathrm{RHS}$ is decreasing, the unique crossing must occur at $\alpha^{NE}(\phi) > \alpha^{NE}(1)$. If the baseline rate is zero, the maintained interiority of $\alpha^{NE}(\phi)$ gives the same strict ordering. A baseline rate of one is incompatible with an interior equilibrium at higher productivity.

\item The cooperative planner maximizes $\sum_i \pi_i = D - \sum_i C_i$.
Total revenue equals aggregate demand~$D$ regardless of how output is allocated across firms: expenditure $D$ is pinned down by~\eqref{eq:demand}, so reallocating output shares does not change total revenue.
The planner's first-order condition therefore depends only on costs:
\[
\frac{\partial(\sum_i\pi_i)}{\partial\alpha_i} = -\ell \LL + \ss\LL - \kk\LL\alpha_i,
\]
which is independent of~$\phi$.
Imposing the bounds gives $\alpha^{CO}(\phi) = \min\{1,\max\{0,(\ss-\ell)/\kk\}\}=\alpha^{CO}(1)$.

For the generalized planner: $S(\mu) = \mu\,\mathcal{W} + (1-\mu)\,\mathcal{K}$.
Worker income $\mathcal{W} = \ww\LL\nn[1-(1-\eta)\abar]$ does not depend on~$\phi$.
Owner surplus $\mathcal{K} = D - \sum_i C_i$ at symmetric profiles, and neither $D$~\eqref{eq:demand} nor $C_i$~\eqref{eq:cost} depends on~$\phi$.
Hence $S(\mu)$ is $\phi$-invariant at every symmetric~$\abar$, and $\alpha^{SP}(\mu;\phi) = \alpha^{SP}(\mu;1)$ for all~$\mu$.

\item Combining (i) and (ii): $\alpha^{NE}(\phi) > \alpha^{NE}(1)$ while $\alpha^{SP}(\mu;\phi) = \alpha^{SP}(\mu;1)$ for all~$\mu$, so the wedge $\alpha^{NE}(\phi) - \alpha^{SP}(\mu;\phi)$ is strictly larger than $\alpha^{NE}(1) - \alpha^{SP}(\mu;1)$ for every~$\mu$.
Since $\alpha^{NE}(\phi)$ is increasing in~$\phi$ (by the same LHS/RHS argument with a larger market-share term), the wedge is strictly increasing in~$\phi$.
\end{enumerate}
\end{proof}

\begin{proof}[\textbf{Proof of \Cref{prop:entry}} (Endogenous entry, frictionless benchmark)]
We assume $\kk = 0$ (frictionless case), $\lambda = 1$ (full recycling), $0 < \kappa < A$ (entry is costly but the market is viable for at least one firm), $\ell > \ss$ (so that $\nn^{*} > 1$ by \Cref{cor:frictionless}), and the genericity condition $\nn^{*} \notin \mathbb{N}$ (so no integer firm count sits exactly at the indifference cutoff and every integer $\nn$ lies strictly on one branch of the profit schedule).

\medskip
\emph{Step 1: the profit schedule is strictly decreasing on~$\mathbb{N}$.}
Let $m = \lfloor \nn^{*} \rfloor$, so $m \leq \nn^{*} < m+1$.
For $\nn \leq \nn^{*}$, \cref{cor:frictionless} gives $\alpha = 0$, so by~\eqref{eq:profit_clean} with $\lambda = 1$: $\Pi^{*}(\nn) = A/\nn$, which is strictly decreasing.
For $\nn > \nn^{*}$, full automation is dominant (\cref{cor:frictionless}), and per-firm profit drops by $\Delta = \LL(\ell - \ss) > 0$, giving $\Pi^{*}(\nn) = A/\nn - \Delta$, also strictly decreasing.
At the crossing: $\Pi^{*}(m) = A/m > A/(m+1) > A/(m+1) - \Delta = \Pi^{*}(m+1)$.
Hence $\Pi^{*}$ is strictly decreasing on~$\mathbb{N}$.

\medskip
\emph{Step 2: existence and uniqueness of~$\nn^{FE}$.}
Since $\kappa < A$, $\Pi^{*}(1) = A > \kappa$.
Since $\Delta > 0$, $\Pi^{*}(\nn) \to -\Delta < 0$ as $\nn \to \infty$.
The set $\mathcal{S} = \{\nn \in \mathbb{N} : \Pi^{*}(\nn) \geq \kappa\}$ is therefore nonempty and finite.
Let $\nn^{FE} = \max \mathcal{S}$.
By strict monotonicity, $\Pi^{*}(\nn^{FE}) \geq \kappa$ and $\Pi^{*}(\nn^{FE}+1) < \kappa$, so $\nn^{FE}$ satisfies~\eqref{eq:free_entry} and is unique.

\medskip
\emph{Step 3: characterization by regime.}
We determine $\nn^{FE}$ by checking which integers are in~$\mathcal{S}$.

On the no-automation branch ($\nn \leq m$): $\Pi^{*}(\nn) = A/\nn \geq \kappa$ iff $\nn \leq A/\kappa$.
So the viable integers on this branch are $\{1, \dots, \min(\lfloor A/\kappa \rfloor, m)\}$.

On the full-automation branch ($\nn \geq m+1$): $\Pi^{*}(\nn) = A/\nn - \Delta \geq \kappa$ iff $\nn \leq A/(\kappa + \Delta)$.
The smallest integer on this branch is $m+1$, so viable integers exist iff $m + 1 \leq A/(\kappa + \Delta)$, i.e.\ $\kappa + \Delta \leq A/(m+1)$.
When nonempty, the viable set is $\{m+1, \dots, \lfloor A/(\kappa + \Delta) \rfloor\}$.

Since $\nn^{FE}$ is the largest viable integer overall:

\emph{Case~(i): (low entry cost) $\kappa + \Delta \leq A/(m+1)$.}
The full-automation branch contains viable integers up to $\lfloor A/(\kappa + \Delta) \rfloor \geq m+1$.
Since $\kappa + \Delta > \kappa$, we have $A/(\kappa+\Delta) < A/\kappa$, so the full-automation branch's maximum does not exceed the no-automation branch's maximum in absolute terms, but it exceeds~$m$ (the cap on the no-automation branch).
Hence $\nn^{FE} = \lfloor A/(\kappa + \Delta) \rfloor \geq m + 1 > \nn^{*}$, and every firm fully automates.
See \cref{fig:entry_profit}(a).

\emph{Case~(ii): (intermediate entry cost) $\kappa + \Delta > A/(m+1)$ and $\kappa < A/m$.}
The first condition means $A/(m+1) - \Delta < \kappa$, so $\Pi^{*}(m+1) = A/(m+1) - \Delta < \kappa$.
Since $\Pi^{*}$ is decreasing, no integer above~$\nn^{*}$ is viable.
The second condition gives $A/\kappa > m$, so $\lfloor A/\kappa \rfloor \geq m$.
Hence the largest viable integer on the no-automation branch is~$m$, and $\nn^{FE} = m$.
Profit is $\Pi^{*}(m) = A/m > \kappa$ (strict, since $A/\kappa > m$ implies $A/m > \kappa$).
In this case since $\nn^{FE} = m < \nn^*$ (strict, since $\nn^{*} \notin \mathbb{N}$), no firm automates.
See \cref{fig:entry_profit}(b).

\emph{Case~(iii): (high entry cost) $\kappa \geq A/m$.}
Then $A/\kappa \leq m$, so $\lfloor A/\kappa \rfloor \leq m$.
Since $\kappa + \Delta > \kappa \geq A/m > A/(m+1)$, no integer on the full-automation branch is viable (same argument as case~(ii)).
Hence $\nn^{FE} = \lfloor A/\kappa \rfloor$.
Since $\nn^{FE} = \lfloor A/\kappa \rfloor \leq m \leq \nn^{*}$, no firm automates.
See \cref{fig:entry_profit}(c).

\medskip
These three cases exhaust all $\kappa \in (0, A)$.
When $\kappa > A$, then $\nn^{FE} = 0$.
\end{proof}

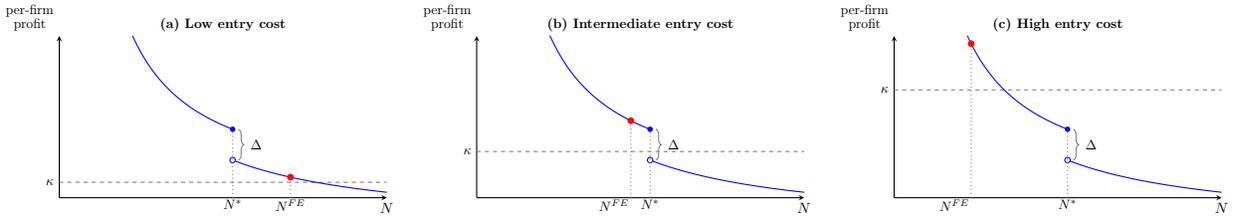
\begin{figure}[ht]
\centering
\pgfmathsetmacro{\A}{8}%
\pgfmathsetmacro{\Del}{0.8}%
\pgfmathsetmacro{\Nstar}{4.5}%
\pgfmathsetmacro{\mInt}{4}%
\pgfmathsetmacro{\xmax}{8.5}%
\pgfmathsetmacro{\ymax}{4.2}

\resizebox{0.32\textwidth}{!}{%
\begin{tikzpicture}[>=stealth, thick]

  \pgfmathsetmacro{\kap}{0.4}
  \pgfmathsetmacro{\NFE}{6}%
  \pgfmathsetmacro{\piNFE}{\A/\NFE - \Del}

  \draw[->] (0,0) -- (\xmax,0) node[below] {$\nn$};
  \draw[->] (0,0) -- (0,\ymax)
    node[above left, align=center, font=\small] {per-firm\\[-2pt]profit};

  \draw[dashed, gray] (0,\kap) node[left, black, font=\small] {$\kappa$}
    -- (\xmax,\kap);

  \draw[blue, thick, domain=1.9:\Nstar, samples=80]
    plot (\x, {\A/\x});

  \draw[blue, thick, domain=\Nstar:\xmax, samples=80]
    plot (\x, {\A/\x - \Del});

  \pgfmathsetmacro{\piUpper}{\A/\Nstar}
  \fill[blue] (\Nstar, \piUpper) circle (2pt);
  \pgfmathsetmacro{\piLower}{\A/\Nstar - \Del}
  \draw[blue, fill=white] (\Nstar, \piLower) circle (2pt);

  \draw[decorate, decoration={brace, amplitude=4pt, mirror},
        thick, gray]
    (\Nstar+0.15, \piLower) -- (\Nstar+0.15, \piUpper)
    node[midway, right=5pt, black, font=\small] {$\Delta$};

  \draw[dotted, gray] (\Nstar,0) -- (\Nstar, \piUpper);
  \node[below, font=\small, inner sep=1pt] at (\Nstar, 0) {$\nn^{*}$};

  \fill[red] (\NFE, \piNFE) circle (2.5pt);
  \draw[dotted, gray] (\NFE, 0) -- (\NFE, \piNFE);
  \node[below, font=\small, inner sep=1pt] at (\NFE, 0) {$\nn^{FE}$};

  \node[font=\small\bfseries] at (0.5*\xmax, \ymax+0.3)
    {(a) Low entry cost};

\end{tikzpicture}%
}%
\hfill
\resizebox{0.32\textwidth}{!}{%
\begin{tikzpicture}[>=stealth, thick]

  \pgfmathsetmacro{\kap}{1.2}
  \pgfmathsetmacro{\NFE}{4}%
  \pgfmathsetmacro{\piNFE}{\A/\NFE}

  \draw[->] (0,0) -- (\xmax,0) node[below] {$\nn$};
  \draw[->] (0,0) -- (0,\ymax)
    node[above left, align=center, font=\small] {per-firm\\[-2pt]profit};

  \draw[dashed, gray] (0,\kap) node[left, black, font=\small] {$\kappa$}
    -- (\xmax,\kap);

  \draw[blue, thick, domain=1.9:\Nstar, samples=80]
    plot (\x, {\A/\x});

  \draw[blue, thick, domain=\Nstar:\xmax, samples=80]
    plot (\x, {\A/\x - \Del});

  \pgfmathsetmacro{\piUpper}{\A/\Nstar}
  \fill[blue] (\Nstar, \piUpper) circle (2pt);
  \pgfmathsetmacro{\piLower}{\A/\Nstar - \Del}
  \draw[blue, fill=white] (\Nstar, \piLower) circle (2pt);

  \draw[decorate, decoration={brace, amplitude=4pt, mirror},
        thick, gray]
    (\Nstar+0.15, \piLower) -- (\Nstar+0.15, \piUpper)
    node[midway, right=5pt, black, font=\small] {$\Delta$};

  \draw[dotted, gray] (\Nstar,0) -- (\Nstar, \piUpper);
  \node[below, font=\small, inner sep=1pt] at (\Nstar, 0) {$\nn^{*}$};

  \fill[red] (\NFE, \piNFE) circle (2.5pt);
  \draw[dotted, gray] (\NFE, 0) -- (\NFE, \piNFE);
  \node[below left, font=\small, inner sep=1pt] at (\NFE, 0) {$\nn^{FE}$};

  \node[font=\small\bfseries] at (0.5*\xmax, \ymax+0.3)
    {(b) Intermediate entry cost};

\end{tikzpicture}%
}%
\hfill
\resizebox{0.32\textwidth}{!}{%
\begin{tikzpicture}[>=stealth, thick]

  \pgfmathsetmacro{\kap}{2.8}
  \pgfmathsetmacro{\NFE}{2}%
  \pgfmathsetmacro{\piNFE}{\A/\NFE}

  \draw[->] (0,0) -- (\xmax,0) node[below] {$\nn$};
  \draw[->] (0,0) -- (0,\ymax)
    node[above left, align=center, font=\small] {per-firm\\[-2pt]profit};

  \draw[dashed, gray] (0,\kap) node[left, black, font=\small] {$\kappa$}
    -- (\xmax,\kap);

  \draw[blue, thick, domain=1.9:\Nstar, samples=80]
    plot (\x, {\A/\x});

  \draw[blue, thick, domain=\Nstar:\xmax, samples=80]
    plot (\x, {\A/\x - \Del});

  \pgfmathsetmacro{\piUpper}{\A/\Nstar}
  \fill[blue] (\Nstar, \piUpper) circle (2pt);
  \pgfmathsetmacro{\piLower}{\A/\Nstar - \Del}
  \draw[blue, fill=white] (\Nstar, \piLower) circle (2pt);

  \draw[decorate, decoration={brace, amplitude=4pt, mirror},
        thick, gray]
    (\Nstar+0.15, \piLower) -- (\Nstar+0.15, \piUpper)
    node[midway, right=5pt, black, font=\small] {$\Delta$};

  \draw[dotted, gray] (\Nstar,0) -- (\Nstar, \piUpper);
  \node[below, font=\small, inner sep=1pt] at (\Nstar, 0) {$\nn^{*}$};

  \fill[red] (\NFE, \piNFE) circle (2.5pt);
  \draw[dotted, gray] (\NFE, 0) -- (\NFE, \piNFE);
  \node[below left, font=\small, inner sep=1pt] at (\NFE, 0) {$\nn^{FE}$};

  \node[font=\small\bfseries] at (0.5*\xmax, \ymax+0.3)
    {(c) High entry cost};

\end{tikzpicture}%
}
\caption{Per-firm profit as a function of the number of firms in the frictionless benchmark ($\kk=0$, $\lambda=1$).
In each panel, profit follows $A/\nn$ for $\nn \leq \nn^{*}$ (no automation) and drops discretely by $\Delta = \LL(\ell - \ss)$ at~$\nn^{*}$ when full automation becomes dominant.
The red dot marks the free-entry equilibrium~$\nn^{FE}$.
(a)~Low entry cost: $\nn^{FE}$ lies above~$\nn^{*}$ and all firms automate.
(b)~Intermediate entry cost: $\nn^{FE} = \lfloor\nn^{*}\rfloor$; the threat of automation deters the marginal entrant, sustaining positive profits without any automation occurring.
(c)~High entry cost: $\nn^{FE}$ falls well below~$\nn^{*}$; entry costs alone limit competition and the automation threshold is never approached.}
\label{fig:entry_profit}
\end{figure}

\begin{proof}[\textbf{Proof of \Cref{prop:entry_convex}}]

First, we show that the per-firm profit is decreasing in $\nn$.  From \cref{eq:profit_clean}, we can write
\begin{equation}
  \pi^{NE}(\nn) = A/\nn + C + g(\nn)
\end{equation}
where 
\begin{align}
  C &\coloneqq (\lambda-1)\ww\LL \\
  g(\nn) &\coloneqq \LL\,\alpha^{NE}(\nn)\inbrak{ (\ss - \ell) - \tfrac{\kk}{2}\alpha^{NE}(\nn) }
\end{align}
Since $A/\nn$ is strictly decreasing on~$\mathbb{N}$, it suffices to show that $g$ is weakly decreasing for $\nn \in \mathbb{N}$.
First, note that $g(\cdot)$ is decreasing in $\alpha^{NE}$
\begin{equation}
  \frac{d g}{d \alpha^{NE}} = \LL\,\inbrak{ (\ss-\ell) - \kk\, \alpha^{NE} }
\end{equation}
By assumption $\ell > \ss$, so the derivative is always negative.
Then, by \Cref{prop:alphastar}, we have 
\begin{equation}
  \alpha^{NE} = \left\{ \begin{array}{ll} 0 & \mbox{ if $\nn \le \frac{\ell}{\ss}$ } \\ \min\inparen{ \frac{\ss - \ell/\nn}{\kk},\, 1 } & \mbox{ if $\nn > \frac{\ell}{\ss}$} \end{array} \right.
\end{equation}
So $\alpha^{NE}$ is non-decreasing in $\nn$.
Thus $g(\alpha^{NE}(\nn))$ is non-increasing in $\nn$.

Since $A/\nn$ is decreasing in $\nn$, 
and $g\inparen{\alpha^{NE}(\nn)}$ is non-increasing in $\nn$, 
so $\pi^{NE}(\nn) = A/\nn + C + g(\nn)$ is strictly decreasing in $\nn$, 
i.e., $\pi^{NE}(\nn) > \pi^{NE}(\nn+1)$.

Now, let 
\begin{equation}
  S \coloneqq \{\nn \in \mathbb{N} : \pi^{NE}(\nn) \geq \kappa\}
\end{equation}
Since $\pi^{NE}(1) = \Piz(1) > \kappa$ by assumption, the set $S$ is non-empty.
To see that~$S$ is finite, we need to show that $\lim_{\nn \to \infty} \pi^{NE} < \kappa$.
First, 
\begin{align}
  \alpha_\infty \coloneqq \lim_{\nn \to \infty} \alpha^{NE}(\nn) = \min\inparen{ 1, \frac{\ss}{\kk} }
\end{align}
So 
\begin{equation}
  \lim_{\nn \to \infty} g(\nn) = \LL\,\alpha_{\infty}\inbrak{ (\ss - \ell) - \tfrac{\kk}{2}\alpha_\infty } \le 0
\end{equation}
because $\ss < \ell$ (by assumption), with strict inequality whenever $\ss > 0$ (at $\ss = 0$, $\alpha_\infty = 0$ and the limit is~$0$).
Thus
\begin{equation}
  \lim_{\nn \to \infty} \pi^{NE}(\nn) = \lim_{\nn \to \infty} \inbrak{ A/\nn + C + g(\nn) } = (\lambda-1)\ww \LL + \lim_{\nn \to \infty} g(\nn) \le 0
\end{equation}
Here $C = (\lambda -1)\ww\LL \le 0$ (equality only at $\lambda = 1$) and $\lim_{\nn \to \infty} g(\nn) \le 0$, so $\lim_{\nn \to \infty} \pi^{NE} \le 0 < \kappa$; since $\kappa > 0$ the inequality relative to~$\kappa$ is strict, and the set $S$ is finite.
Let 
\begin{equation}
  \nn^{FE} \coloneqq \max \inparen{ \nn \in S }
\end{equation}
Then $\pi^{NE}(\nn^{FE}) \geq \kappa$ and $\pi^{NE}(\nn^{FE}+1) < \kappa$, so $\nn^{FE}$ satisfies~\eqref{eq:free_entry}.
If $\nn^{FE} > \nn^{*}$, then $\alpha^{NE}(\nn^{FE}) = \min \inparen{ (\ss - \ell/\nn^{FE})/\kk, 1 } > 0$; since $\ell > \ss$ implies $\alpha^{CO} = \max\{0,\,(\ss - \ell)/\kk\} = 0$, we have $\alpha^{NE}(\nn^{FE}) > \alpha^{CO}$, so over-automation persists.
\end{proof}

\begin{proof}[\textbf{Proof of \Cref{prop:endo_wages}} (Endogenous wages)]
The symmetric equilibrium is a fixed point: $\abar$ such that $\abar = \alpha^{NE}(\ww(\abar))$.

(i)~At a symmetric rate $\alpha_i = \alpha$ for all~$i$, we have $\abar = \alpha$ and, from~\eqref{eq:profit_clean},
\[
  \pi_i \;=\; \Piz + \LL\!\left(\ss\,\alpha - \ell\,\alpha - \tfrac{\kk}{2}\alpha^2\right).
\]
Aggregate profit is $\sum_i \pi_i = \nn\,\pi_i$. Consistent with the wage-taking convention, we differentiate at fixed~$\ww$ and evaluate the resulting condition at the equilibrium wage $\ww(\alpha)$; differentiating with respect to~$\alpha$ and dividing by~$\nn\LL$ yields the planner's per-firm marginal benefit of automation:
\[
  g(\alpha) \coloneqq \ss(\alpha) - \ell(\alpha) - \kk\alpha
    = \ww(\alpha) - \cc - \lambda(1-\eta)\ww(\alpha) - \kk\alpha
    = \ww(\alpha)\bigl[1-\lambda(1-\eta)\bigr] - \cc - \kk\alpha,
\]
where we used $\ss = \ww - \cc$ and $\ell = \lambda(1-\eta)\ww$.
By contrast, each firm's profit from~\eqref{eq:profit_expanded} depends on~$\abar$ only through its own revenue share~$\ell/\nn$. Differentiating with respect to~$\alpha_i$ and dividing by~$\LL$ gives the private marginal benefit:
\[
  h(\alpha) \coloneqq \ss(\alpha) - \frac{\ell(\alpha)}{\nn} - \kk\alpha
    = \ww(\alpha)\!\left[1 - \frac{\lambda(1-\eta)}{\nn}\right] - \cc - \kk\alpha.
\]
Differentiating, $g'(\alpha) = \ww'(\alpha)[1-\lambda(1-\eta)] - \kk$; the first term is weakly negative and $\kk > 0$, so $g' < 0$.
The same argument applies to~$h$ (whose coefficient $1-\lambda(1-\eta)/\nn > 0$ for $\nn \geq 2$), so both are strictly decreasing.
Note that $g(\alpha) = h(\alpha) - \ell(\alpha)(1-1/\nn)$ for all~$\alpha$.
At the NE fixed point, $h(\alpha^{NE}) = 0$ by definition, so
\[
  g(\alpha^{NE}) = 0 - \ell(\alpha^{NE})(1-1/\nn) = -\ell(\alpha^{NE})(1-1/\nn) < 0.
\]
Since $g$ is strictly decreasing and $g(\alpha^{CO}) = 0$ (the planner's optimality condition), the inequality $g(\alpha^{NE}) < 0 = g(\alpha^{CO})$ implies $\alpha^{CO} < \alpha^{NE}$. This argument zeroes both first-order conditions ($h(\alpha^{NE}) = 0$ and $g(\alpha^{CO}) = 0$), so it presumes both rates interior. The remaining cases follow from the boundary signs. If $\alpha^{CO}=0<\alpha^{NE}$, the ordering is immediate. If $\alpha^{NE}=1$, then either $g(1)<0$ and $\alpha^{CO}<1$, or $g(1)\geq0$ and both rates equal one.

(ii)~From~\eqref{eq:Nstar}, the threshold is $\nn^{*}(\ww) = \lambda(1-\eta)\ww/(\ww-\cc)$.
Applying the quotient rule:
\[
  \frac{\partial \nn^{*}}{\partial \ww}
  = \frac{\lambda(1-\eta)(\ww - \cc) - \lambda(1-\eta)\ww}{(\ww-\cc)^2}
  = \frac{-\lambda(1-\eta)\cc}{(\ww-\cc)^2} < 0,
\]
so $\nn^{*}$ is strictly decreasing in~$\ww$ when $\cc > 0$ and $\eta < 1$ (and weakly otherwise: at $\cc = 0$ or $\eta = 1$ the derivative is zero and $\nn^{*}$ is wage-invariant).
Since $\ww'(\abar) \leq 0$ by assumption, $\ww(\abar) \leq \ww(0)$ for all $\abar \geq 0$.
Combining: $\nn^{*}(\ww(\abar)) \geq \nn^{*}(\ww(0))$, with strict inequality whenever $\ww(\abar) < \ww(0)$, $\cc > 0$, and $\eta < 1$.
\end{proof}

\begin{corollary}[Generalized planner under wage adjustment]\label{cor:endo_mu}
Under the conditions of \Cref{prop:endo_wages}, let $\mu \in [0,1)$ and define
$g_\mu(\alpha) \coloneqq \ss(\ww(\alpha)) - \ell(\ww(\alpha))\bigl[1 + \mu/(\lambda(1-\mu))\bigr] - \kk\alpha$.
If $g_\mu$ is strictly decreasing, then $\alpha^{NE}>\alpha^{SP}(\mu)$ whenever $0<\alpha^{SP}(\mu)<1$.
\end{corollary}

\begin{proof}[\textbf{Proof sketch}]
The argument mirrors part~(i) of \Cref{prop:endo_wages}, replacing the cooperative planner's marginal benefit~$g$ with the $\mu$-planner's.
The $\mu$-planner maximizes $S(\mu) = \mu\,\mathcal{W} + (1-\mu)\,\mathcal{K}$ over a common rate~$\alpha$, taking wages as given.
From the proof of \Cref{prop:surplus}, the first-order condition is
\[
  (1-\mu)\bigl[\ss(\alpha) - \ell(\alpha) - \kk\alpha\bigr] = \frac{\mu\,\ell(\alpha)}{\lambda},
\]
which can be rewritten as $g_\mu(\alpha) = 0$.
This function relates to the private marginal benefit $h(\alpha) = \ss(\alpha) - \ell(\alpha)/\nn - \kk\alpha$ from the proof of \Cref{prop:endo_wages} by
\[
  g_\mu(\alpha) = h(\alpha) - \ell(\alpha)\!\left(1 - \frac{1}{\nn} + \frac{\mu}{\lambda(1-\mu)}\right).
\]
If the Nash rate is interior, $h(\alpha^{NE})=0$, so
\[
  g_\mu(\alpha^{NE}) = -\ell(\alpha^{NE})\!\left(1 - \frac{1}{\nn} + \frac{\mu}{\lambda(1-\mu)}\right) < 0,
\]
since $\ell > 0$, $\nn \geq 2$, and $\mu > 0$.
If the Nash rate is one and the planner's rate is interior, their strict ordering is immediate. At the planner's fixed point, $g_\mu(\alpha^{SP}(\mu))=0$.
If $g_\mu$ is strictly decreasing, the inequality $g_\mu(\alpha^{NE}) < 0 = g_\mu(\alpha^{SP})$ implies $\alpha^{SP}(\mu) < \alpha^{NE}$, exactly as in the $\mu = 0$ case. As in \Cref{prop:endo_wages}, this zeroes $g_\mu$ and so holds on the interior; once $\ss(\abar) \geq \kk + \ell(\abar) + \delta(\mu)$ pins $\alpha^{SP}(\mu) = \alpha^{NE} = 1$ (\Cref{lem:boundary}), the gap closes and the strict inequality lapses.

It remains to verify monotonicity.
Differentiating:
\[
  g_\mu'(\alpha) = \ww'(\alpha)\!\left[1 - (1-\eta)\!\left(\lambda + \frac{\mu}{1-\mu}\right)\right] - \kk.
\]
Define the bracket as $C_\mu \coloneqq 1 - (1-\eta)[\lambda + \mu/(1-\mu)]$.
At $\mu = 0$, $C_0 = 1 - \lambda(1-\eta) \geq 0$, which is the coefficient used in the proof of \Cref{prop:endo_wages}; since $\ww' \leq 0$, the product $\ww' C_0 \leq 0$ and $g_0' < 0$ follows immediately.
For $\mu > 0$, $C_\mu$ decreases.
As long as $C_\mu \geq 0$, the same argument applies; this holds for all $\mu \leq \bar\mu \coloneqq [1-(1-\eta)\lambda]/[2-\eta-(1-\eta)\lambda]$ (approximately $0.48$ at $\lambda = 0.5$, $\eta = 0.30$).
When $C_\mu < 0$, the product $\ww' C_\mu \geq 0$, so $g_\mu' < 0$ requires $\kk > |\ww'(\alpha)| \cdot |C_\mu|$: integration frictions must dominate wage sensitivity.

As a numerical illustration, consider $\ww(\abar) = 1 - 0.5\,\abar$, $\cc = 0.30$, $\lambda = 0.5$, $\eta = 0.30$, $\kk = 1$, $\nn = 7$, and $\mu = 0.3$.
The equilibrium rates are $\alpha^{NE} \approx 0.44$, $\alpha^{CO} \approx 0.26$, and $\alpha^{SP}(0.3) \approx 0.04$, confirming $\alpha^{SP} < \alpha^{CO} < \alpha^{NE}$.
The distributional premium is substantial: the $\mu$-planner would reduce automation to near zero, well below the cooperative optimum that already lies far below the Nash rate.
\end{proof}

\begin{proof}[\textbf{Proof of \Cref{prop:recycling}} (Capital income recycling)]
We derive the result for general $\kk \geq 0$; the proposition's two parts follow from the $\kk = 0$ specialization.

Aggregate demand is $D = A + \lambda\ww\LL\nn - \ell\LL\nn\abar + \ehat\Pi$, where total profit is $\Pi = D - \nn\LL(\ww - \ss\abar) - \tfrac{\kk}{2}\LL\sum_j\alpha_j^2$. As in the proof of \Cref{prop:equity}, we maintain nonnegative operating profit at the profiles considered, which keeps $D>0$.
Substituting and solving for~$D$:
\[
  D(1-\ehat) = A + (\lambda-\ehat)\ww\LL\nn - \ell_{\ehat}\,\LL\nn\abar - \tfrac{\ehat\kk}{2}\LL\textstyle\sum_j\alpha_j^2,
\]
with $\ell_{\ehat} = \ell - \ehat\ss$. When $\kk = 0$ this gives~\eqref{eq:D_delta}.
Since $\nn\abar = \alpha_i + \sum_{j\ne i}\alpha_j$, differentiating:
\[
  \frac{\partial D}{\partial \alpha_i} = -\frac{\ell_{\ehat}\,\LL}{1-\ehat} - \frac{\ehat\kk\LL}{1-\ehat}\,\alpha_i.
\]
Revenue $\Rev_i = D/\nn$ gives $\partial\Rev_i/\partial\alpha_i = -\ell_{\ehat}\LL/[\nn(1-\ehat)] - \ehat\kk\LL\alpha_i/[\nn(1-\ehat)]$.
Firm~$i$'s marginal profit, including the direct cost saving~$\ss\LL$ and marginal friction~$\kk\LL\alpha_i$, is
\[
  \frac{\partial\pi_i}{\partial\alpha_i}
  = \LL\!\left(\ss - \frac{\ell_{\ehat}}{\nn(1-\ehat)} - \kk\alpha_i\cdot\frac{\nnhat}{\nn(1-\ehat)}\right),
\]
where $\nnhat \coloneqq \nn(1-\ehat)+\ehat$.
For $\kk>0$, this gives a unique optimal rate independent of rivals; the $\kk=0$ cases follow below.

\emph{Part~(i): $\kk = 0$.}
The marginal profit reduces to $\LL(\ss - \ell_{\ehat}/[\nn(1-\ehat)])$, a constant independent of~$\alpha_i$.
This is positive if and only if $\ss\nn(1-\ehat) > \ell_{\ehat}$, i.e. $\nn > \frac{\ell_{\ehat}}{\ss(1-\ehat)} = \nntil$.
Thus, when $\ell_{\ehat} > 0$ so that $\nntil > 0$, full automation is strictly dominant when $\nn > \nntil$ and no automation is strictly dominant when $\nn < \nntil$, reproducing the structure of \Cref{cor:frictionless}; when $\ell_{\ehat} \leq 0$ the marginal profit $\LL(\ss - \ell_{\ehat}/[\nn(1-\ehat)])$ is positive for every~$\nn$, so full automation is unconditionally dominant, the case characterized in part~(ii).

\emph{Part~(ii).}
$\ell_{\ehat} = \ell - \ehat\ss \leq 0$ if and only if $\ehat \geq \ell/\ss$.

\emph{Extension to $\kk > 0$.}
For an interior equilibrium, the first-order condition gives:
\[
  \alpha^{NE} = \frac{\nn(1-\ehat)}{\kk\nnhat}\!\left(\ss - \frac{\ell_{\ehat}}{\nn(1-\ehat)}\right) = \frac{\ss\nnhat - \ell}{\kk\nnhat} = \frac{\ss - \ell/\nnhat}{\kk}.
\]
This is positive if and only if $\nnhat > \nn^{*}$, reproducing \Cref{prop:alphastar} with~$\nn$ replaced by~$\nnhat$. Since $\alpha^{NE}=0$ once $\nnhat \leq \nn^{*}$, recycling at or above $\ehat=(\nn-\nn^{*})/(\nn-1)$ stops private automation at the corner even when $\ell_{\ehat}>0$.
At a symmetric profile ($\alpha_j = \alpha$ for all~$j$), total profit is $\Pi = \frac{\nn}{1-\ehat}[\Piz + \LL(\ss-\ell)\alpha - \tfrac{\kk}{2}\LL\alpha^2]$.
The $1/(1-\ehat)$ multiplier scales the objective without changing the optimizer, so $\alpha^{CO}=\min\{1,\max\{0,(\ss-\ell)/\kk\}\}$, the same as in \Cref{prop:alphastar}.
\end{proof}

\clearpage
\begingroup
\color{black}
\setcounter{equation}{0}
\setcounter{proposition}{0}
\setcounter{table}{0}
\renewcommand{\theequation}{\thesection.\arabic{equation}}
\renewcommand{\theproposition}{\thesection.\arabic{proposition}}
\renewcommand{\thetable}{\thesection.\arabic{table}}
\renewcommand{\theHequation}{ge.\arabic{equation}}
\renewcommand{\theHproposition}{ge.\arabic{proposition}}
\renewcommand{\theHtable}{ge.\arabic{table}}
\section{General equilibrium and owner welfare}\label{ge:appendix}

This appendix establishes that owners of automating firms can lose consumption
utility in general equilibrium, even when firms anticipate all equilibrium
responses. We first solve household spending and market clearing, then derive
the gap between private and coordinated automation. The final subsection
specifies an example in which every operating owner prefers restraint and a
funded transfer to outside owners makes restraint a Pareto improvement.

\subsection{Economic environment}\label{ge:environment}

Firms choose automation at date 0 under inherited contracts and capacity.
Households choose consumption and saving over two dates; bonds are repaid
at date 1.

\emph{Production.} There are two goods: a mass-consumption composite $M$ and
an outside resource $X$, used for consumption and automation installation.
$N>1$ firms produce varieties of $M$. A symmetric CES aggregator, shared by
all users, has elasticity $\sigma>1$ and is normalized so that equal prices
of one imply a composite price of one and sales $y_i=Y/N$, with gross output $Y$.
Sales contracts fix current prices at one and require firms to fill orders up
to capacity $\bar Y/N$; the contracted wage is $w$.
Firms choose $a_i\in[0,1]$. Price changes, renegotiation, entry, and strategic
withholding are excluded, and capacity cannot be converted to production of $X$.

Firm $i$ uses $(1-a_i)y_i$ units of labor and $ca_i y_i$ units of $M$ as
recurring automation services, supplied at zero profit. Installation uses
$k_i=\kappa a_i^2/(2N)$ units of $X$, with $\kappa>0$. Let $a$ be mean
automation, $K$ total installation, $z_i$ the operating margin, and $p$ the
relative price of $X$:
\begin{equation}\label{ge:technology}
 a=\frac1N\sum_i a_i,\qquad K=\sum_i k_i,\qquad
 z_i=1-w+(w-c)a_i,\qquad \pi_i=z_i\frac YN-pk_i.
\end{equation}
We take $0<c<w<1$, so automation reduces recurring unit costs.

\emph{Households.} Workers, operating owners, and outside owners each have
population one. Operating-owner group $i$ has mass $1/N$, owns firm $i$, and
controls automation to maximize its members' utility. Ownership is fixed and
separate across firms. Workers can supply $\bar Y$ units of labor but are
rationed to employment $H=(1-a)Y$.

The current endowment of $X$ is normalized to one: workers own $e\in(0,1)$ and outside
owners own $1-e$. Positive terminal $M$ endowments are $\omega_W$, $\omega_F$, and
$\omega_X$ for the three groups; operating owners share $\omega_F$ equally.
Outside owners also hold one unit of terminal $X$. There is no current $M$
endowment. Additional terminal productivity from automation is omitted.

Per-capita utility is $\log C_{h0}+\beta\log C_{h1}$, less $\chi H$ for workers,
where $\beta>0$, $\chi>0$, and
\begin{equation}\label{ge:preferences}
 C_{ht}=\left(\frac{C_{M,ht}}{\theta_h}\right)^{\theta_h}
 \left(\frac{C_{X,ht}}{1-\theta_h}\right)^{1-\theta_h},
 \qquad \theta_h\in(0,1).
\end{equation}
Workers have expenditure share $\theta_W$; both owner groups have share
$\theta_O$.

\emph{Saving and policy.} A bond pays one unit of terminal $M$ at gross return
$R$. With group expenditure $m_{ht}$ and bond payoff $b_h$, budgets and bond
clearing are
\begin{equation}\label{ge:budgets}
 m_{h0}+b_h/R=I_h,\qquad m_{h1}=E_h+b_h,\qquad \sum_h b_h=0.
\end{equation}
Current incomes are $I_W=wH+ep$, $I_i=\pi_i$, and $I_X=(1-e)p$;
terminal incomes are $E_W=\omega_W$, $E_i=\omega_F/N$, and
$E_X=\omega_X+p_1$, where $p_1$ is the terminal relative price of $X$.
Workers cannot borrow ($b_W\geq0$); owners can borrow against positive lifetime
wealth. Policy seeks capacity output subject to $R\geq1$, so
$(R-1)(\bar Y-Y)=0$. This floor is specified in units of $M$; a nominal
interpretation also requires $P_{M0}=P_{M1}=1$, since
$R=(1+i)P_{M0}/P_{M1}$.

Evidence on liquidity-dependent consumption, infrequent price changes,
downward wage rigidity, and costly capital redeployment motivates these
short-run restrictions \citep{FagerengHolmNatvik2021,NakamuraSteinsson2008,
FallickLettauWascher2020,rameyshapiro2001displaced}. Constrained monetary accommodation
is a standard setting for demand externalities \citep{farhi2016macroprudential}.
Earlier contract formation and optimal policy are not modeled.

\subsection{The allocation for a given automation profile}\label{ge:equilibrium}

Fix firms' choices, which affect aggregates through $(a,K)$. We solve the
regime with $b_W=0$ and $R=1$, then verify its validity for every strategy
profile in the example. Workers spend current income; owner spending follows
from prices and bond clearing.

With $b_W=0$, terminal goods clearing and the owner Euler equations give
\begin{equation}\label{ge:terminal}
 p_1=\frac{(1-\theta_W)\omega_W+(1-\theta_O)(\omega_F+\omega_X)}{\theta_O},
 \qquad B=\frac{\omega_F+\omega_X+p_1}{\beta R},
\end{equation}
where $B$ is aggregate current owner expenditure. Bond clearing equates owners' aggregate
terminal spending to their endowment income. Thus $B$ is independent of
automation at $R=1$.

Current goods clearing requires
\begin{equation}\label{ge:clearing}
 Y=\theta_W I_W+\theta_O B+caY,\qquad
 p(1-K)=(1-\theta_W)I_W+(1-\theta_O)B.
\end{equation}
These equations allocate $M$ to consumption and recurring inputs, and $X$ to
consumption and installation. Define the average
margin $z=1-w(1-a)-ca$, $d=1-\theta_Ww(1-a)-ca$, and
$\Delta=d(1-K)-(1-\theta_W)e(1-ca)$. Substituting worker income yields
\begin{equation}\label{ge:solution}
 \begin{aligned}
 Y&=\frac{B[\theta_O(1-K)+e(\theta_W-\theta_O)]}{\Delta},\qquad
 p=\frac{B(d-\theta_Oz)}{\Delta},\\
 \Pi&\equiv\sum_i\pi_i=B-(1-e)p.
 \end{aligned}
\end{equation}
Each owner group spends $m_{h0}=(I_h+E_h/R)/(1+\beta)$ and
$m_{h1}=\beta Rm_{h0}$; workers spend $I_W$ and $\omega_W$.
Consumption demands are $\theta_hm_{ht}$ units of $M$ and
$(1-\theta_h)m_{ht}/p_t$ units of $X$.

The identity $\Pi+(1-e)p=B$ accounts for the income offset: lower operating
profits mean higher current income of outside owners. We next compare
operating owners' consumption utility.

\subsection{Private choices and coordinated restraint}\label{ge:owners}

Firm $i$ chooses automation to maximize its owners' utility, anticipating the
allocation derived above. Write $J_i=N\pi_i$ for its owners' per-capita current profit and $x$ for its own
choice. At $R=1$, lifetime utility is
\begin{equation}\label{ge:ownerobjective}
 \begin{aligned}
 u_i={}&(1+\beta)\log\frac{J_i+\omega_F}{1+\beta}+\beta\log\beta
       -(1-\theta_O)(\log p+\beta\log p_1),\\
 J(x,a,K)={}&[1-w+(w-c)x]Y(a,K)-\tfrac\kappa2 p(a,K)x^2.
 \end{aligned}
\end{equation}
The objective includes both lifetime wealth and consumption prices. Holding
rivals' choices fixed, the unilateral derivative is
\begin{equation}\label{ge:gradient}
 \mathcal D=\partial_x+\frac1N\partial_a+\frac{\kappa x}{N}\partial_K,
 \qquad g=\frac{(1+\beta)\mathcal DJ}{J+\omega_F}
          -(1-\theta_O)\frac{\mathcal Dp}{p}.
\end{equation}
The aggregate effects $\partial a/\partial x=1/N$ and
$\partial K/\partial x=\kappa x/N$ follow from \eqref{ge:technology}.

Coordination maximizes mean operating-owner utility. Suppose the maintained
regime has $\Delta>0$ and positive lifetime wealth throughout the strategy
set, as verified below. Then $p_K>0$, since $\Delta$ falls with $K$ while
$d-\theta_Oz=(1-\theta_W)\theta_Ow(1-a)+(1-\theta_O)d>0$: dispersion at a
fixed mean raises $K$ and $p$. Jensen's inequality bounds mean utility by the value from equalizing
profits at the same $\Pi$ and $p$. Since $\Pi=B-(1-e)p$, this bound decreases
in $p$. Symmetry attains the bound at the lowest $K$ for each $a$.
The coordinated choice therefore minimizes $p$ along $K=\kappa a^2/2$ and
also maximizes joint profits.

We can now isolate the private incentive at an interior coordinated optimum.
Let primes denote derivatives along a common symmetric change and set
$D=(w-c)Y-\kappa ap$. Then
\begin{equation}\label{ge:spillover}
 \mathcal DJ=D+\frac{zY'-Kp'}N,\qquad
 \Pi'=D+zY'-Kp'.
\end{equation}
Here $D$ is the direct cost saving net of additional installation. Sales and
installation prices respond by $1/N$ of their response to a common change.

At the coordinated optimum, $\Pi'=-(1-e)p'=0$. The outside price is thus
locally unchanged, while \eqref{ge:clearing} implies
$Y'=-(\theta_Ww-c)Y/d$. When $\theta_Ww>c$, substitution gives
\begin{equation}\label{ge:derived}
 g(a^{CO})=\frac{(1+\beta)(1-1/N)z(\theta_Ww-c)Y}{(\Pi+\omega_F)d}>0,
 \qquad \frac{\partial\pi_j}{\partial a_i}=\frac{zY'}{N^2}<0
 \quad(j\ne i).
\end{equation}
Each firm wants to raise automation at the coordinated allocation, reducing
rivals' profits through sales. The demand spillover follows from market
clearing and separate ownership. We now establish that this local incentive
produces an owner-welfare loss at a global Nash equilibrium.

\subsection{A solved example and the welfare comparison}\label{ge:example}

Choose $N=7$, $w=0.8$, $c=0.1$, $\kappa=0.5$, $\bar Y=2$, $\beta=1$,
$\chi=0.1$, $\theta_W=0.95$, $\theta_O=0.20$, and endowments
\begin{equation}\label{ge:parameters}
 e=0.01,\qquad \omega_W=1,\qquad \omega_F=0.10,\qquad \omega_X=0.05.
\end{equation}
Thus $p_1=0.85$ and, at the floor, $B=1$.
These values specify an existence example; they are not a calibration.

\begin{proposition}\label{ge:proposition}
For these parameters, the conditional GE allocation is unique for every
automation profile. The unique Nash equilibrium in owner utility is symmetric,
with $a^{NE}\simeq0.502780412$. The unique maximizer of mean operating-owner
utility is symmetric, with $a^{CO}\simeq0.154777371$, and also maximizes joint
current profits. Every operating owner strictly prefers $a^{CO}$ to Nash
without transfers. A balanced terminal transfer makes $a^{CO}$ strictly
better for all household groups.
\end{proposition}

\begin{proof}
\emph{Equilibrium for every strategy profile.}
For every profile, $a^2/4\leq K\leq a/4$, so
$\Delta\geq0.24(0.74)$ and $Y<1.169<2$ at $R=1$.
Also $p<1.07$: substituting $K\leq a/4$ bounds its denominator below by
$0.2395+0.60005a-0.165a^2$; $1.07$ times this polynomial minus
$0.2+0.52a$ is concave and positive at both endpoints.
Hence $I_W<0.946<\omega_W$, verifying the borrowing constraint and willingness
to work, $w/I_W-\chi>0$. Owner wealth is positive: $J$ is concave in $x$, and $Y/p\geq7/32$ gives
$J+0.1\geq0.1-1.07(17/320)>0$ at both endpoints $x\in\{0,1\}$.

An unconstrained-worker branch at the floor would instead give
$Y_f=1.15/(1-0.1a)$ and $p_f=(2.725-1.875T)/(1-K+0.01875)$, where
$T=0.8(1-a)Y_f\leq0.92$. Its bond holding $(T+0.01p_f-1)/2$ is negative
because $(1-K)(1-T)\geq0.06>0.0085$, violating the borrowing limit.
Both candidate branches have output below capacity at the floor and
proportional to $1/R$, excluding a capacity-clearing solution with $R>1$.

\emph{Coordinated optimum and Nash equilibrium.}
The coordinated optimum minimizes $p$ along $K=a^2/4$. Its rate is the unique
root of the strictly increasing cubic
\begin{equation}\label{ge:co}
 5720a^3+4340a^2+800a-249=0.
\end{equation}
For the private game, two signs suffice: $\mathcal Dg<0$ along every unilateral
deviation and $\partial_xg<0$ holding aggregates fixed. The first gives strict
concavity. The second excludes asymmetric equilibria: if $a_i>a_j$, then
$g_i<g_j$, whereas optimality requires $g_i\geq0$ and $g_j\leq0$.

We verify both signs in exact arithmetic. For fixed-aggregate ordering,
substitute $K=[a^2+va(1-a)]/4$, with $(a,v,x)\in[0,1]^3$.
For unilateral concavity, let $m$ denote the rivals' mean and substitute
\begin{equation}\label{ge:moments}
 a=(x+6m)/7,\qquad
 K=\{x^2+6[m^2+vm(1-m)]\}/28,
 \qquad (m,v,x)\in[0,1]^3.
\end{equation}
These boxes contain every feasible profile. After certifying and orienting
denominators to be positive, the numerators of $-\partial_xg$ and
$-\mathcal Dg$ have strictly positive tensor Bernstein coefficients:
160 of degree $(7,3,4)$ and 845 of degree $(12,4,12)$.
Nonnegative basis functions that sum to one establish both signs; the
replication code verifies every coefficient in
exact arithmetic. It will be released with the paper's public replication package and is available from the authors in the meantime.
At symmetry, the degree-six numerator of $g$ has Bernstein coefficient signs
$(+,+,+,+,-,-,-)$ and a positive denominator. Descartes' rule after
$a=t/(1+t)$ gives exactly one interior root, since endpoint signs differ.
Strict unilateral concavity makes this root a global best response, establishing
existence and uniqueness of $a^{NE}$.

\emph{Welfare and compensation.}
Since $a^{NE}\ne a^{CO}$, $p^{NE}>p^{CO}$ and $\Pi^{NE}<\Pi^{CO}$;
\eqref{ge:ownerobjective} gives a strict loss for every operating owner.
Transfer $L=\Pi^{CO}-\Pi^{NE}$ units of terminal $M$ from operating owners,
equally across firms, to outside owners. Identical preferences preserve
aggregate spending and prices, while each owner group's nominal lifetime
wealth becomes its Nash value. Both owner groups gain
$0.8\log(p^{NE}/p^{CO})>0$.
Exact root brackets $a^{CO}\in(0.154,0.155)$ and
$a^{NE}\in(0.502,0.503)$ imply $p^{CO}>0.8$ and $p^{NE}<0.9$, so
$0<L<0.099<\omega_F$, verifying funding.
The same brackets give $I_W^{CO}>0.4$ and $I_W^{NE}<0.16$.
Since $p^{CO}<p^{NE}$ and additional hours are below two, worker utility rises
by more than $\log(2.5)-0.2>0$.
\end{proof}

Table~\ref{ge:table} reports the uncompensated comparison. Operating owners
prefer restraint, while outside owners prefer Nash.

\begin{table}[htbp]
\centering
\caption{Equilibrium and coordinated restraint}\label{ge:table}
\begin{tabular}{lrr}
\toprule
 & Nash & Coordinated\\
\midrule
Common automation & $0.502780$ & $0.154777$\\
Gross output $Y$ & $0.364074$ & $0.607465$\\
Operating profits $\Pi$ & $0.146464$ & $0.182362$\\
Outside price $p$ & $0.862158$ & $0.825897$\\
Operating-owner utility & $-3.938707$ & $-3.632382$\\
Worker utility & $-1.876997$ & $-0.903506$\\
Outside-owner utility & $-0.014357$ & $-0.021351$\\
\bottomrule
\end{tabular}
\par\smallskip
\parbox{\textwidth}{Utilities are per capita and include both dates.
The coordinated column has no transfers.}
\end{table}

The terminal transfer is $L\simeq0.035898310$. After compensation, each owner
group gains $0.034375$ utility; workers gain $0.973490$, including the cost of
additional work. This funded Pareto improvement establishes a constrained
deadweight loss under the stated contracts and policy rule.

\FloatBarrier
\paragraph{Interpretation and scope.}
The example describes an adjustment interval in which automation can change
before contracts and productive capacity adjust. Owners' consumption and
saving, payments for recurring and installation inputs, and the outside-good
price response are included in the welfare comparison. The loss to operating
owners therefore survives these GE adjustments. The economy operates with
substantial slack: output is 18\% of capacity at Nash, 30\% under restraint,
and 43\% with no automation. Capacity enters only through the rate floor, so
any $\bar Y$ above the output bounds in the proof leaves
Table~\ref{ge:table} unchanged.

Beyond the inherited contracts, four assumptions carry the result. Workers
cannot borrow, so lost wages are not replaced by spending out of future
endowments; in the unconstrained-worker branch of the proof, output rises
with automation. The return cannot fall below one, so aggregate owner
spending $B$ does not rise to absorb the shortfall; with full accommodation,
policy would hold output at capacity. Installation spending leaves the mass
sector; a numerical variant in which installation is purchased from the
sector reverses the ordering of Nash and coordinated automation. Finally, $\theta_Ww>c$, which makes
output fall with common automation at the coordinated optimum: the spending
a displaced task withdraws from the sector exceeds the spending on automation
services, produced within the sector, that replaces it. This condition plays
the role of the baseline's $\ell>0$ with $\lambda=\theta_W$, $\eta=0$ (lost
wages are not replaced here), and the recurring cost netted out. The owners'
share $\theta_O$ does not enter this local sign condition.

Reemployment, production retooling, additional investment, future productivity
gains, and more effective policy can mitigate the loss. The main paper's
partial-equilibrium analysis shows that incomplete income replacement and wage
adjustment under its wage-taking benchmark can leave a positive automation gap
when automation remains active and the cooperative rate is below one
(\Cref{sec:eta,sec:endo_wages}). Full income replacement removes the baseline
demand channel. The presence of an offset therefore does not itself establish
full correction. In a richer GE model, the owner-welfare comparison would
depend on how strongly and quickly these responses restore spending or raise
owners' lifetime consumption.

\clearpage

\ifRFCombined\else
\bibliographystyle{plainnat}
\bibliography{refs_final}

\begin{thebibliography}{60}
\providecommand{\natexlab}[1]{#1}
\providecommand{\url}[1]{\texttt{#1}}
\expandafter\ifx\csname urlstyle\endcsname\relax
  \providecommand{\doi}[1]{doi: #1}\else
  \providecommand{\doi}{doi: \begingroup \urlstyle{rm}\Url}\fi

\bibitem[Acemoglu(2025)]{acemoglu2025simple}
Daron Acemoglu.
\newblock The simple macroeconomics of {AI}.
\newblock \emph{Economic Policy}, 40\penalty0 (121):\penalty0 13--58, 2025.

\bibitem[Acemoglu and Restrepo(2018)]{acemoglu2018race}
Daron Acemoglu and Pascual Restrepo.
\newblock The race between man and machine: Implications of technology for
  growth, factor shares, and employment.
\newblock \emph{American Economic Review}, 108\penalty0 (6):\penalty0
  1488--1542, 2018.

\bibitem[Acemoglu and Restrepo(2019)]{acemoglu2019automation}
Daron Acemoglu and Pascual Restrepo.
\newblock Automation and new tasks: How technology displaces and reinstates
  labor.
\newblock \emph{Journal of Economic Perspectives}, 33\penalty0 (2):\penalty0
  3--30, 2019.

\bibitem[Acemoglu and Restrepo(2020)]{acemoglu2020wrong}
Daron Acemoglu and Pascual Restrepo.
\newblock The wrong kind of {AI}? {A}rtificial intelligence and the future of
  labour demand.
\newblock \emph{Cambridge Journal of Regions, Economy and Society}, 13\penalty0
  (1):\penalty0 25--35, 2020.

\bibitem[Acemoglu and Restrepo(2026)]{acemoglu2024rent}
Daron Acemoglu and Pascual Restrepo.
\newblock Automation and rent dissipation: Implications for wages, inequality,
  and productivity.
\newblock \emph{Quarterly Journal of Economics}, 141\penalty0 (2):\penalty0
  1521--1579, 2026.

\bibitem[Autor et~al.(2003)Autor, Levy, and Murnane]{autor2003skill}
David~H. Autor, Frank Levy, and Richard~J. Murnane.
\newblock The skill content of recent technological change: An empirical
  exploration.
\newblock \emph{Quarterly Journal of Economics}, 118\penalty0 (4):\penalty0
  1279--1333, 2003.

\bibitem[Autor et~al.(2024)Autor, Chin, Salomons, and Seegmiller]{autor2024new}
David~H. Autor, Caroline Chin, Anna Salomons, and Bryan Seegmiller.
\newblock New frontiers: The origins and content of new work, 1940--2018.
\newblock \emph{Quarterly Journal of Economics}, 139\penalty0 (3):\penalty0
  1399--1465, 2024.

\bibitem[Babina et~al.(2024)Babina, Fedyk, He, and
  Hodson]{babina2024artificial}
Tania Babina, Anastassia Fedyk, Alex He, and James Hodson.
\newblock Artificial intelligence, firm growth, and product innovation.
\newblock \emph{Journal of Financial Economics}, 151:\penalty0 103745, 2024.
\newblock \doi{10.1016/j.jfineco.2023.103745}.

\bibitem[Banchio and Mantegazza(2022)]{banchio2022artificial}
Martino Banchio and Giacomo Mantegazza.
\newblock Artificial intelligence and spontaneous collusion.
\newblock Technical Report 2202.05946, arXiv, 2022.
\newblock URL \url{https://arxiv.org/abs/2202.05946}.

\bibitem[Bastani and Cachon(2025)]{bastani2025human}
Hamsa Bastani and G\'{e}rard~P. Cachon.
\newblock The human-{AI} contracting paradox.
\newblock Working paper, The Wharton School, University of Pennsylvania,
  December 2025.
\newblock SSRN 5962739.

\bibitem[Bena and Simintzi(2025)]{bena2025machines}
Jan Bena and Elena Simintzi.
\newblock Machines could not compete with {C}hinese labor: Evidence from {US}
  firms' innovation.
\newblock \emph{Review of Finance}, 29\penalty0 (6):\penalty0 1619--1661, 2025.
\newblock \doi{10.1093/rof/rfaf040}.

\bibitem[Benzell et~al.(2015)Benzell, Kotlikoff, LaGarda, and
  Sachs]{benzell2015robots}
Seth~G. Benzell, Laurence~J. Kotlikoff, Guillermo LaGarda, and Jeffrey~D.
  Sachs.
\newblock Robots are {U}s: Some economics of human replacement.
\newblock Working Paper 20941, National Bureau of Economic Research, February
  2015.

\bibitem[Beraja and Zorzi(2025)]{beraja2025inefficient}
Martin Beraja and Nathan Zorzi.
\newblock Inefficient automation.
\newblock \emph{Review of Economic Studies}, 92\penalty0 (1):\penalty0 69--96,
  2025.

\bibitem[Bhaimiya(2025)]{cnbc2025layoffs}
Sawdah Bhaimiya.
\newblock {AI} was behind over 50,000 layoffs in 2025 --- here are the top
  firms to cite it for job cuts.
\newblock CNBC, December~21 2025.
\newblock URL
  \url{https://www.cnbc.com/2025/12/21/ai-job-cuts-amazon-microsoft-and-more-cite-ai-for-2025-layoffs.html}.

\bibitem[Bhaimiya(2026)]{cnbc2026amodei}
Sawdah Bhaimiya.
\newblock Anthropic {CEO} {D}ario {A}modei warns {AI} may cause `unusually
  painful' disruption to jobs.
\newblock CNBC, January~27 2026.
\newblock URL
  \url{https://www.cnbc.com/2026/01/27/dario-amodei-warns-ai-cause-unusually-painful-disruption-jobs.html}.

\bibitem[Bondi and Johnson(2026)]{bondi2026skill}
Tommaso Bondi and Gentry Johnson.
\newblock Skill atrophy and {AI} productivity measurement.
\newblock Technical report, SSRN, 2026.
\newblock SSRN 6169671.

\bibitem[Boppart(2014)]{boppart2014structural}
Timo Boppart.
\newblock Structural change and the {K}aldor facts in a growth model with
  relative price effects and non-{G}orman preferences.
\newblock \emph{Econometrica}, 82\penalty0 (6):\penalty0 2167--2196, 2014.

\bibitem[Brynjolfsson et~al.(2025)Brynjolfsson, Li, and
  Raymond]{brynjolfsson2025generative}
Erik Brynjolfsson, Danielle Li, and Lindsey Raymond.
\newblock Generative {AI} at work.
\newblock \emph{Quarterly Journal of Economics}, 140\penalty0 (2):\penalty0
  889--942, 2025.

\bibitem[Brynjolfsson et~al.(2026)Brynjolfsson, Chandar, and
  Chen]{brynjolfsson2025canaries}
Erik Brynjolfsson, Bharat Chandar, and Ruyu Chen.
\newblock Canaries in the coal mine? {S}ix facts about the recent employment
  effects of artificial intelligence.
\newblock Working paper, Stanford Digital Economy Lab, August 2026.
\newblock URL
  \url{https://digitaleconomy.stanford.edu/app/uploads/2026/08/Canaries_August2026.pdf}.

\bibitem[Budman(2025)]{cnbc2025salesforce}
Scott Budman.
\newblock Salesforce {CEO} confirms 4,000 layoffs ``because {I} need less
  heads'' with {AI}.
\newblock CNBC, September~2 2025.
\newblock URL
  \url{https://www.cnbc.com/2025/09/02/salesforce-ceo-confirms-4000-layoffs-because-i-need-less-heads-with-ai.html}.

\bibitem[Caballero(2026)]{caballero2026speculative}
Ricardo~J. Caballero.
\newblock Speculative growth and the {AI} ``{B}ubble''.
\newblock Working Paper 34722, National Bureau of Economic Research, 2026.
\newblock URL \url{https://www.nber.org/papers/w34722}.

\bibitem[Caosun and Aral(2026)]{caosun2026augmentation}
Michael Caosun and Sinan Aral.
\newblock The augmentation trap: {AI} productivity and the cost of cognitive
  offloading.
\newblock Technical Report 2604.03501, arXiv, 2026.
\newblock URL \url{https://arxiv.org/abs/2604.03501}.

\bibitem[{Challenger, Gray \& Christmas}(2026)]{challenger2026september}
{Challenger, Gray \& Christmas}.
\newblock Challenger report: September 2026.
\newblock
  \url{https://www.challengergray.com/wp-content/uploads/2026/10/Challenger-Report-September-2026.pdf},
  October 2026.
\newblock Released October~1, 2026. Accessed October~17, 2026.

\bibitem[Chod et~al.(2019)Chod, Lyandres, and Yang]{chod2019trade}
Jiri Chod, Evgeny Lyandres, and S.~Alex Yang.
\newblock Trade credit and supplier competition.
\newblock \emph{Journal of Financial Economics}, 131\penalty0 (2):\penalty0
  484--505, 2019.
\newblock \doi{10.1016/j.jfineco.2018.08.008}.

\bibitem[{Citrini} and Shah(2026)]{shah2026intelligence}
{Citrini} and Alap Shah.
\newblock The 2028 {G}lobal {I}ntelligence {C}risis.
\newblock Citrini Research, February 2026.
\newblock URL \url{https://www.citriniresearch.com/p/2028gic}.

\bibitem[Coase(1960)]{coase1960problem}
R.~H. Coase.
\newblock The problem of social cost.
\newblock \emph{Journal of Law and Economics}, 3:\penalty0 1--44, 1960.

\bibitem[Comin et~al.(2021)Comin, Lashkari, and
  Mestieri]{cominlashkari2021structural}
Diego Comin, Danial Lashkari, and Mart\'{\i} Mestieri.
\newblock Structural change with long-run income and price effects.
\newblock \emph{Econometrica}, 89\penalty0 (1):\penalty0 311--374, 2021.

\bibitem[Cooper and John(1988)]{cooper1988coordinating}
Russell Cooper and Andrew John.
\newblock Coordinating coordination failures in {K}eynesian models.
\newblock \emph{Quarterly Journal of Economics}, 103\penalty0 (3):\penalty0
  441--463, 1988.

\bibitem[Cooper and Haltiwanger(2006)]{cooperhaltiwanger2006adjustment}
Russell~W. Cooper and John~C. Haltiwanger.
\newblock On the nature of capital adjustment costs.
\newblock \emph{Review of Economic Studies}, 73\penalty0 (3):\penalty0
  611--633, 2006.

\bibitem[Costinot and Werning(2023)]{costinot2023robots}
Arnaud Costinot and Iv{\'a}n Werning.
\newblock Robots, trade, and {L}uddism: A sufficient statistic approach to
  optimal technology regulation.
\newblock \emph{Review of Economic Studies}, 90\penalty0 (5):\penalty0
  2261--2291, 2023.

\bibitem[Eloundou et~al.(2024)Eloundou, Manning, Mishkin, and
  Rock]{eloundou2024gpts}
Tyna Eloundou, Sam Manning, Pamela Mishkin, and Daniel Rock.
\newblock {GPTs} are {GPTs}: Labor market impact potential of {LLMs}.
\newblock \emph{Science}, 384\penalty0 (6702):\penalty0 1306--1308, 2024.
\newblock \doi{10.1126/science.adj0998}.

\bibitem[Fagereng et~al.(2021)Fagereng, Holm, and
  Natvik]{FagerengHolmNatvik2021}
Andreas Fagereng, Martin~B. Holm, and Gisle~J. Natvik.
\newblock {MPC} heterogeneity and household balance sheets.
\newblock \emph{American Economic Journal: Macroeconomics}, 13\penalty0
  (4):\penalty0 1--54, 2021.
\newblock \doi{10.1257/mac.20190211}.

\bibitem[Fallick et~al.(2020)Fallick, Villar, and
  Wascher]{FallickLettauWascher2020}
Bruce Fallick, Daniel Villar, and William~L. Wascher.
\newblock Downward nominal wage rigidity in the {United States} during and
  after the {Great Recession}.
\newblock Finance and Economics Discussion Series 2016-001r1, Board of
  Governors of the Federal Reserve System, 2020.

\bibitem[Farhi and Werning(2016)]{farhi2016macroprudential}
Emmanuel Farhi and Iv{\'a}n Werning.
\newblock A theory of macroprudential policies in the presence of nominal
  rigidities.
\newblock \emph{Econometrica}, 84\penalty0 (5):\penalty0 1645--1704, 2016.
\newblock \doi{10.3982/ECTA11883}.

\bibitem[Guerreiro et~al.(2022)Guerreiro, Rebelo, and
  Teles]{guerreiro2022should}
Jo{\~a}o Guerreiro, Sergio Rebelo, and Pedro Teles.
\newblock Should robots be taxed?
\newblock \emph{Review of Economic Studies}, 89\penalty0 (1):\penalty0
  279--311, 2022.

\bibitem[Jacobson et~al.(1993)Jacobson, LaLonde, and
  Sullivan]{jacobson1993earnings}
Louis~S. Jacobson, Robert~J. LaLonde, and Daniel~G. Sullivan.
\newblock Earnings losses of displaced workers.
\newblock \emph{American Economic Review}, 83\penalty0 (4):\penalty0 685--709,
  1993.

\bibitem[Kaldor(1956)]{kaldor1956alternative}
Nicholas Kaldor.
\newblock Alternative theories of distribution.
\newblock \emph{Review of Economic Studies}, 23\penalty0 (2):\penalty0 83--100,
  1956.

\bibitem[Keppo et~al.(2026{\natexlab{a}})Keppo, Li, Tsoukalas, and
  Yuan]{keppo2026fragility}
Jussi Keppo, Yuze Li, Gerry Tsoukalas, and Nuo Yuan.
\newblock On the fragility of {AI} agent collusion.
\newblock Technical Report 2603.20281, arXiv, 2026{\natexlab{a}}.
\newblock URL \url{https://arxiv.org/abs/2603.20281}.

\bibitem[Keppo et~al.(2026{\natexlab{b}})Keppo, Li, Tsoukalas, and
  Yuan]{keppo2026trusting}
Jussi Keppo, Yuze Li, Gerry Tsoukalas, and Nuo Yuan.
\newblock Trusting {AI} in competitive markets.
\newblock Technical Report 2608.26539, arXiv, 2026{\natexlab{b}}.
\newblock URL \url{https://arxiv.org/abs/2608.26539}.

\bibitem[Keynes(1931)]{keynes1930economic}
John~Maynard Keynes.
\newblock Economic possibilities for our grandchildren.
\newblock In \emph{Essays in Persuasion}, pages 358--373. Macmillan, London,
  1931.
\newblock First published 1930.

\bibitem[Korinek and Stiglitz(2019)]{korinek2019artificial}
Anton Korinek and Joseph~E. Stiglitz.
\newblock Artificial intelligence and its implications for income distribution
  and unemployment.
\newblock In Ajay Agrawal, Joshua Gans, and Avi Goldfarb, editors, \emph{The
  Economics of Artificial Intelligence: An Agenda}, pages 349--390. University
  of Chicago Press, Chicago, 2019.

\bibitem[Li et~al.(2026)Li, Huang, and Shi]{li2025forced}
Boshuo Li, Ni~Huang, and Wei Shi.
\newblock Forced to change? {M}edia exposure of labor issues and firm
  artificial intelligence investment.
\newblock \emph{Information Systems Research}, 37\penalty0 (1):\penalty0
  156--175, 2026.
\newblock \doi{10.1287/isre.2022.0402}.

\bibitem[Lucas(1967)]{lucas1967adjustment}
Robert~E. Lucas.
\newblock Adjustment costs and the theory of supply.
\newblock \emph{Journal of Political Economy}, 75\penalty0 (4):\penalty0
  321--334, 1967.

\bibitem[Mankiw and Whinston(1986)]{mankiw1986free}
N.~Gregory Mankiw and Michael~D. Whinston.
\newblock Free entry and social inefficiency.
\newblock \emph{RAND Journal of Economics}, 17\penalty0 (1):\penalty0 48--58,
  1986.

\bibitem[Matsuyama(2002)]{matsuyama2002rise}
Kiminori Matsuyama.
\newblock The rise of mass consumption societies.
\newblock \emph{Journal of Political Economy}, 110\penalty0 (5):\penalty0
  1035--1070, 2002.

\bibitem[Murphy et~al.(1989)Murphy, Shleifer, and
  Vishny]{murphy1989industrialization}
Kevin~M. Murphy, Andrei Shleifer, and Robert~W. Vishny.
\newblock Industrialization and the big push.
\newblock \emph{Journal of Political Economy}, 97\penalty0 (5):\penalty0
  1003--1026, 1989.

\bibitem[Nakamura and Steinsson(2008)]{NakamuraSteinsson2008}
Emi Nakamura and J{\'o}n Steinsson.
\newblock Five facts about prices: A reevaluation of menu cost models.
\newblock \emph{Quarterly Journal of Economics}, 123\penalty0 (4):\penalty0
  1415--1464, 2008.
\newblock \doi{10.1162/qjec.2008.123.4.1415}.

\bibitem[Palmer(2026)]{cnbc2026block}
Annie Palmer.
\newblock Block laying off about 4,000 employees, nearly half of its workforce.
\newblock CNBC, February~26 2026.
\newblock URL
  \url{https://www.cnbc.com/2026/02/26/block-laying-off-about-4000-employees-nearly-half-of-its-workforce.html}.

\bibitem[Pigou(1920)]{pigou1920}
Arthur~Cecil Pigou.
\newblock \emph{The Economics of Welfare}.
\newblock Macmillan, London, 1920.

\bibitem[Ramey and Shapiro(2001)]{rameyshapiro2001displaced}
Valerie~A. Ramey and Matthew~D. Shapiro.
\newblock Displaced capital: A study of aerospace plant closings.
\newblock \emph{Journal of Political Economy}, 109\penalty0 (5):\penalty0
  958--992, 2001.

\bibitem[Ricardo(1821)]{ricardo1821principles}
David Ricardo.
\newblock \emph{On the Principles of Political Economy and Taxation}.
\newblock John Murray, London, 3rd edition, 1821.

\bibitem[Rosenstein-Rodan(1943)]{rosenstein1943problems}
Paul~N. Rosenstein-Rodan.
\newblock Problems of industrialisation of {E}astern and {S}outh-{E}astern
  {E}urope.
\newblock \emph{Economic Journal}, 53\penalty0 (210/211):\penalty0 202--211,
  1943.

\bibitem[Roth and Xing(1994)]{roth1994jumping}
Alvin~E. Roth and Xiaolin Xing.
\newblock Jumping the gun: {I}mperfections and institutions related to the
  timing of market transactions.
\newblock \emph{American Economic Review}, 84\penalty0 (4):\penalty0 992--1044,
  1994.

\bibitem[Shapiro and Stiglitz(1984)]{shapiro1984equilibrium}
Carl Shapiro and Joseph~E. Stiglitz.
\newblock Equilibrium unemployment as a worker discipline device.
\newblock \emph{American Economic Review}, 74\penalty0 (3):\penalty0 433--444,
  1984.

\bibitem[Son(2025)]{cnbc2025devin}
Hugh Son.
\newblock Goldman {S}achs is piloting its first autonomous coder in major {AI}
  milestone for {W}all {S}treet.
\newblock CNBC, July~11 2025.
\newblock URL
  \url{https://www.cnbc.com/2025/07/11/goldman-sachs-autonomous-coder-pilot-marks-major-ai-milestone.html}.

\bibitem[Summers(2014)]{summers2014secular}
Lawrence~H. Summers.
\newblock {U.S.} economic prospects: Secular stagnation, hysteresis, and the
  zero lower bound.
\newblock \emph{Business Economics}, 49\penalty0 (2):\penalty0 65--73, 2014.

\bibitem[Summers(2015)]{summers2015demand}
Lawrence~H. Summers.
\newblock Demand side secular stagnation.
\newblock \emph{American Economic Review}, 105\penalty0 (5):\penalty0 60--65,
  2015.

\bibitem[Weitzman(1985)]{weitzman1985sharing}
Martin~L. Weitzman.
\newblock The simple macroeconomics of profit sharing.
\newblock \emph{American Economic Review}, 75\penalty0 (5):\penalty0 937--953,
  1985.

\bibitem[Yang and Zhang(2024)]{yang2024generative}
S.~Alex Yang and Angela~Huyue Zhang.
\newblock Generative {AI} and copyright: A dynamic perspective.
\newblock Technical report, SSRN, 2024.
\newblock SSRN 4716233.

\bibitem[Zeira(1998)]{zeira1998workers}
Joseph Zeira.
\newblock Workers, machines, and economic growth.
\newblock \emph{Quarterly Journal of Economics}, 113\penalty0 (4):\penalty0
  1091--1117, 1998.

\end{thebibliography}
\fi
\clearpage
\endgroup
\fi
\end{document}